\documentclass[11pt, leqno]{amsart}
\usepackage[a4paper,margin=1in,footskip=0.25in]{geometry}

\usepackage{graphicx}
\usepackage{amssymb,amsmath,amsthm,amscd}
\usepackage{mathrsfs}
\usepackage{lscape}
\usepackage{enumerate}
\usepackage[usenames,dvipsnames]{color}
\usepackage[colorlinks=true, pdfstartview=FitV, linkcolor=blue,citecolor=blue,urlcolor=blue]{hyperref}

\usepackage{tikz}
\usepackage[all]{xy}
\usetikzlibrary{matrix}
\usetikzlibrary{decorations.pathreplacing}
\usetikzlibrary{arrows.meta, positioning}

\usepackage{mathtools}

\allowdisplaybreaks[3]

\usepackage{comment}
\usepackage{amsmath}

\usepackage{color}

\usepackage{todonotes}
\usepackage{bbm}

\usepackage{lmodern}

\numberwithin{equation}{section}

\newcommand{\ber}{\begin{red}}
\newcommand{\er}{\end{red}}
\newcommand{\beb}{\begin{blue}}
\newcommand{\eb}{\end{blue}}
\newcommand{\tens}[1]{%
  \mathbin{\mathop{\otimes}\displaylimits_{#1}}%
}

\theoremstyle{plain}
\newtheorem{lemma}{Lemma}[section]
\newtheorem{proposition}[lemma]{Proposition}
\newtheorem{theorem}[lemma]{Theorem}

\newtheorem{corollary}[lemma]{Corollary}

\theoremstyle{definition}
\newtheorem{remark}[lemma]{Remark}
\newtheorem{example}[lemma]{Example}

\newtheorem{definition}[lemma]{Definition}

\newtheorem{pdefinition}[lemma]{Proposition-Definition}

\newcommand{\vac}{\left| 0 \right>}
\newcommand{\n}{\mathfrak{n}}
\newcommand{\g}{\mathfrak{g}}
\newcommand{\h}{\mathfrak{h}}
\newcommand{\RR}{\mathbb{R}}
\newcommand{\CC}{\mathbb{C}}

\newcommand{\ZZ}{\mathbb{Z}}

\newcommand{\w}{\omega}
\newcommand{\HH}{\mathcal{H}}
\newcommand{\superz}{Y\big(\left|-\tfrac{\alpha}{\nu}\right>,Z\big)}
\newcommand{\superw}{Y\big(\left|-\tfrac{\alpha}{\nu}\right>,W\big)}
\newcommand{\expo}{e^{-\frac{1}{\nu}\int \alpha(Z)}}

\newlength{\mylength}
\makeatletter
\newcommand*{\rom}[1]{\expandafter\@slowromancap\romannumeral #1@}
\makeatother

\title{Free Field Realization of Supersymmetric~W-algebras}

\author
[A. Song]{Arim Song}
\address[A. Song]{Department of Mathematical Sciences, Seoul National University, GwanAkRo 1, Gwanak-Gu, Seoul 08826}
\email{ireansong@snu.ac.kr}

\thanks{This research was supported by Basic Science Research Program through the National Research Foundation of Korea (NRF) funded by the Ministry of Education (RS-2023-00272036,  RS-2024-00409689).}

\begin{document}
\maketitle

\begin{abstract}
  We show that supersymmetric(SUSY) W-algebra of generic level can be realized as an intersection of the kernels of the screening operators. Applying this result to principal SUSY W-algebras, we get their free field realization inside SUSY Heisenberg vertex algebras. Furthermore, the screening operators for principal SUSY W-algebras allow us to present them as intersections of the principal SUSY W-algebras associated with $\mathfrak{osp}(1|2)$,  $\mathfrak{osp}(2|2)$, $\mathfrak{osp}(3|2)$ and $\mathfrak{osp}(4|2)$, tensored with SUSY Heisenberg vertex algebras.
\end{abstract}

\section{Introduction} \label{sec: intro}
W-algebras were first introduced in the physical context, particularly in the study of $2$-dimensional conformal field theory \cite{IMY88,DRS92,MR94}. When a Lie (super)algebra with an $\mathfrak{sl}(2)$ embedding is given, we can define the corresponding W-algebra by quantum Hamiltonian reduction \cite{FF90}. After the invention of the vertex algebra language, this process was translated in terms of the vertex algebra cohomology \cite{KRW04}, which in turn indicates the vertex algebra structure of W-algebras. As a supersymmetric analogue of W-algebras, supersymmetric(SUSY) W-algebras are introduced. Since the vertex algebra language did not align with the physicists' words to study supersymmetric field theory, the appropriate language to study them was required. Thanks to the work \cite{Kac98,HK07} of Heluani and Kac, the notion of $N_K=N$ SUSY vertex algebras was developed, which is later used in \cite{MRS21} to define SUSY W-algebras via SUSY BRST construction.

In this paper, we mainly focus on $N_K=1$ SUSY vertex algebras, which we simply refer to as SUSY vertex algebras. The SUSY vertex algebra is a vertex algebra equipped with an odd linear operator $D$ whose square is equal to the translation operator $\partial$. With respect to this odd map $D$, one can find the supersymmetric counterpart for each element. The smallest but the most important toy example of a SUSY vertex algebra is the Neveu-Schwarz vertex algebra, which is the supersymmetric analogue of a Virasoro vertex algebra. The Neveu-Schwarz vertex algebra $NS^c$ of central charge $c\in \CC$ is a vertex algebra freely generated by an even element $L$ and an odd element $G$. The OPEs between the generators are given by 
\begin{equation} \label{eq: Neveu-Schwarz OPE}
  \begin{aligned}
  L(z)L(w) &\sim \frac{\partial_w L(w)}{z-w}+\frac{2L(w)}{(z-w)^2}+\frac{c}{6(z-w)^4}, \\
  L(z)G(w) & \sim \frac{\partial_w G(w)}{z-w}+\frac{3 G(w)}{2(z-w)^2},\\
  G(z)G(w)&\sim \frac{2L(w)}{z-w}+\frac{c}{3(z-w)^3}.
  \end{aligned}
\end{equation}
Denote the coefficients of the fields by $L(z)=\sum_{n\in \ZZ}L_{(n)}z^{-n-1}$ and $G(z)=\sum_{n\in \ZZ}G_{(n)}z^{-n-1}$. Since we have $D^2=\partial$ for $D:=G_{(0)}$, the Neveu-Schwarz vertex algebra is a SUSY vertex algebra. Furthermore, as one can check in \eqref{eq: Neveu-Schwarz OPE}, the element $L$ generates a Virasoro vertex algebra and $G$ is primary of conformal weight $\frac{3}{2}$ along with the simplest OPE with itself. Therefore, the Neveu-Schwarz vertex algebra $NS^c$ can be regarded as the $N=1$ supersymmetric analogue of a Virasoro vertex algebra.

Moreover, one can observe that the commutator with the supersymmetry $D$ connects the two fields $G(z)$ and $2L(z)$. To be specific, it draws the following commutative diagram:
\begin{center}
  \begin{tikzpicture}
    \node (Ln-1) at (-1.5,1) {$-2nL_{(n-1)}$};
    \node (Ln) at (-1.5,-1) {$2L_{(n)}$};
    \node (Gn-1) at (1.5,0) {$-nG_{(n-1)}$};
    \node (Gn) at (1.5,-2) {$G_{(n)}$};
    \draw[thick, ->] (Gn) -- (Ln) node[midway, sloped, below] {\tiny{$[D, -]$}};
    \draw[thick, ->] (Ln) -- (Gn-1) node[midway, sloped, above] {\tiny{$[D, -]$}};
    \draw[thick, ->] (Gn-1) -- (Ln-1) node[midway, sloped, above] {\tiny{$[D, -]$}};
    \draw[thick, ->] (Ln) -- (Ln-1) node[midway, left] {\tiny{$[\partial, -]$}};
    \draw[thick, ->] (Gn) -- (Gn-1) node[midway, right] {\tiny{$[\partial, -]$}};
  \end{tikzpicture}
\end{center}
Emphasizing that the two fields are associated, we write the vertex operators of $G$ and $2L$ as a single unit in either of the following ways
\begin{equation} \label{eq: two ways for SUSY VA}
  G(z)+2\theta L(z), \quad [2L{}_{\lambda}\cdot]+\chi[G{}_{\lambda}\cdot]
\end{equation}
by introducing an odd formal variable $\theta$ or $\chi$. The first part of \eqref{eq: two ways for SUSY VA} generalizes to the definition of SUSY vertex algebras via superfields, and the second part generalizes to the definition via $\Lambda$-brackets \cite{HK07}. We choose the latter one as a main definition of SUSY vertex algebras since the $\Lambda$-brackets are particularly useful in describing the SUSY BRST construction.

The SUSY BRST construction is the cohomological construction of a SUSY W-algebra when a Lie superalgebra $\g$ and its $\mathfrak{osp}(1|2)$ subalgebra $\mathfrak{s}$ are given. The corresponding SUSY W-algebra of level $k\in \CC$ is denoted by  $W^k(\bar{\g},f)$, where $f$ is an odd nilpotent in $\mathfrak{s}$. We expect the SUSY W-algebras to have subalgebras isomorphic to Neveu-Schwarz vertex algebras, just as W-algebras have Virasoro vertex algebras as subalgebras. Furthermore, we expect the SUSY W-algebras to be decomposed into the direct sum of the eigenspaces with respect to the action of the Neveu-Schwarz subalgebra, which enables us to consider the concept of conformal weight in SUSY W-algebras. This property is known as a superconformality, and numerous papers \cite{Ademollo76, MR94} conjectured the superconformality of SUSY W-algebras. We provide properties of SUSY W-algebras in Section \ref{sec: SUSY W-algebras} along with the proof of the superconformality of SUSY W-algebras in Theorem \ref{thm: superconformal for susy W-alg} for noncritical level. In particular, it implies that the SUSY W-algebra of the noncritical level made up of the Lie superalgebra $\g=\mathfrak{osp}(1|2)$ is isomorphic to the Neveu-Schwarz vertex algebra.

The main goal of this paper is to show the free field realization of SUSY W-algebra as a kernel of screening operators. Free field realization is to give a representation of a conformal field theory in free bosons or free fermions \cite{Frenkel95}. In terms of vertex algebras, it means to embed a conformal vertex algebra into a tensor product of free field vertex algebras, such as Heisenberg or free fermion vertex algebras. Since the free field vertex algebras have far simpler structures than any other vertex algebras, showing the free field realization of SUSY W-algebras would facilitate the study of their properties. Furthermore, the screening operators particularly simplify the theory, since each operator ``screens out'' the need to consider other simple roots. To elaborate on the meaning of the screening, we recall the free field realization of W-algebras in \cite{FFduality92}. In the paper, they show that any W-algebra made up of a finite simple Lie algebra and its principal $\mathfrak{sl}(2)$ subalgebra can be embedded into a Heisenberg vertex algebra. Also, the realization is given by the kernel of screening operators, where each screening operator corresponds to one simple root of the Lie algebra. If we only focus on the kernel of a single screening operator, we can see that it is isomorphic to the Virasoro vertex algebra tensored with the Heisenberg vertex algebra. This fact led to the proof of Feigin-Frenkel duality for principal W-algebras. Additionally, it was a main tool in verifying the coset realizations of W-algebras \cite{ACL19}. In the case of Lie superalgebras instead of Lie algebras, the free field realization of them was shown by \cite{Genra17}, when the level $k\in \CC$ is generic. However, further analogous problems for W-algebras associated with Lie superalgebras are still open.

In this paper, we use spectral sequences to prove the free field realization of SUSY W-algebras as kernels of screening operators. This method is inspired by the paper of Genra in \cite{Genra17} which only needs intrinsic information of W-algebras. We set the $\ZZ_{\geq 0}$-grading called the weight on the SUSY BRST complex and consider the induced spectral sequence. Then we observe that the cohomology of the total complex on the second page gives the original SUSY W-algebra. To be explicit, we have
\begin{equation} \label{eq: intro main}
  W^k(\bar{\g},f)\simeq H(E^k_2, d_2),
\end{equation}
where $E^k_2$ and $d_2$ are the total complex and the differential on the second page. Eventually, the isomorphism \eqref{eq: intro main} gives the free field realization of $W^k(\bar{\g},f)$ under some conditions. In particular, the structure of the total complex $E^k_2$ is given by the SUSY Heisenberg vertex algebra, and the action of the differential $d_2$ gives the formula of screening operators for SUSY W-algebras. We dedicate Section \ref{bigsec: FF realization of SUSY W-algebras} to investigate \eqref{eq: intro main} and prove the following main theorem. Furthermore, we see that the isomorphism in \eqref{eq: intro thm statement} is equal to the Miura map $\mu_k$ for SUSY W-algebras. As a corollary, we get the injectivity of $\mu_k$ for generic $k\in \CC$.

\begin{theorem}[Theorem \ref{thm: main} and \ref{thm: main-screening}] \label{thm: intro main}
For generic $k\in \CC$, the SUSY W-algebra of level $k$ is isomorphic to
  \begin{equation} \label{eq: intro W-algebra as kernel}
    W^k(\bar{\g},f)\simeq \bigcap_{[\alpha]\in[I_0]}\ker\Big(\sum_{\beta\in[\alpha]}(f|u_{\beta}){\phi^{\beta}}_{(0|0)}\Big)\subset V^{\xi_k}(\bar{\g}_0),
  \end{equation}
  where ${\phi^{\alpha}}_{(0|0)}$ is given by the adjoint action on $E^k_2$ and $V^{\xi_k}(\bar{\g}_0)$ is the SUSY affine vertex algebra of shifted level $\xi_k$ associated with $\g_0$. In particular, when $\g_0=\h$, \eqref{eq: intro W-algebra as kernel} can be written as  
\begin{equation} \label{eq: intro thm statement}
    W^k(\bar{\g},f)\simeq \bigcap_{\substack{\alpha\in \Pi\\(f|u_{\alpha})\neq 0}}\! \ker \int  e^{-\frac{1}{\nu} \int \alpha(Z)} dZ\subset \widehat{\pi}.
  \end{equation}
  Here, $\widehat{\pi}$ is the SUSY Heisenberg vertex algebra associated with $\h$ and $\int e^{-\frac{1}{\nu}\int \alpha(Z)} dZ$ is the screening operator for each simple root $\alpha \in \Pi$ with $(f|u_{\alpha})\neq 0$.
\end{theorem}
 In Theorem \ref{thm: intro main}, we denote by $\g_0$ the subspace of $\g$ given by the eigenspace of eigenvalue $0$ for the adjoint action of $H\in \mathfrak{osp}(1|2)$ and $\h$ the Cartan subalgebra of $\g$. One of the main ideas in proving this theorem is that we use the quasi-classical limit of SUSY W-algebras called the classical SUSY W-algebras. They have the structure of SUSY Poisson vertex algebras, which are the supersymmetric analogue of Poisson vertex algebras. As for the classical W-algebras cases, the classical SUSY W-algebras have a great advantage in doing computations. Due to Lemma \ref{lem: Zariski dense}, this advantage is especially useful, since the classical SUSY W-algebras explains the behavior of $W^k(\bar{\g},f)$ for generic $k\in \CC$.



The most important class of Lie superalgebras that satisfies $\g_0=\h$ would be the Lie superalgebras that admit the principal $\mathfrak{osp}(1|2)$ embedding. Denote the odd nilpotent element of the principal $\mathfrak{osp}(1|2)$ subalgebra by $f_{\textup{prin}}$ and the dual coxeter number of $\g$ by $h^{\vee}(\g)$. Then for the shifted level $\Psi:=k-h^{\vee}(\g)$, we can analyze the structure of a principal SUSY W-algebra $W^{\Psi}(\bar{\g}):=W^{\Psi}(\bar{\g},f_{\textup{prin}})$ by applying Theorem \ref{thm: intro main}.
\begin{theorem}[Theorem \ref{thm: principal SUSY W-algebra}] \label{thm: intro principal}
  Let $\g$ be a finite basic simple Lie superalgebra which is not the type of $D(2,1;\alpha)$. Then the principal SUSY W-algebra $W^{\Psi}(\bar{\g})$ of generic level can be expressed as the intersection of the vertex algebras in the following list:
  \begin{equation} \label{eq: intro principal list}
    W^{\Psi}(\overline{\mathfrak{osp}(1|2)}),\quad W^{\Psi}(\overline{\mathfrak{osp}(2|2)}),\quad W^{\Psi}(\overline{\mathfrak{osp}(3|2)}),\quad W^{\Psi}(\overline{\mathfrak{osp}(4|2)}),
  \end{equation}
  where the list \eqref{eq: intro principal list} is up to the tensor product of SUSY Heisenberg vertex algebra.
\end{theorem}
Note that the Lie superalgebras $\mathfrak{osp}(1|2)$, $\mathfrak{osp}(2|2)$, $\mathfrak{osp}(3|2)$ and $\mathfrak{osp}(4|2)$ are conjecturally the minimal Lie superalgebras whose associated principal W-algebras have $N=1,2,3,$ and $4$ superconformal structure up to the tensor product of free field vertex algebras \cite{LS84,RSS96,GS88}. As a result of Theorem \ref{thm: intro principal}, we can analyze the structure of any principal SUSY W-algebras on a much smaller scale. To be specific, the study of SUSY W-algebras in \eqref{eq: intro principal list} would help us understand $W^k(\bar{\g})$ for generic $k\in \CC$. Recall from earlier in Section \ref{sec: intro} that the analogous statement of Theorem \ref{thm: intro principal} holds in the study of principal W-algebras associated with a Lie algebra $\g$ \cite{FFduality92}. In this case, the list \eqref{eq: intro principal list} is replaced with the W-algebra associated with $\mathfrak{sl}(2)$, in other words, the Virasoro vertex algebra. Similarly to the proof of Feigin-Frenkel duality, the author anticipates that each SUSY W-algebra in the list \eqref{eq: intro principal list} would also have nontrivial dualities, which in turn gives the Feigin-Frenkel type duality for principal SUSY W-algebras.

The structure of this paper is as follows. In Section \ref{bigsec: SUSY VA vs. VA}, we recall SUSY vertex algebras and compare this structure with vertex algebras in Theorem \ref{thm: SUSY VA vs. VA}. Next, we see examples of SUSY vertex algebras in Section \ref{sec: examples of susy va}, namely, SUSY affine vertex algebras, SUSY fermionic vertex algebras, and SUSY Heisenberg vertex algebras. Section \ref{sec: SUSY W-algebras} provides a summary of SUSY W-algebras and their properties, accompanied by Theorem \ref{thm: superconformal for susy W-alg}, which shows the superconformality of SUSY W-algebras. Next, we prove the main Theorem \ref{thm: main-screening}, the free field realization of SUSY W-algebras in Section \ref{bigsec: FF realization of SUSY W-algebras}. The application of the main theorem to principal SUSY W-algebras is presented in Section \ref{sec: application to principal SUSY W-algebras}. Finally, we deal with the technical proofs supporting Section \ref{bigsec: FF realization of SUSY W-algebras}, and provide another definition of SUSY vertex algebra with superfield formalism in the appendix.

\section{Supersymmetric vertex algebras vs. vertex algebras}\label{bigsec: SUSY VA vs. VA} \hfill \\
In this section, we review the definition of supersymmetric vertex algebras, especially the one via the $\Lambda$-brackets. Next, by comparing their structure with vertex algebra, we get Theorem \ref{thm: SUSY VA vs. VA}, the simpler description of supersymmetric vertex algebras. Here, the supersymmetric vertex algebra means an $N_K=1$ supersymmetric vertex algebra for the $N_K=N$ supersymmetric vertex algebra introduced by Heluani-Kac \cite{HK07}.

\subsection{Notations}\label{sec: notations} \hfill \\
The notations in this section will be used throughout this paper. Consider any vector superspace $V=V_{\bar{0}}\oplus V_{\bar{1}}$. We write $p(A)=0$ (resp. $1$) if $A\in V_{\bar{0}}$ (resp. $V_{\bar{1}}$). For an odd endomorphism $D$ and an even endomorphism $\partial$ on $V$, we write these as a tuple
\begin{equation} \label{eq: nabla}
  \nabla=(\partial, D)
\end{equation}
if they satisfy $D^2=\partial$. For the above $\nabla$, denote the unital superassociative algebra generated by the elements as $\CC[\nabla]$. For formal variables $\lambda, \chi$ with $\lambda$ even and $\chi$ odd, we represent them as $\Lambda=(\lambda, \chi)$ when
\begin{equation} \label{eq: Lambda}
  [\chi,\chi]=-2\lambda, \quad [\lambda,\chi]=0.
\end{equation}
Let $\CC[\Lambda]$ be the unital superassociative algebra generated by $\Lambda$.
The superspace $\CC[\Lambda]\otimes V$ can be viewed as a $\CC[\nabla]$-module via
\begin{equation} \label{eq: D chi commutator}
  D(\chi A)=-\chi(DA)+2\lambda A
\end{equation}
for $A\in V$, in which the tensor notations are omitted.

\subsection{Supersymmetric vertex algebra}\hfill \\
We review the definitions of supersymmetric vertex algebras and superconformal vectors. In the original paper \cite{HK07} of Heluani and Kac, they proposed two equivalent definitions of supersymmetric vertex algebras, via superfields and via $\Lambda$-brackets. We give the latter as a definition of supersymmetric vertex algebras in this paper. For the former, we briefly introduce it in Appendix \ref{appendix: superfield formalism}. We first introduce supersymmetric Lie conformal algebra, which attributes to define supersymmetric vertex algebras.

\begin{definition}[\cite{HK07}] \label{def: SUSY LCA}
  Let $R$ be a superspace with a tuple of endomorphisms $\nabla$. We say that $R$ is a \textit{supersymmetric(SUSY) Lie conformal algebra(LCA)} if it is an $\CC[\nabla]$-module equipped with an odd linear map
  \begin{equation} \label{eq: Lambda bracket}
    [\cdot {}_{\Lambda} \cdot] : R\otimes R  \rightarrow \CC[\Lambda]\otimes R
  \end{equation}
  satisfying
  \begin{enumerate}[]
    \item (super sesquilinearity) $[DA{}_{\Lambda}B]=\chi[A{}_{\Lambda}B]$, $[A{}_{\Lambda}DB]=(-1)^{p(A)+1}(\chi+D)[A{}_{\Lambda}B]$,
    \item (super skew-symmetry) $[B{}_{\Lambda}A]=(-1)^{p(A)p(B)}[A{}_{-\Lambda-\nabla}B]$,
    \item (super Jacobi identity)
    \begin{equation*}
    [A{}_{\Lambda}[B{}_{\Gamma}C]]=(-1)^{p(A)+1}[[A{}_{\Lambda}B]{}_{\Lambda+\Gamma}C]+(-1)^{(p(A)+1)(p(B)+1)}[B{}_{\Gamma}[A{}_{\Lambda}B]]
    \end{equation*}
  \end{enumerate}
  for any $A,B,C\in R$. The action of $\nabla$ on the space $\CC[\Lambda]\otimes R$ is given by \eqref{eq: D chi commutator}, and $\Gamma=(\gamma, \mu)$ is a tuple of variables satisfying \eqref{eq: Lambda}, where $\chi$ and $\mu$ are supercommutative.
\end{definition}
Given any SUSY LCA, we denote the coefficients of the $\Lambda$-bracket \eqref{eq: Lambda bracket} by
\begin{equation} \label{eq: Lambda bracket coefficients}
  [A{}_{\Lambda}B]=\sum_{n\geq 0}\frac{\lambda^n}{n!} A_{(n|0)}B+\chi \sum_{n\geq 0} \frac{\lambda^n}{n!} A_{(n|1)}B.
\end{equation}
The computation of the RHS of super skew-symmetry needs to be noted. If we ignore the sign $(-1)^{p(A)p(B)}$, the RHS can be computed by applying the tuple $-\Lambda-\nabla=(-\lambda-\partial, -\chi-D)$ in replace of $\Gamma$ in $[A{}_{\Gamma}B]$. Using the notation in \eqref{eq: Lambda bracket coefficients},
\begin{equation*}
  [A{}_{-\Lambda-\nabla}B]=\sum_{n\geq 0}\frac{(-\lambda-\partial)^n}{n!} A_{(n|0)}B-(\chi+D) \sum_{n\geq 0} \frac{(-\lambda-\partial)^n}{n!} A_{(n|1)}B.
\end{equation*}

For any binary operation $:\ :$ on a SUSY LCA $R$, we say that the operation is \textit{super quasi-commutative} and \textit{super quasi-associative} if the following equations \eqref{eq: quasi-commutativity} and \eqref{eq: quasi-associativity} hold for any elements of $R$, respectively:
\begin{gather}
  :\!AB\!:-(-1)^{p(A)p(B)}:\!BA\!:=-\sum_{n\geq 1}\frac{(-\partial)^n}{n!}A_{(n-1|1)}B, \label{eq: quasi-commutativity}\\
  ::\!AB\!:\!C\!:-:\!A\!:\!BC\!::=\sum_{n\geq 0}A_{(-n-2|1)}B_{(n|1)}C+(-1)^{p(A)p(B)}\sum_{n\geq 0}B_{(-n-2|1)}A_{(n|1)}C. \label{eq: quasi-associativity}
\end{gather}
In \eqref{eq: quasi-associativity}, $A_{(-n-2|1)}B=\frac{1}{(n+1)!}:\!(\partial^{n+1} A)B\!:$ for positive $n$.

\begin{definition}[\cite{HK07}] \label{def: SUSY VA}
  Consider a superspace $V$ with a tuple of endomorphisms $\nabla$ and an even distinguished element $\vac\in V$. The superspace $V$ is called a \textit{supersymmetric(SUSY) vertex algebra} if
  \begin{itemize}
    \item $V$ is a SUSY LCA with $\Lambda$-bracket $[\cdot {}_{\Lambda}\cdot]$,
    \item $V$ is equipped with a parity-preserving binary operation $:\ :$ such that the algebra $(V, \vac, :\ :, D\in \nabla)$ is a super quasi-commutative and super quasi-associative unital algebra with an odd derivation $D$,
    \item the $\Lambda$-bracket and the binary operation $:\ :$ is related by the super Wick formula, i.e., 
    \begin{equation}\label{eq: Wick formula}
      [A{}_{\Lambda}:\!BC\!:]=:\![A{}_{\Lambda}B]C\!:+(-1)^{(p(A)+1)p(B)}:\!B[A{}_{\Lambda}C]\!:+\sum_{n\geq 1}\frac{\lambda^n}{n!}[A{}_{\Lambda}B]_{(n-1|1)}C
    \end{equation}
    for any $A,B,C\in V$.
  \end{itemize}
\end{definition}
Note that the $(n-1|1)$th product on the last term of \eqref{eq: Wick formula} only applies to the coefficients. To be specific, the last term equals the following expression.
\begin{equation*}
  \sum_{n\geq 1} \frac{\lambda^n}{n!} \big(\sum_{m\geq 0} \frac{\lambda^m}{m!}(A_{(m|0)}B)_{(n-1|1)}C+\chi \sum_{m\geq 0} \frac{\lambda^m}{m!}(A_{(m|1)}B)_{(n-1|1)}C\big).
\end{equation*}

The binary operation $:\ :$ of SUSY vertex algebra is called the \textit{normally ordered product}. Note that the notation and the name for this binary operation are the same as the one in vertex algebra. In fact, SUSY vertex algebra has a vertex algebra structure whose normally ordered product is given by $:\ :$. We check this in Theorem \ref{thm: SUSY VA vs. VA} shortly.

Recall that one of the most important features of vertex algebra that should be checked is whether it has a conformal vector. The action of the conformal vector gives a grading to the vertex algebra, called conformal weight. In SUSY vertex algebra, we have the notion of a superconformal vector, which plays a similar role in the supersymmetric setting.
\begin{definition}[\cite{HK07}] \label{def: superconformal vector}
  Let $V$ be a SUSY vertex algebra. An odd element $G\in V$ is called a \textit{superconformal vector of central charge $c\in \CC$} if it satisfies the two following conditions.
  \begin{enumerate}[(i)]
    \item The $\Lambda$-bracket between $G$ and itself is
    \begin{equation} \label{eq: conformal Lambda-bracket}
      [G{}_{\Lambda}G]=\left(2\partial+3\lambda+\chi D\right)G+\frac{\lambda^2\chi}{3}c.
    \end{equation} 
    \item $V$ decomposes into the direct sum $V=\bigoplus_{\Delta \geq M}V_{\Delta}$ for some $M\in \RR$, where each $V_\Delta$ for $\Delta\in \RR$ consists of the elements $a\in V$ satisfying
    \begin{equation} \label{eq: conformal weight Delta}
      [G{}_{\Lambda}a]=\left(2\partial+2\Delta \lambda +\chi D\right)a+O(\Lambda^2).
    \end{equation}     
    In the above, $O(\Lambda^2)$ is a polynomial in $\lambda$ and $\chi$, whose constant and linear terms are vanishing. 
  \end{enumerate}
When $V$ has a superconformal vector and $a\in V_{\Delta}$, we say that $a$ has a \textit{conformal weight} $\Delta$. Furthermore, we call $V=\bigoplus V_{\Delta}$ the \textit{conformal weight decomposition} of $V$. In particular, when $O(\Lambda^2)=0$ in \eqref{eq: conformal weight Delta},  we say that $a$ is \textit{primary}.
\end{definition}

\subsection{SUSY vertex algebras vs. Vertex algebras} \label{sec: SUSY VA vs. VA} \hfill \\
In this section, we compare the definitions of vertex algebras and SUSY vertex algebras to obtain the simpler description of SUSY vertex algebra in Theorem \ref{thm: SUSY VA vs. VA}. The idea of the proof originates from the author's previous paper \cite{RSS23N=2} with Ragoucy and Suh, in which they simplified the notion of SUSY Poisson vertex algebra. For the definition of vertex algebras, refer to Chapter 2 in \cite{Kac98}. Recall that on a vertex algebra $V$, an odd endomorphism $D$ is said to be an \textit{odd derivation} if it is an odd derivation for all the $(n)$th product, i.e., $D$ is an odd derivation for the normally ordered product $:\ :$, and satisfies
\begin{equation} \label{eq: D commutator}
  D[A{}_{\lambda}B]=[DA{}_{\lambda}B]+(-1)^{p(A)}[A{}_{\lambda}DB]
\end{equation}
for any $A,B\in V$. In particular, we have $D\vac =0$ for an odd derivation $D$.

\begin{theorem} \label{thm: SUSY VA vs. VA}
  Every SUSY vertex algebra is a vertex algebra. Conversely, any vertex algebra with an odd derivation $D$ satisfying $D^2=\partial$ is a SUSY vertex algebra.
\end{theorem}
\begin{proof}
  Let $V$ be a SUSY vertex algebra. For the $\Lambda$-bracket of $V$, take the parity preserving linear map
  \begin{equation} \label{eq: VA structure of SUSY VA}
    \begin{aligned}
    [\cdot {}_{\lambda} \cdot]: V\otimes V &\rightarrow \CC[\lambda]\otimes V,\\ A\otimes B&\mapsto [A{}_{\lambda}B]=\sum_{n\geq 0}\frac{\lambda^n}{n!}A_{(n|1)}B.
    \end{aligned}
  \end{equation}
  Using \eqref{eq: VA structure of SUSY VA}, the $\Lambda$-bracket for $V$ is written as
  \begin{equation} \label{eq: SUSY nonSUSY lambda bracket relation}
    [A{}_{\Lambda}B]=[DA{}_{\lambda}B]+\chi[A{}_{\lambda}B],
  \end{equation}
  because the super sesquilinearity $[DA{}_{\Lambda}B]=\chi[A{}_{\Lambda}B]$ implies that $(DA)_{(n|1)}B=A_{(n|0)}B$ for any $A,B\in V$. Note from \eqref{eq: SUSY nonSUSY lambda bracket relation} that $\eqref{eq: VA structure of SUSY VA}$ captures the $\chi$ containing part of the $\Lambda$-bracket. By comparing the $\chi$ containing part of the super sesquilinearity, super skew-symmetry and super Jacobi identity in Definition \ref{def: SUSY LCA}, one can say that $V$ with $\lambda$-bracket \eqref{eq: VA structure of SUSY VA} is a Lie conformal algebra. For the normally ordered product $:\, :$ on $V$, the quasi-commutativity and quasi-associativity of $:\, :$ directly follow from \eqref{eq: quasi-commutativity} and \eqref{eq: quasi-associativity}. Lastly, one gets the Wick formula for $\eqref{eq: VA structure of SUSY VA}$ by comparing $\chi$ containing part of the super Wick formula as follows. Using the equality \eqref{eq: SUSY nonSUSY lambda bracket relation}, the super Wick formula can be written as
  \begin{equation*}
    \begin{aligned}
      [DA&{}_{\lambda}:\!BC\!:]+\chi[A{}_{\lambda}:\!BC\!:]\\
      =&:\![DA{}_{\lambda}B]C\!:+\chi\!:\![A{}_{\lambda}B]C\!:+(-1)^{(p(A)+1)p(B)}:\!B[DA{}_{\lambda}C]\!:+(-1)^{p(A)p(B)}\chi\!:\!B[DA{}_{\lambda}C]\!:\\
      &+\sum_{n\geq 1} \frac{\lambda^n}{n!} [DA{}_{\lambda}B]_{(n-1|1)}C+\chi \sum_{n\geq 1} \frac{\lambda^n}{n!} [A{}_{\lambda}B]_{(n-1|1)}C.
    \end{aligned}
  \end{equation*}
  Picking up the $\chi$ containing term, one gets the Wick formula for \eqref{eq: VA structure of SUSY VA}.

  Conversely, assume that $V$ is a vertex algebra with an odd derivation $D$ satisfying $D^2=\partial$. Using the $\lambda$-bracket $[\cdot{}_{\lambda}\cdot]$ for vertex algebra $V$, define an odd linear map $[\cdot{}_{\Lambda} \cdot]$ on $V$ by \eqref{eq: SUSY nonSUSY lambda bracket relation}. Then one can show all the required properties of the $\Lambda$-bracket by direct computation. Refer to Section 3 of \cite{RSS23N=2} for the detailed proof. In \cite{RSS23N=2}, they deal with the classical version, but the computation process is similar.
\end{proof}

In the rest of this paper, whenever we mention the vertex algebra structure of a SUSY vertex algebra, it means that we are considering \eqref{eq: VA structure of SUSY VA}. Furthermore, we denote
\begin{equation*}
  A_{(n)}=A_{(n|1)}
\end{equation*}
for any $n\in \ZZ_{\geq 0}$ in SUSY vertex algebra. Note that this coincides with the $(n)$th product convention of vertex algebra related to the coefficients of the $\lambda$-bracket. In addition, we define \textit{SUSY vertex subalgebras(resp. SUSY ideals)} by requiring them not just to be vertex subalgebras(resp. ideals), but also to be closed under $D$.

One of the advantages of Theorem \ref{thm: SUSY VA vs. VA} is that we can find a new supersymmetric structure whenever we can take an appropriate odd derivation $D$ on a vertex algebra. We present here one of the examples, the Neveu-Schwarz vertex algebra.
\begin{example} [Neveu-Schwarz vertex algebra]
  For $c\in \CC$, the Neveu-Schwarz vertex algebra $NS^c$ is a vertex algebra freely generated by even $L$ and odd $G$. The $\lambda$-brackets between them are given by
  \begin{equation*}
    [L{}_{\lambda}L]=(\partial +2\lambda)L+\frac{c}{6}\lambda^3, \quad
    [L{}_{\lambda}G]=\big(\partial+\frac{3}{2}\lambda\big)G, \quad [G{}_{\lambda}G]=2L+\frac{c}{3}\lambda^2.
  \end{equation*}
  Denote the coefficients of the $\lambda$-brackets by $[A{}_{\lambda}B]=\sum_{n\geq 0} \frac{\lambda^n}{n!}A_{(n)}B$. Then for $D=G_{(0)}$, we have
  \begin{equation*}
    [D,D]=[G_{(0)},G_{(0)}]=(G_{(0)}G)_{(0)}=(2L)_{(0)}=2\partial,
  \end{equation*}
  which implies that $D^2=\partial$. Any $(0)$th product of an odd element is an odd derivation on a vertex algebra. Thus, $NS^c$ is a SUSY vertex algebra with $D=G_{(0)}$. In particular, its $\Lambda$-bracket is given by \eqref{eq: SUSY nonSUSY lambda bracket relation}. Note that the $\Lambda$-bracket between $G$ with itself is written as
  \begin{equation*}
    [G{}_{\Lambda}G]=[2L{}_{\lambda}G]+\chi[G{}_{\lambda}G]=(2\partial+3\lambda+\chi D)G+\frac{\lambda^2 \chi}{3}c,
  \end{equation*}
  which shows that $G$ is a superconformal vector of $NS^c$.
\end{example}

\begin{remark}
  Note in the above example that $\frac{1}{2}DG=L$ is a conformal vector for a superconformal vector $G$.  This fact holds in general and the conformal weight \eqref{eq: conformal weight Delta} is the same as the one given by the conformal vector $\frac{1}{2}DG$. These can be easily checked by combining \eqref{eq: conformal Lambda-bracket} and \eqref{eq: SUSY nonSUSY lambda bracket relation}.
\end{remark}

Furthermore, Theorem \ref{thm: SUSY VA vs. VA} lets us think of the quasi-classical limit of SUSY vertex algebras naturally since it already has a vertex algebra structure. In Proposition \ref{prop: quasi-classical limit}, we see that the quasi-classical limit of SUSY vertex algebra gives a SUSY Poisson vertex algebra. Refer to \cite{DK06} for the quasi-classical limit of vertex algebras. We briefly recall the notion of SUSY Poisson vertex algebra.

\begin{pdefinition}[\cite{RSS23N=2}] Let $P$ be a Poisson vertex algebra with a $\lambda$-bracket $\{\cdot {}_{\lambda}\cdot\}$. Assume that $P$ has an odd derivation $D$ satisfying
  \begin{equation} \label{eq: susy pva D assumption}
    D^2=\partial, \quad D\{a{}_{\lambda}b\}=\{Da{}_{\lambda}b\}+(-1)^{p(a)}\{a{}_{\lambda}Db\}.
  \end{equation}
  Then, $P$ is a \textit{supersymmetric Poisson vertex algebra}. In this case, we write its $\Lambda$-bracket $\{\cdot {}_{\Lambda}\cdot\}: P \otimes P \rightarrow \CC[\Lambda]\otimes P$ as
  \begin{equation}
    \{a{}_{\Lambda}b\}=\{Da{}_{\lambda}b\}+\chi \{a{}_{\lambda}Db\}.
  \end{equation}
\end{pdefinition}

\begin{proposition} \label{prop: quasi-classical limit}
  Consider the family of SUSY vertex algebras
  \[(V_y, \vac_y, \nabla_y, [\cdot{}_{\Lambda}\cdot]_y, :\ :_y)\]
  over $\CC[y]$ such that $V_y$ is a free $\CC[y]$-module, and $[V_y{}_{\Lambda}V_y]\subset \CC[\Lambda]\otimes y V_y$. Take the \textit{quasi-classical limit} of this family $V_y$ of SUSY vertex algebras by regarding them as vertex algebras over $\CC[y]$. As a result, the algebra obtained from the quasi-classical limit is a SUSY Poisson vertex algebra.
\end{proposition}
\begin{proof}
  Due to Theorem \ref{thm: SUSY VA vs. VA}, each $V_y$ is a vertex algebra with an odd derivation $D_y$ satisfying $D_y^2=\partial_y$. To be specific, $D_y$ is an odd derivation with respect to $:\ :_y$ and satisfies
  \begin{equation} \label{eq: quasi-classical limit process}
    D_y[A{}_{\lambda}B]_y=[D_yA{}_{\lambda}B]_y+(-1)^{p(A)}[A{}_{\lambda}D_yB]_y
  \end{equation}
  for any $A,B\in V_y$. Take a quasi-classical limit of this family so that we get a Poisson vertex algebra $P:=V_y/y V_y$ with Poisson $\lambda$-bracket
  \begin{equation*}
    \{A{}_{\lambda}B\}=\frac{1}{y}[A{}_{\lambda}B]_y \vert_{y=0},
  \end{equation*}
  and an even derivation $\partial$ induced from $\partial_y$. Let $D$ be the odd endomorphism on $P$ induced by $D_y$. From the properties of $D_y$ including \eqref{eq: quasi-classical limit process}, one can conclude that $D$ is an odd derivation satisfying the assumptions \eqref{eq: susy pva D assumption}. Thus, the resulting Poisson vertex algebra is a SUSY Poisson vertex algebra with Poisson $\Lambda$-bracket
  \begin{equation} \label{eq: quasi-classical Lambda bracket}
    \{A{}_{\Lambda}B\}=\{DA{}_{\lambda}B\}+\chi\{A{}_{\lambda}B\}=\frac{1}{y}[A{}_{\Lambda}B]_y \vert_{y=0}.
  \end{equation}
\end{proof}

\section{Examples of SUSY vertex algebras}\label{sec: examples of susy va}
In this section, we present a few major examples of SUSY vertex algebras, namely, SUSY affine vertex algebras, SUSY fermionic vertex algebras, and SUSY Heisenberg vertex algebras. The first two, introduced in Section \ref{sec: SUSY affine and SUSY fermionic}, will be used as ingredients to describe the SUSY BRST complexes for SUSY W-algebras. In Section \ref{sec: SUSY Heisenberg and Fock}, we introduce the last example, which plays a key role in Section \ref{bigsec: FF realization of SUSY W-algebras}.
\subsection{SUSY affine vertex algebras and SUSY fermionic vertex algebras} \label{sec: SUSY affine and SUSY fermionic} \hfill \\
For a SUSY vertex algebra $V$ and its totally ordered subset $\mathcal{B}=\{A_i\mid i\in \mathcal{I}\}$, we say that the SUSY vertex algebra $V$ is \textit{freely generated} by $\mathcal{B}$ if $V$ is freely generated by $\mathcal{B}\sqcup D\mathcal{B}$ as a vertex algebra. To avoid confusion, we make it precise here that the `freely generated' means the supersymmetric sense from now on. In this section, we denote the parity-reversed space of a vector superspace $\mathfrak{U}$ by $\bar{\mathfrak{U}}$ and denote its element by $\bar{u}$ for $u\in \mathfrak{U} $. Before we jump into the two examples, we introduce the reconstruction theorem, which demonstrates the supersymmetric version of the universal enveloping vertex algebra for a given SUSY LCA.
 
\begin{theorem} [Reconstruction theorem \cite{HK07}]\label{thm: universal enveloping SUSY va}
  Let $R$ be a SUSY LCA and $\mathcal{B}$ be its totally ordered subset. Assume that $\mathcal{B}\sqcup D\mathcal{B}$ forms a $\CC$-basis of $R$. Then there exists a unique supersymmetric vertex algebra $V(R)$ such that
  \begin{enumerate}[(i)]
    \item $V(R)$ is freely generated by $\mathcal{B}$,
    \item $D$ is an odd derivation of $V(R)$,
    \item the $\Lambda$-bracket of $R$ extends to that of $V(R)$ via the super Wick formula.
  \end{enumerate}
\end{theorem}
For any SUSY LCA $R$, the SUSY vertex algebra $V(R)$ in Theorem \ref{thm: universal enveloping SUSY va} is called \textit{the universal enveloping supersymmetric vertex algebra} of $R$.

\begin{remark}
  Taking the $\lambda$-bracket as in \eqref{eq: VA structure of SUSY VA}, one can see that a SUSY LCA $R$ also has a Lie conformal algebra structure. Then its universal enveloping vertex algebra $\widetilde{V}(R)$ is isomorphic to $V(R)$ as vertex algebras. It follows from the uniqueness of $V(R)$, since we can extend the endomorphism $D$ of $R$ to $\widetilde{V}(R)$ by
  \[D(:\!AB\!:)=:\!(DA)B\!:+(-1)^{p(A)}:\!A(DB)\!:.\]
\end{remark}

\begin{example} [SUSY affine vertex algebras] \label{ex: susy affine va}
  Let $\g$ be a finite simple or abelian Lie superalgebra endowed with a nondegenerate even supersymmetric invariant bilinear form $(\cdot |\cdot).$ Consider the $\CC[\nabla]$-module $R_{\textup{cur}}=\CC[\nabla]\otimes \bar{\g} \oplus \CC K$ with $\nabla$ acting trivially on $K$. Define the $\Lambda$-bracket on the space $R_{\textup{cur}}$ by
  \begin{equation} \label{eq: susy affine Lambda-bracket}
    [\bar{A}{}_{\Lambda}\bar{B}]=(-1)^{p(A)(p(B)+1)}\overline{[A,B]}+\chi (A|B) K
  \end{equation}
  for $\bar{A},\bar{B}\in \bar{\g}$, and let $K$ be central. Note that the $\Lambda$-bracket \eqref{eq: susy affine Lambda-bracket} can be extended to the whole space $R_{\textup{cur}}$ via super sesquilinearity. Then the extended $\Lambda$-bracket makes $R_{\textup{cur}}$ a SUSY LCA. For $k\in \CC$, the SUSY vertex algebra
  \begin{equation*}
    V^k(\bar{\g})=V(R_{\textup{cur}})/\langle K-(k+h^{\vee})\vac \rangle
  \end{equation*}
  is called the \textit{SUSY affine vertex algebra of level $k\in \CC$} associated with $\g$. Here, $h^{\vee}$ is the dual coxeter number of $\g$, i.e., $2 h^{\vee}$ is the eigenvalue of the Casimir element on $\g$ with the adjoint action. Also, $V(R_{\textup{cur}})$ is the universal enveloping SUSY vertex algebra of $R_{\textup{cur}}$, and $\langle K-(k+h^{\vee})\vac \rangle$ is the SUSY vertex algebra ideal of $V(R_{\textup{cur}})$ generated by $K-(k+h^{\vee})\vac$. This SUSY vertex algebra has a superconformal vector
  \begin{equation} \label{eq: Kac-Todorov}
    \omega=\frac{1}{k+h^{\vee}}\sum_{i,j} (v_i|v_j):\bar{v}^i (D\bar{v}^j):+\frac{1}{2(k+h^{\vee})^2}\sum_{i,j,r} (-1)^{p(v_j)}(v_i|[v_j,v_r]):\bar{v}^k \bar{v}^j \bar{v}^r:
  \end{equation}
  of central charge $c_k=\frac{k \textup{sdim} \g}{k+h^{\vee}}+\frac{1}{2}\textup{sdim} \g$, due to the Kac-Todorov construction \cite{KT85}, where $\textup{sdim}(\g)$ is the superdimension of $\g$. In \eqref{eq: Kac-Todorov}, $\{v_i\}$ and $\{v^i\}$ are the homogeneous dual bases of $\g$ such that $(v^i|v_j)=\delta_{i,j}$. Each $\bar{A}\in \bar{\g}$ is primary of conformal weight $\frac{1}{2}$ with respect to $\omega$.
\end{example}
\begin{remark}
  The definition of SUSY affine vertex algebra is different from the original definition in Example 5.9 of \cite{HK07}, where the only difference is the sign in \eqref{eq: susy affine Lambda-bracket}. One can easily see that the map \eqref{eq: susy affine equivalence map} from the SUSY affine vertex algebra in Example \ref{ex: susy affine va} to the original SUSY affine vertex algebra gives an isomorphism.
  \begin{equation} \label{eq: susy affine equivalence map}
    \bar{A}\mapsto \left\{\begin{array}{cc}
      \bar{A} & \text{ if $A$ is even}\\
      \sqrt{-1}\bar{A} & \text{ if $A$ is odd.}
    \end{array}\right.
  \end{equation}
\end{remark}
\begin{example}[Generalization of SUSY affine vertex algebras] \label{ex: susy affine va generalized}
  We can generalize the definition of SUSY affine vertex algebras to any finite Lie superalgebra $\g$ with an even supersymmetric invariant bilinear form $(\cdot |\cdot)$.  Consider the killing form $\kappa_{\g}$ on $\g$. Using the similar argument in Example \ref{ex: susy affine va}, one can construct the SUSY vertex algebra freely generated by the basis elements of $\bar{\g}$ with the $\Lambda$-bracket
  \begin{equation*}
    [\bar{A}{}_{\Lambda}\bar{B}]=(-1)^{p(A)(p(B)+1)}\overline{[A,B]}+\chi \tau_k(A|B)
  \end{equation*}
  for $A,B\in \g$, where $\tau_k$ is the bilinear form on $\g$ as follows:
  \[\tau_k(A|B)=k(A|B)+\frac{1}{2}\kappa_{\g}(A|B).\]
  This SUSY vertex algebra is also called the \textit{SUSY affine vertex algebra of level $k$} associated with $\g$ and denoted by $V^k(\bar{\g})$. For a simple or abelian $\g$, the following relation \cite{KW94}
  \begin{equation*}
    \kappa_{\g}(A|B)=2h^{\vee}(A|B), \quad A,B\in \g
  \end{equation*}
  shows that the two definitions in Example \ref{ex: susy affine va} and \ref{ex: susy affine va generalized} are identical.
\end{example}

\begin{example} [SUSY fermionic vertex algebras] \label{ex: susy fermionic}
  Let $\mathfrak{A}$ be a vector superspace with an even nondegenerate skew-supersymmetric bilinear form $\langle \cdot |\cdot \rangle$. The skew-supersymmetry means that the induced bilinear form on the parity-reversed space $\bar{\mathfrak{A}}$ is supersymmetric. Consider the $\CC[\nabla]$-module $R_{\textup{ch}}=\CC[\nabla]\otimes \bar{\mathfrak{A}} \oplus \CC K$. This superspace has a SUSY LCA structure via the $\Lambda$-bracket
  \begin{equation} \label{eq: fermionic Lambda-bracket}
    [\bar{A}{}_{\Lambda}\bar{B}]=\langle A|B \rangle K
  \end{equation}
  for $\bar{A},\bar{B}\in \bar{\mathfrak{A}}$, and central $K$. Note that \eqref{eq: fermionic Lambda-bracket} satisfies the conditions in Definition \ref{def: SUSY LCA}. For example, the super skew-symmetry can be checked as follows.
  \begin{equation*}
    [\bar{B}{}_{\Lambda}\bar{A}]=\langle B|A \rangle K=(-1)^{(p(A)+1)(p(B)+1)}\langle A|B\rangle K = (-1)^{p(\bar{A})p(\bar{B})}[\bar{A}{}_{-\Lambda-\nabla}\bar{B}].
  \end{equation*}
   For the universal enveloping SUSY vertex algebra $V(R_{\textup{ch}})$ of $R_{\textup{ch}}$, define the \textit{SUSY fermionic vertex algebra} associated with $\mathfrak{A}$ to be the SUSY vertex algebra
  \begin{equation*}
    F(\bar{\mathfrak{A}})=V(R_{\textup{ch}})/\langle K-\vac \rangle,
  \end{equation*}
  where $\langle K-\vac \rangle$ is the SUSY vertex algebra ideal of $V(R_{\text{ch}})$.
\end{example}

\subsection{SUSY Heisenberg vertex algebras and Fock representations} \label{sec: SUSY Heisenberg and Fock} \hfill \\
Recall that the direct sum of the Fock representations of Heisenberg algebra gives rise to the vertex algebra called lattice vertex algebra in nonSUSY theory. The supersymmetric analogue of lattice vertex algebra has been explored in the paper \cite{HK08lattice} of Heluani-Kac, so it conversely motivates the definitions of SUSY Heisenberg algebras and their Fock representations in this section. The algebras defined in this section are the essential ingredients for the free field realization of SUSY W-algebras in Section \ref{bigsec: FF realization of SUSY W-algebras}.

\begin{definition} \label{def: susy heisenberg}
  Let $\mathfrak{a}$ be a finite-dimensional vector space with a symmetric bilinear form $(\cdot |\cdot)$. Let $l=\dim(\mathfrak{a})$ and fix a basis $\{v_1,\cdots v_l\}$ of $\mathfrak{a}$. The \textit{SUSY Heisenberg algebra} $\mathcal{H}^{\mathfrak{a}}$ associated with $\mathfrak{a}$ is an infinite-dimensional Lie superalgebra spanned by the linearly independent elements
  \[\bar{b}_i(n),\ D\bar{b}_i(n),\ \mathbbm{1}, \quad i=1, \cdots l,\ n\in \ZZ,\]
  where $D\bar{b}_i(n)$ and $\mathbbm{1}$ are even, while $\bar{b}_i(n)$ are odd. The Lie brackets between them are given by
    \begin{equation*}
    \begin{gathered}
    [D\bar{b}_i(m), D\bar{b}_j(n)]=m(v_i|v_j)\delta_{m+n,0}\mathbbm{1},\\
    [\bar{b}_i(m),\bar{b}_j(n)]=(v_i|v_j)\delta_{m+n,-1}\mathbbm{1},\quad  [D\bar{b}_i(m),  \bar{b}_j(n)]=0,
    \end{gathered}
  \end{equation*}
  while $\mathbbm{1}$ being central.
\end{definition}
Notice that in Definition \ref{def: susy heisenberg}, the space $\mathfrak{a}$ consists only of even vectors but $\HH^{\mathfrak{a}}$ is a vector superspace because $\bar{b_i}(n)$'s have an odd parity.

\begin{remark} \label{rmk: susy heisenberg enveloping algebra}
  The SUSY Heisenberg algebra $\HH^{\mathfrak{a}}$ is independent of the choice of basis of $\mathfrak{a}$. It follows from the fact that it is the universal enveloping Lie algebra of the vertex algebra $V^{1}(\bar{\mathfrak{a}})$. By decomposing the $\Lambda$-bracket \eqref{eq: susy affine Lambda-bracket} of $V^{1}(\bar{\mathfrak{a}})$ with the equality \eqref{eq: SUSY nonSUSY lambda bracket relation}, we have
  \begin{equation} \label{eq: susy heisenberg Lambda bracket decomposition}
    [D\bar{v}_i{}_{\lambda} D\bar{v}_j]=(v_i|v_j)\lambda, \quad [\bar{v}_i{}_{\lambda}\bar{v}_j]=(v_i|v_j), \quad [D\bar{v}_i{}_{\lambda} \bar{v}_j]=0.
  \end{equation}
  If we denote the $(n)$th product of $\bar{v}_i$ (resp., $D\bar{v}_i$) by $\bar{b}_i (n)$ (resp., $D\bar{b}_i (n)$), then the equality \eqref{eq: susy heisenberg Lambda bracket decomposition} gives the Lie bracket relations in Definition \ref{def: susy heisenberg}.
\end{remark}

\begin{definition} \label{def: susy fock rep}
  For a SUSY Heisenberg algebra $\HH^{\mathfrak{a}}$, consider its abelian Lie subalgebra
  \[\mathfrak{b}^{\mathfrak{a}}=\textup{Span}_{\CC}\{D\bar{b}_i(n), \bar{b}_i(n), \mathbbm{1}\,|\,n\geq 0,\ 1 \leq i \leq l \}.\]
  For $\alpha\in \mathfrak{a}^*$, let $\CC_{\alpha}:=\CC|\alpha \rangle$(resp., $\widetilde{\CC}_{\alpha}:=\CC|\alpha\rangle$) be the $\mathfrak{b}^{\mathfrak{a}}$-module such that
  \begin{itemize}
    \item $|\alpha\rangle$ is an even(resp., odd) vector,
    \item $D\bar{b}_i(m), \bar{b}_i(n)$ for $m>0, n\geq 0$ act trivially on $|\alpha \rangle$,
    \item $D\bar{b}_i (0)|\alpha \rangle=\alpha(v_i)|\alpha \rangle$
  \end{itemize}
  for any $i=1, \cdots, l$. Then, the \textit{even Fock representation} and \textit{odd Fock representation} of $\HH^{\mathfrak{a}}$ of highest weight $\alpha\in \mathfrak{a}^*$ are the induced modules
  \begin{equation*}
    \pi_{\alpha}^{\mathfrak{a}}=\textup{Ind}_{{\mathfrak{b}}^{\mathfrak{a}}}^{{\HH}^{\mathfrak{a}}}\ \CC_{\alpha}, \quad \widetilde{\pi}_{\alpha}^{\mathfrak{a}}=\textup{Ind}_{{\mathfrak{b}}^{\mathfrak{a}}}^{{\HH}^{\mathfrak{a}}}\ \widetilde{\CC}_{\alpha},
  \end{equation*}
  respectively. In particular, when $\alpha=0$, we usually omit the subscript for the even Fock representation so that $\pi^{\mathfrak{a}}_0={\pi}^{\mathfrak{a}}$.
\end{definition}

Note that the linear map ${\pi}^{\mathfrak{a}} \rightarrow V^{1}(\bar{\mathfrak{a}})$
\begin{equation} \label{eq: SUSY Heisenberg VA structure}
  \begin{aligned}
  D\bar{b}_{i_1}(-m_1) \cdots D\bar{b}_{i_r}(-m_r)&\bar{b}_{j_1}(-n_1)\cdots \bar{b}_{j_s}(-n_s)\vac\\
  \mapsto &(D\bar{v}_{i_1})_{(-m_1)}\cdots (D\bar{v}_{i_r})_{(-m_r)}(\bar{v}_{j_1})_{(-n_1)}\cdots(\bar{v}_{j_s})_{(-n_s)}\vac
  \end{aligned}
\end{equation}
gives a bijection between the two spaces. Via this map, one can pull the SUSY vertex algebra structure of $V^1(\bar{\mathfrak{a}})$ to the Fock representation ${\pi}^{\mathfrak{a}}$. Thus, we regard ${\pi}^{\mathfrak{a}}$ as a SUSY vertex algebra which we call the \textit{SUSY Heisenberg vertex algebra} associated with $\mathfrak{a}$. From Remark \ref{rmk: susy heisenberg enveloping algebra}, we already have that ${\pi}^{\mathfrak{a}}_{\alpha}$ and $\widetilde{\pi}^{\mathfrak{a}}_{\alpha}$ are $U({\pi}^{\mathfrak{a}})$-modules, where $U(V)$ is the vector superspace spanned by the $(n)$th coefficients of the elements of a vertex algebra $V$ for all $n\in \ZZ$. Using the well-known equivalence of the two module categories in Lemma \ref{lem: equivalence of module categories}, one can deduce that the Fock representations $\pi^{\mathfrak{a}}_{\alpha}$ and $\widetilde{\pi}^{\mathfrak{a}}_{\alpha}$ are actually a $\pi^{\mathfrak{a}}$-module. For the statement of Lemma \ref{lem: equivalence of module categories}, we introduce the notion of coherence for $U(V)$-modules. For a $U(V)$-module $M$, denote the action of $A_{(n)}$ on $M$ by $A_{[n]}\in \textup{End}_{\CC}(M)$ and write them as a formal power series
\begin{equation} \label{eq: module action power series}
  Y[A,z]=\sum_{n\in \ZZ} A_{[n]}z^{-n-1}\in \textup{End}_{\CC}(M)[\![z^{\pm 1}]\!]
\end{equation}
for each $A\in V$. We say that $M$ is \textit{coherent} if $A_{[n]}=0$ for $n>\!\!>0$, $Y[\vac,z]=\textup{Id}_M$ and
\begin{equation*}
  Y[\,:\!AB\!:\,,z]=\,:\!Y[A,z]Y[B,z]\!:
\end{equation*}
for any $A,B\in V$.

\begin{lemma}[\cite{FB04}] \label{lem: equivalence of module categories}
 For a vertex algebra $V$, let $U(V)$ be the superspace spanned by the $(n)$th coefficients of the elements of $V$ for all $n\in \ZZ$. Then, the category of $V$-modules is equivalent to the category of coherent $U(V)$-modules.
\end{lemma}
Note that the coherences of the Fock representations as $U(\pi^{\mathfrak{a}})$-modules directly follows from Definition \ref{def: susy fock rep}. Therefore, Lemma \ref{lem: equivalence of module categories} implies the following proposition.
\begin{proposition} \label{prop: fock representation vertex alg module}
  For each $\alpha\in \mathfrak{a}^*$, the Fock representations ${\pi}^{\mathfrak{a}}_{\alpha}$ and $\widetilde{\pi}^{\mathfrak{a}}_{\alpha}$ of ${\HH}^{\mathfrak{a}}$ is a ${\pi}^{\mathfrak{a}}$-module, considering ${\pi}^{\mathfrak{a}}$ as a vertex algebra. \qed
 \end{proposition}
 From now on, we call $\pi^{\mathfrak{a}}_{\alpha}$ and $\widetilde{\pi}^{\mathfrak{a}}_{\alpha}$ the Fock representations of $\pi^{\mathfrak{a}}$ of highest weight $\alpha\in \mathfrak{a}^*$. Inspired by \eqref{eq: module action power series} and the decomposition of $\Lambda$-bracket in \eqref{eq: SUSY nonSUSY lambda bracket relation}, denote
 \begin{equation} \label{eq: module action Lambda bracket}
  [v{}_{\Lambda}m]=\sum_{n \geq 0}\frac{\lambda^n}{n!}(Dv)_{[n]}m +\chi\sum_{n \geq 0}\frac{\lambda^n}{n!}v_{[n]}m
 \end{equation}
 for any $v\in V$ and $m\in M$, where $M$ is a $V$-module for a vertex algebra $V$ and $v_{[n]}$'s are the coefficients in \eqref{eq: module action power series}. Then, with use of \eqref{eq: module action Lambda bracket}, the action of $\pi^{\mathfrak{a}}$ on the Fock representations $\pi^{\mathfrak{a}}_{\alpha}$ or $\widetilde{\pi}^{\mathfrak{a}}_{\alpha}$ can be written as
 \begin{equation} \label{eq: lattice VA Lambda bracket}
  [\bar{v}_i{}_{\Lambda}\bar{v}_j]=(v_i|v_j)\chi, \quad [\bar{v}{}_{\Lambda}|\alpha \rangle]=\alpha (v) |\alpha \rangle
\end{equation}
 for any $v, v_i, v_j \in \mathfrak{a}$ and $\alpha\in \mathfrak{a}^*$.

\section{Supersymmetric W-algebras} \label{sec: SUSY W-algebras}
We recall the definition of supersymmetric W-algebras via the SUSY BRST cohomology. In this section, let $\g$ be a finite basic simple Lie superalgebra with an odd nilpotent $f\in \g.$ We assume that $f$ lies inside a subalgebra $\mathfrak{s}=\text{Span}_{\CC}\{E,e,H,f,F\}\simeq \mathfrak{osp}(1|2)$ of $\g$. Denote by $h^{\vee}$ the dual coxeter number of $\g$ and let $\kappa_{\g}$ be the killing form of a Lie superalgebra $\g$.

\subsection{SUSY BRST cohomology for supersymmetric W-algebras}\label{sec: SUSY BRST} \hfill\\
In the subspace $\mathfrak{s}=\text{Span}_{\CC}\{E,e,H,f,F\}$, the even elements $E,H,F$ form an $\mathfrak{sl}(2)$ triple and the odd elements $e,f$ satisfy the relations $[e,e]=2E$ and $[f,f]=-2F$. Suppose that $\g$ has a nondegenerate even supersymmetric invariant bilinear form $(\cdot|\cdot)$ normalized by $(E|F)=1.$ With respect to the $\text{ad}\frac{H}{2}$ action, $\g$ decomposes into the direct sum of eigenspaces given by
\begin{equation} \label{eq: g grading}
  \g=\bigoplus_{i\in \frac{\ZZ}{2}}\g_i, \quad \g_i=\{a\in \g\,|\,[\tfrac{H}{2},a]=i a\}.
\end{equation}
In particular, we have $f\in \g_{-\frac{1}{2}}$. Before we define the SUSY BRST complex for a SUSY W-algebra, we introduce the SUSY charged free fermion vertex algebra. Let $\mathfrak{n}$ and $\mathfrak{n}_-$ be the two vector subspaces given by
\begin{equation} \label{eq: n and n-}
  \n=\bigoplus_{i>0} \g_i, \quad \n_-=\bigoplus_{i<0} \g_i.
\end{equation}
Take a basis $\{u_{\alpha}\}_{\alpha\in I_+}$ of $\n$, where each $u_{\alpha}$ is a root vector for some positive root $\alpha$. One can choose such basis by the result in \cite{Hoyt12}, which states that any simple root vector of $\g$ has an $\textup{ad}\frac{H}{2}$-eigenvalue $0, \frac{1}{2}$ or $1$. Take the dual basis $\{u^{\alpha}\}_{\alpha\in I_+}$ of $\n_-$ satisfying $(u^{\alpha}\,|\,u_{\beta})=\delta_{\alpha, \beta}$. By the invariance of the bilinear form, it follows that each $u^{\alpha}$ is a root vector for $-\alpha$. Consider the superspaces
\begin{equation} \label{eq: phi n space}
  \phi_{\bar{\n}}\simeq \bar{\n}, \quad \phi^{\n_-}\simeq \n_-,
\end{equation}
where the bar notation means the parity-reversing. Define a skew-supersymmetric bilinear form $\langle \cdot | \cdot \rangle$ on $\mathfrak{A}:=\phi_{\bar{\n}}\oplus \phi^{\n_-}$ by
\begin{equation*}
  \langle \phi^a | \phi_{\bar{b}} \rangle:=(a|b)=:\langle \phi_{\bar{b}}|\phi^{a} \rangle,
\end{equation*}
and $\langle \phi_{\bar{\n}}|\phi_{\bar{\n}} \rangle=\langle \phi^{\n_-}|\phi^{\n_-} \rangle=0$.
Then the SUSY fermionic vertex algebra $F(\bar{\mathfrak{A}})$ associated with $\mathfrak{A}$ is called the \textit{SUSY charged free fermion vertex algebra} associated with $\g$. Recall Example \ref{ex: susy fermionic} for the definition of SUSY fermionic vertex algebras.

For $k\in \CC$, the SUSY BRST complex for $\g$ and $\mathfrak{s}$ is the SUSY vertex algebra
\begin{equation} \label{eq: BRST complex}
  C^k(\bar{\g},f):=V^{k}(\bar{\g})\otimes F(\bar{\mathfrak{A}}),
\end{equation}
where $V^k(\bar{\g})$ is the SUSY affine vertex algebra of level $k$ in Example \ref{ex: susy affine va}, and $F(\bar{\mathfrak{A}})$ is the SUSY charged free fermion vertex algebra associated with $\g$. Consider the following two even elements in $C^k(\bar{\g},f)$,
\begin{align}
  d_{\text{st}}&=\sum_{\alpha\in I_+} :\bar{u}_{\alpha} \phi^{\alpha}:+\frac{1}{2}\sum_{\alpha,\beta\in I_+}(-1)^{p(\alpha)(p(\beta)+1)}:\phi_{[u_{\alpha},u_{\beta}]}\phi^{\beta}\phi^{\alpha}:, \label{eq: d_st def}\\
  d_f&=-\sum_{\alpha\in I_+} (f|u_{\alpha})\phi^{\alpha}, \label{eq: d_f def}
\end{align}
where $\phi^{\alpha}=\phi^{\bar{u}^{\alpha}}$, $\phi_{\alpha}=\phi_{u_{\alpha}}$ and $p(\alpha)=p(u_{\alpha})$ for each $\alpha\in I_+$. For $d=d_{\text{st}}+d_f$, it follows from \cite{MRS21} that the odd endomorphism $d_{(0|0)}$ of $C^k(\bar{\g},f)$ is an odd derivation whose square vanishes. Thus, one can define the corresponding cohomology using this map.
\begin{definition}[\cite{MRS21}] \label{def: susy W-alg}
  Let $\g$ be a finite basic simple Lie superalgebra with an odd nilpotent $f$ that lies in a subalgebra isomorphic to $\mathfrak{osp}(1|2)$.
  The \textit{SUSY W-algebra associated with $\g$ and $f$ of level $k\in \CC$} is defined by
  \begin{equation*}
    W^{k}(\bar{\g},f):=H(C^{k}(\bar{\g},f), d_{(0|0)})
  \end{equation*}
  for the SUSY BRST complex $C^k(\bar{\g},f)$ and $d=d_{\text{st}}+d_f$.
\end{definition}
\subsection{Building blocks for SUSY W-algebras} \label{sec: building blocks for susy w-alg} \hfill \\
The SUSY BRST complex $C^k(\bar{\g}, f)$ has a subalgebra isomorphic to the SUSY affine vertex algebra other than $V^k(\bar{\g})$ constituting the SUSY BRST complex. It eventually leads to the simpler realization of SUSY W-algebra in Theorem \ref{thm: susy W-alg charge 0}. To observe this, we present the following ``building block'' $J_{\bar{a}}$ for $\bar{a}\in \bar{\g}$.
\begin{equation} \label{eq: building block J_a}
  J_{\bar{a}}=\bar{a}+\sum_{\beta,\gamma\in I_+}(-1)^{(p(a)+1)(p(\beta)+1)}(u^{\gamma}|[u_{\beta},a]):\phi^{\beta} \phi_{\gamma}:\in C^k(\bar{\g},f).
\end{equation}
Let $\widetilde{C}^k(\bar{\g},f)$ be the SUSY vertex subalgebra of the SUSY BRST complex $C^k(\bar{\g},f)$ generated by the subspace $J_{\bar{\g}_{\leq 0}}\oplus \phi^{\bar{\n}_-}$. Here, $J_{\bar{\g}_{\leq 0}}=\{J_{\bar{a}}\,|\,a\in \g_{\leq 0}\}$ for $\g_{\leq 0}=\oplus_{i\leq 0}\g_i$ and $\phi^{\bar{\n}_-}$ is the parity reversed space of $\phi^{\n_-}$ inside the SUSY charged free fermion vertex algebra.
\begin{proposition}[\cite{MRS21}] \label{prop: C^k elements relation}
  Let $\widetilde{C}^k(\bar{\g},f)$ be the SUSY vertex subalgebra of $C^k(\bar{\g},f)$ as above. Then, the generators of $\widetilde{C}^k(\bar{\g},f)$ satisfy the following $\Lambda$-bracket relation:
  \begin{align*}
    [J_{\bar{a}}{}_{\Lambda}J_{\bar{b}}]&=(-1)^{p(a)(p(b)+1)}J_{\overline{[a,b]}}+(k+h^{\vee})(D+\chi)(a|b),\\
    [\phi^{\alpha}{}_{\Lambda}J_{\bar{a}}]&=\sum_{\beta\in I_+}(-1)^{p(\alpha)+1}([a,u^{\alpha}]|u_{\beta})\phi^{\beta}
  \end{align*}
  for any $\alpha\in I_+$ and $a,b\in \g_{\leq 0}$.
\end{proposition}

Note from Proposition \ref{prop: C^k elements relation} that the building blocks $J_{\bar{a}}$'s for $a\in \g_{\leq 0}$ satisfy the $\Lambda$-bracket relation in Example \ref{ex: susy affine va generalized}, if we replace the bilinear form $\tau_k$ by
\begin{equation} \label{eq: susy affine level shift}
\tau_k+\frac{1}{2}\kappa_{\g}-\frac{1}{2}\kappa_{\g_{\leq 0}}.
\end{equation}
Therefore, we denote the vertex subalgebra of $\widetilde{C}^k(\bar{\g}, f)$ generated by $J_{\bar{a}}$'s for $a\in \g_{\leq 0}$ by $V^{\zeta_k}(\bar{\g}_{\leq 0})$, where $\zeta_k:=k+\frac{1}{2}\kappa_{\g}-\frac{1}{2}\kappa_{\g_{\leq 0}}$.

\begin{proposition} \label{prop: d, C^k relation}
  Let $d_{\text{st}}$ and $d_f$ be the elements of $C^k(\bar{\g}, f)$ in \eqref{eq: d_st def} and \eqref{eq: d_f def}. Then
  \begin{align*}
    &[d_{\textup{st}}{}_{\Lambda}J_{\bar{a}}]=\sum_{\beta\in I_+}(-1)^{(p(a)+1)p(\beta)}:\phi^{\beta} J_{\overline{\pi_{\leq 0}[u_{\beta},a]}}:-\sum_{\beta\in I_+}(-1)^{p(\beta)}(k+h^{\vee})(D+\chi)(u_{\beta}|a)\phi^{\beta},\\
    &[d_{\textup{st}}{}_{\Lambda}\phi^{\alpha}]=\frac{1}{2} \sum_{\beta,\gamma\in I_+}(-1)^{(p(\alpha)+1)p(\beta)}([u_{\beta},u^{\alpha}]|u_{\gamma}):\phi^{\beta}\phi^{\gamma}:,\\
    &[d_f{}_{\Lambda}J_{\bar{a}}]=\sum_{\beta\in I_+}(-1)^{(p(a)+1)p(\beta)}(f|[u_{\beta},a])\phi^{\beta},\\
    &[d_f{}_{\Lambda}\phi^{\alpha}]=0,
  \end{align*}
  where $\pi_{\leq 0}: \g\rightarrow \g_{\leq 0}$ is the projection map.
\end{proposition}
\begin{proof}
  The equalities follow from the direct computation.
\end{proof}
The complex $\widetilde{C}^k(\bar{\g},f)$ has a $\ZZ_{\geq 0}$-grading called \textit{charge}, which satisfies
\begin{equation} \label{eq: charge}
  \text{ch}(J_{\bar{a}})=0, \quad \text{ch}(\phi^{\alpha})=1,\quad \text{ch}(DA)=\text{ch}(A), \quad \text{ch}(:\!AB\!:)=\text{ch}(A)+\text{ch}(B)
\end{equation}
for any $a\in \g_{\leq 0}$, $\alpha\in I_+$ and $A,B\in \widetilde{C}^k(\bar{\g},f)$. In particular, the charge $0$ subspace of $\widetilde{C}^k(\bar{\g},f)$ is the SUSY affine vertex algebra $V^{\zeta_k}(\bar{\g}_{\leq 0})$ generated by the building blocks. Furthermore, the complex $\widetilde{C}^k(\bar{\g},f)$ is closed under the differential $d_{(0|0)}$ which increases the charge by $1$. Thanks to Proposition 4.10 of \cite{MRS21}, we have that the cohomology with respect to the complex $\widetilde{C}^k(\bar{\g},f)$ and the differential $d_{(0|0)}$ is isomorphic to the SUSY W-algebra.
\begin{theorem}[\cite{MRS21}] \label{thm: susy W-alg charge 0}
  The SUSY vertex subalgebra $\widetilde{C}^k(\bar{\g},f)$ of $C^k(\bar{\g},f)$ is closed under the differential $d_{(0|0)}$ and the corresponding cohomology $H(\widetilde{C}^k(\bar{\g},f),d_{(0|0)})$ is concentrated on the charge $0$ space. Furthermore,
  \begin{equation} \label{eq: SUSY W-alg can be described by building blocks}
    W^k(\bar{\g},f)=H^0(\widetilde{C}^k(\bar{\g},f),d_{(0|0)})\subset V^{\zeta_k}(\bar{\g}_{\leq 0}),
  \end{equation}
  where $H^0$ is the cohomology at charge $0$.
\end{theorem}


\subsection{Properties of supersymmetric W-algebras} \label{sec: properties of susy W-alg} \hfill \\
Recall that W-algebras have a conformal vector when the level is noncritical. The grading on W-algebras given by the conformal vector, called conformal weight, plays a crucial role in analyzing W-algebras. Similarly, one can expect that SUSY W-algebras have a superconformal vector and the grading given by this vector would be crucial. We prove in Theorem \ref{thm: superconformal for susy W-alg} that SUSY W-algebra of a noncritical level has a superconformal vector, and see in Theorem \ref{thm: free generators of W-algebra} that SUSY W-algebra has free generators of homogeneous conformal weight. Refer to Definition \ref{def: superconformal vector} for a superconformal vector. 

Recall that the SUSY BRST complex \eqref{eq: BRST complex} consists of a SUSY affine vertex algebra and a SUSY charged free fermion vertex algebra. We first prove in Proposition \ref{prop: SUSY charged free fermion superconformal} that the SUSY charged free fermion vertex algebra has a superconformal vector.
\begin{proposition} \label{prop: SUSY charged free fermion superconformal}
  Let $F(\bar{\mathfrak{A}})$ be the SUSY charged free fermion vertex algebra associated with $\g$. For a sequence $(m_{\alpha})_{\alpha\in I_+}$ of complex numbers, the element
  \begin{equation*}
    \tau =\sum_{\alpha\in I_+} (-1)^{p(\alpha)}2m_{\alpha}:(\partial\phi_{\alpha})\phi^{\alpha}:-\sum_{\alpha\in I_+} (-1)^{p(\alpha)}(1-2m_{\alpha}):\phi_{\alpha}\partial \phi^{\alpha}:+\sum_{\alpha\in I_+}:(D\phi_{\alpha})(D\phi^{\alpha}):
  \end{equation*}
  in $F(\bar{\mathfrak{A}})$ is a superconformal vector of central charge $12\sum_{\alpha\in I_+}(-1)^{p(\alpha)}m_{\alpha}-3\,\textup{sdim}(\n)$, where $\textup{sdim}(\n)$ is the superdimension of $\n$. Furthermore, $\phi^{\alpha}$(resp., $\phi_{\alpha}$) is primary of conformal weight $m_{\alpha}$(resp., $\frac{1}{2}-m_{\alpha}$). \qed
\end{proposition}
\begin{proof}
By direct calculations, we get
\begin{equation} \label{eq: SUSY charged free superconformal proof}
  \begin{aligned}
  [\tau{}_{\Lambda}\phi^{\alpha}]&=(2\partial+2m_{\alpha}+\chi D)\phi^{\alpha}, \\
  [\tau{}_{\Lambda}\phi_{\alpha}]&=(2\partial+(1-2m_{\alpha})+\chi D)\phi_{\alpha}
  \end{aligned}
\end{equation}
for each $\alpha\in I_+$. Using the relations in \eqref{eq: SUSY charged free superconformal proof}, one can show that
\begin{equation*}
  [\tau{}_{\Lambda}\tau]=(2\partial+3\lambda+\chi D)\tau+\chi \lambda^2\Big(4\sum_{\alpha \in I_+}(-1)^{p(\alpha)}m_{\alpha}-\textup{sdim}(\n) \Big).
\end{equation*}
Thus, the statement follows.
\end{proof}
Recall the element $H$ of $\mathfrak{s}\simeq \mathfrak{osp}(1|2)$, which gives the eigenspace decomposition \eqref{eq: g grading}. 
By shifting the sum of $\omega$ in Example \ref{ex: susy affine va} and $\tau$ in Proposition \ref{prop: SUSY charged free fermion superconformal}  by $\partial \bar{H}$, one can get the superconformal vector of the SUSY W-algebra.
\begin{theorem} \label{thm: superconformal for susy W-alg}
  For $k\neq -h^{\vee}$, the SUSY W-algebra $W^k(\bar{\g},f)$ is superconformal. In particular, for $\omega$ in Example \ref{ex: susy affine va} and $\tau$ in Proposition \ref{prop: SUSY charged free fermion superconformal},
  \begin{equation} \label{eq: superconformal vector for susy W-alg}
    G:=\omega+\tau+\partial \bar{H}
  \end{equation}
  is a superconformal vector of $W^k(\bar{\g}, f)$ of central charge 
  \begin{equation} \label{eq: superconformal susy W-alg central charge}
    \frac{k\, \textup{sdim}{\g}}{k+h^{\vee}}+\frac{1}{2}\textup{sdim}\g+12\sum_{\alpha\in I_+}(-1)^{p(\alpha)}m_{\alpha}-3 \textup{sdim}(\n)-\frac{2}{3}(k+h^{\vee}).
  \end{equation}
\end{theorem}
\begin{proof}
 Since we already know that $\w$ and $\tau$ are superconformal vectors of $V^k(\bar{\g})$ and $F(\bar{\mathfrak{A}})$, respectively, one can check by direct computation that $G=\w+\tau +\partial\bar{H}$ is a superconformal vector of $C^k(\bar{\g},f)$ of central charge \eqref{eq: superconformal susy W-alg central charge}. In particular, we have
 \begin{equation} \label{eq: Ja conformal weight}
  [G{}_{\Lambda}J_{\bar{a}}]=\left(2\partial+2\left(\frac{1}{2}-j_a\right)\lambda+\chi D\right)J_{\bar{a}}-(k+h^{\vee})(e|[f,a])\lambda\chi
 \end{equation}
 for any $a\in \g_{j_a}$. It means that each $J_{\bar{a}}$ has a conformal weight $\frac{1}{2}-j_a$. Hence, it is enough to show that $G$ belongs to the SUSY W-algebra, i.e., $G$ is killed by the differential $d_{(0|0)}$. Recall the facts that each $\bar{a}\in \bar{\g}$ is primary of conformal weight $\frac{1}{2}$ with respect to $\omega$ and $\phi^{\alpha}$(resp., $\phi_{\alpha}$) is primary of conformal weight $m_{\alpha}$(resp., $\frac{1}{2}-m_{\alpha}$) with respect to $\tau$. The direct computation shows that
 \begin{equation*}
  \begin{aligned}
  [\omega+\tau {}_{\Lambda} d]&=(2\partial+\chi D)d+\lambda\Big(\sum_{\alpha\in I_+}(2m_{\alpha}+1):\bar{u}_{\alpha}\phi^{\alpha}:-\sum_{\alpha\in I_+}2m_{\alpha}(f|u_{\alpha})\phi^{\alpha}\Big)\\
  &+\frac{1}{2}\lambda\sum_{\alpha, \beta, \gamma\in I_+}(-1)^{p(\alpha)(p(\beta)+1)}(u^\gamma|[u_{\alpha},u_{\beta}]):\!\phi_{\gamma}\phi^{\beta}\phi^{\alpha}\!:.
  \end{aligned} 
\end{equation*}
Hence, using the super skew-symmetry, we have
\begin{equation} \label{eq: superconformal computation process}
  d_{(0|0)}(\omega+\tau)=-2\partial \Big(\sum_{\alpha\in I_+}m_{\alpha} :\bar{u}_{\alpha}\phi^{\alpha}:\Big).
\end{equation}
Similarly, one can compute that
\begin{equation*}
  [\partial \bar{H}{}_{\Lambda}d]=-2\lambda \sum_{\alpha\in I_+} m_{\alpha} :\bar{u}_{\alpha} \phi^{\alpha}:,
\end{equation*}
and the super skew-symmetry shows us that $d_{(0|0)} \partial \bar{H}$ is equal to the opposite sign of \eqref{eq: superconformal computation process}.
\end{proof}
\begin{remark}
  There is a superconformal vector $\tau^{\textup{ch}}$ of SUSY charged free fermion vertex algebra in \cite{SY23}, whose formula is different from $\tau$ in Proposition \ref{prop: SUSY charged free fermion superconformal}. It is also true that $\phi^{\alpha}$ and $\phi_{\alpha}$'s are primary with respect to $\tau^{\textup{ch}}$. Since we only used the properties of $\tau$, not the formulation of $\tau$ in the proof of Theorem \ref{thm: superconformal for susy W-alg}, one can show that the statement in Theorem \ref{thm: superconformal for susy W-alg} holds even though we replace $\tau$ with another superconformal vector $\tau^{\textup{ch}}$.
\end{remark}
Motivated from Theorem \ref{thm: superconformal for susy W-alg}, we say that the level $k$ of the SUSY W-algebra is \textit{critical} if $k=-h^{\vee}$. Note from the proof of Theorem \ref{thm: superconformal for susy W-alg} that $G$ in \eqref{eq: superconformal vector for susy W-alg} is a superconformal vector for both $C^k(\bar{\g},f)$ and $W^k(\bar{\g},f)$. Therefore, in the rest of the paper, the conformal weight of any element of $C^k(\bar{\g},f)$ or $W^k(\bar{\g},f)$ means the conformal weight given by $G$ when the level is noncritical.

\begin{theorem}[\cite{MRS21}] \label{thm: free generators of W-algebra}
 Take any basis $\{a_1, \cdots ,a_m\}$ of $\ker(\textup{ad}f)\subset \g$. For a noncritical level $k$, there is a set of elements $S=\{J_{\bar{a}_i}+A_i\}_{i=1}^m$ in $W^k(\bar{\g},f)$ such that
  \begin{enumerate}[(i)]
  \item $W^k(\bar{\g},f)$ is freely generated by $S$,
  \item $A_i\in \widetilde{C}^k(\bar{\g},f)$ has charge $0$ and conformal weight $\frac{1}{2}-j_{a_i}$, where $a_i\in \g_{j_{a_i}}$,
  \item the linear term of $A_i$ is a total derivative with respect to $D$.
 \end{enumerate}
\end{theorem}

We note here that each element $J_{\bar{a}_i}+A_i$ in Theorem \ref{thm: free generators of W-algebra} is homogeneous of conformal weight $\frac{1}{2}-j_{a_i}$. This is due to the equality \eqref{eq: Ja conformal weight}.

\begin{remark}
  In the original paper \cite{MRS21}, the theorem is stated for any $k\in \CC$. A cohomological approach is used to accomplish this by setting the degree of $J_{\bar{a}}$ by $\frac{1}{2}-j_a$, without requiring the existence of a superconformal vector. Note that this degree coincides with the conformal weight \eqref{eq: Ja conformal weight} for noncritical levels.
\end{remark}

\subsection{Classical SUSY W-algebra} \label{sec: classical SUSY W-algebra} \hfill\\
In this section, we briefly study the quasi-classical limits of SUSY W-algebras, which give the description of SUSY classical W-algebras via SUSY BRST complexes. We use all the notations introduced in Section \ref{sec: SUSY BRST} and \ref{sec: building blocks for susy w-alg} and refer to Proposition \ref{prop: quasi-classical limit} for the quasi-classical limit of SUSY vertex algebras.

For an even formal variable $y$, let $\widehat{C}_y(\bar{\g},f)$ be the SUSY vertex algebra over $\CC[y]$ generated by $\hat{J}_{\bar{a}}$, $a\in \g_{\leq 0}$ and $\hat{\phi}^{\alpha}$, $\alpha\in I_+$ with the following $\Lambda$-bracket relations
\begin{equation} \label{eq: classical relations}
  \begin{aligned}
 [\hat{J}_{\bar{a}}{}_{\Lambda} \hat{J}_{\bar{b}}]_y&=(-1)^{p(a)(p(b)+1)}y \hat{J}_{\overline{[a,b]}}+y (D+\chi)(a|b),\\
 [\hat{\phi}^{\alpha}{}_{\Lambda}\hat{J}_{\bar{a}}]_y&=y \sum_{\beta\in I_+}(-1)^{p(\alpha)+1}([a,u^{\alpha}]|u_{\beta})\phi^{\beta},\\
 [\hat{\phi}^{\alpha}{}_{\Lambda} \hat{\phi}^{\beta}]_y&=0.
  \end{aligned}
\end{equation}
Suppose $\widehat{C}_y(\bar{\g},f)$ is equipped with the odd differentials ${d_{\textup{st}}}_{(0|0)}$ and ${d_f}_{(0|0)}$ defined by
\begin{equation} \label{eq: classical SUSY differential}
\begin{aligned}
    [d_{\textup{st}}{}_{\Lambda}\hat{J}_{\bar{a}}]_y=&\sum_{\beta\in I_+}(-1)^{(p(a)+1)p(\beta)}:\hat{\phi}^{\beta} \hat{J}_{\overline{\pi_{\leq 0}[u_{\beta},a]}}:-\sum_{\beta\in I_+}(-1)^{p(\beta)}(D+\chi)(u_{\beta}|a)\hat{\phi}^{\beta},\\
    [d_{\textup{st}}{}_{\Lambda}\hat{\phi}^{\alpha}]_y=&\frac{1}{2} \sum_{\beta,\gamma\in I_+}(-1)^{(p(\alpha)+1)p(\beta)}([u_{\beta},u^{\alpha}]|u_{\gamma}):\hat{\phi}^{\beta}\hat{\phi}^{\gamma}:,\\
    [d_f{}_{\Lambda}\hat{J}_{\bar{a}}]_y=&\sum_{\beta\in I_+}(-1)^{(p(a)+1)p(\beta)}(f|[u_{\beta},a])\hat{\phi}^{\beta},\quad   [d_f{}_{\Lambda}\hat{\phi}^{\alpha}]_y=0.
\end{aligned}
\end{equation}
The quasi-classical limit of this family of SUSY vertex algebras with differentials can be described with the evaluation of $y$ at $0$, as can be seen in Proposition \ref{prop: quasi-classical limit}.
For $\epsilon\in \CC$, let the evaluation of $\widehat{C}_y(\bar{\g},f)$ at $y=\epsilon$ be the space over $\CC$ defined by
\begin{equation} \label{eq: C epsilon}
  \widehat{C}_{\epsilon}:=\widehat{C}_y(\bar{\g},f)\tens{\CC[y]}\CC_{\epsilon},
\end{equation}
where $\CC_{\epsilon}$ is a $1$-dimensional $\CC[y]$-module with $y$ acting as $\epsilon$. Note that the evaluation at $y=\frac{1}{k+h^{\vee}}$ recovers the original SUSY BRST complex for noncritical $k$. To be explicit,
\begin{equation} \label{eq: C 1/k}
  \widehat{C}_{\frac{1}{k+h^{\vee}}}\simeq \widetilde{C}^k(\bar{\g},f),
\end{equation}
if we take $\hat{J}_{\bar{a}}=\frac{1}{k+h^{\vee}} J_{\bar{a}}$ and $\hat{\phi}^{\alpha}=\phi_{\alpha}$ for $a\in \g_{\leq 0}$ and $\alpha \in I_+$. Now, the quasi-classical limit of the family results in the supersymmetric SUSY vertex algebra
\begin{equation} \label{eq: classical SUSY complex}
  \widetilde{C}^{\infty}:=\widehat{C}_0,
\end{equation}
which is also a SUSY Poisson vertex algebra with a Poisson $\Lambda$-bracket $\{A{}_{\Lambda}B\}=\frac{1}{y}[A{}_{\Lambda}B]_y \vert_{y=0}$. Let $d_{(0|0)}$ be the induced differential on $\widetilde{C}^{\infty}$ given by the sum of ${d_{\textup{st}}}_{(0|0)}$ and ${d_f}_{(0|0)}$. Then, the cohomology with this differential defines the classical SUSY W-algebra.
\begin{definition} \label{def: classical SUSY W-algebra}
  Let $\g$ be a finite basic simple Lie superalgebra with an odd nilpotent $f$ that lies in a subalgebra isomorphic to $\mathfrak{osp}(1|2)$. The \textit{classical SUSY W-algebra associated with $\g$ and $f$} is the cohomology
 \begin{equation*}
  W^{\infty}(\bar{\g}, f):=H(\widetilde{C}^{\infty},d_{(0|0)}),
 \end{equation*}
 which has a Poisson vertex algebra structure inherited from $\widetilde{C}^{\infty}$.
\end{definition}
\begin{remark}
  The classical SUSY W-algebras are often defined as the invariant subspace of the SUSY affine Poisson vertex algebra under the gauge action. In \cite{Suh20}, it is shown that this definition is equivalent to Definition \ref{def: classical SUSY W-algebra}.
\end{remark}

\section{Free field realization of supersymmetric W-algebras}\label{bigsec: FF realization of SUSY W-algebras}
In this section, we inherit the notations in Section \ref{sec: SUSY W-algebras} to show the free field realization of SUSY W-algebra as a kernel of screening operators. The statements in this section that involve a complex parameter $k\in \CC$ hold generically, even if it is not mentioned.


\subsection{Notations for Lie superalgebras} \label{sec: Notations for Lie superalgebras} \hfill \\
For later use, we briefly introduce the notations that are used throughout this section. Let $\h\subset \g$ be the Cartan subalgebra, and $\Phi$, $\Phi_+$ and $\Pi$ be the set of roots, positive roots, and simple roots of $\g$, respectively. Also, denote
\begin{equation} \label{eq: Phi_0}
  \Phi_0:= \{\alpha\in \Phi\,|\, \text{a root vector of }\alpha\text{ is in }\g_0\},
\end{equation}
and the corresponding root lattice is denoted by $Q_0:=\sum_{\alpha\in \Phi_0}\ZZ \alpha$. Recall the index set $I_+$ in Section \ref{sec: SUSY BRST} parametrizing the basis $\{u_{\alpha}\}_{\alpha\in I_+}$ of $\mathfrak{n}$. Since we assume that each $u_{\alpha}$ is a root vector of some positive root $\alpha$, we have $I_+\subset \Phi_+$. For $\alpha \in I_+$, we say that $\alpha$ is indecomposable in $I_+$ if it cannot be expressed as a sum of two elements of $I_+$. Otherwise, we say it is decomposable. Denote
\begin{equation} \label{eq: set of indecomposables}
  I_0:=\{\alpha\in I_+\mid \alpha \text{ is indecomposable in }I_+\}
\end{equation}
and define the equivalence relation $\sim$ on $I_0$ by
\begin{equation} \label{eq: equivalence relation}
  \alpha \sim \beta \Leftrightarrow \alpha-\beta\in Q_0.
\end{equation}
For $\alpha\in I_0$, let $[\alpha]$ be the equivalence class in $I_0$ containing $\alpha$, and $[I_0]$ be the set of equivalence classes. Due to the Lie superalgebra theory \cite{Musson12}, the Lie bracket of a root vector of $\alpha$ and $\beta$ is a root vector of $\alpha+\beta$ whenever $\alpha+\beta$ is a root of $\g$. Hence, the conditions in \eqref{eq: set of indecomposables} and \eqref{eq: equivalence relation} can be understood as the Lie brackets between the root vectors $u_{\alpha}$'s and $u^{\alpha}$'s.
 \begin{lemma} \label{lem: conditions with Lie brackets}
 Recall the dual bases $\{u_{\alpha}\}_{\alpha\in I_+}$ and $\{u^{\alpha}\}_{\alpha\in I_+}$ introduced in Section \ref{sec: SUSY BRST}.
  \begin{enumerate}
    \item A positive root $\alpha\in I_+$ is indecomposable if and only if $(u^{\alpha}|[u_{\beta},u_{\gamma}])=0$ for any $\beta, \gamma\in I_+$.
    \item  For $\alpha, \beta\in I_0$, $\alpha\sim \beta$ if and only if there exists $a\in \g_0$ such that $(a|[u^{\alpha},u_{\beta}])\neq 0$.
  \end{enumerate}
 \end{lemma}
 \begin{proof}
  For $\beta, \gamma\in I_+$, their sum $\beta+\gamma$ is also an element of $I_+$ if it is a root. Hence, $[u_{\beta},u_{\gamma}]=c\, u_{\alpha+\beta}$ for nonzero constant $c$. It shows the statement (1) of the lemma. Similarly, if $\alpha-\beta$ is a root, $[u^{\alpha}, u_{\beta}]$ is a root vector of $-\alpha+\beta$, which proves the statement (2).
 \end{proof}

  \subsection{Weight filtration} \label{sec: weight filtration} \hfill \\
Recall from Theorem \ref{thm: susy W-alg charge 0} that the SUSY W-algebra $W^k(\bar{\g},f)$ is a subspace of the charge $0$ space of $\widetilde{C}^k(\bar{\g},f)$. For the sake of  simplicity, we denote $\widetilde{C}^k=\widetilde{C}^k(\bar{\g},f)$ in this section. In addition to the charge degree, $\widetilde{C}^k$ has a $\ZZ_{\geq 0}$-grading called the \textit{weight} given by
\begin{equation} \label{eq: weight}
\begin{aligned}
  &\textup{wt}\, J_{\bar{a}}=-2j_a\quad \text{if } a\in \g_{j_a}\ (j_a\leq 0),\\
  &\textup{wt}\, \phi^{\alpha}=2j_{\alpha}\quad \text{if } u_{\alpha}\in \g_{j_{\alpha}}\ (j_{\alpha}>0),\\
  &\textup{wt}\, \vac=0, \quad \textup{wt}\, (D A)=\textup{wt}(A),\\
  &\textup{wt}\, (:\!AB\!:)=\textup{wt}\, (A)+\textup{wt}\, (B)
\end{aligned}
\end{equation}
for any $A,B\in \widetilde{C}^k$. By Proposition \ref{prop: d, C^k relation}, one can immediately see that $d_{\textup{st}}$ preserves the weight and $d_f$ raises the weight by $2$. Therefore, $d_{(0|0)}(F_p \widetilde{C}^k)\subset F_p\widetilde{C}^k$ for any $p\in \ZZ_{\geq 0}$ if we consider the filtration 
\begin{equation*}
  F_p \widetilde{C}^k=\{A\in \widetilde{C}^k\,|\,\textup{deg}\, A \geq p\}.
\end{equation*}
Hence, with use of the differential $d_{(0|0)}$, the filtered space $F_{\bullet}\widetilde{C}^k$ gives the spectral sequence, whose total complex $\{E_n^k\}_{n\geq 0}$ converges to $H(\widetilde{C}^k, d_{(0|0)})=W^k(\bar{\g},f)$. Denote the differential of the spectral sequence in the $n$th page by $d_n$. Observing the weight change by $d_{\textup{st}}$ and $d_f$, one can see that
\begin{equation} \label{eq: second page differential}
  d_0={{d_{\textup{st}}}}_{(0|0)},\quad d_2={{d_f}}_{(0|0)}
\end{equation}
and $d_n=0$ for all other $n$'s. This implies $E_1^k=E_2^k$, and
\begin{equation} \label{eq: SUSY W-alg total complex isomorphism}
  W^k(\bar{\g},f)=H(\widetilde{C}^k, d_{(0|0)})\simeq H(E_2^k, d_2).
\end{equation}
The description \eqref{eq: SUSY W-alg total complex isomorphism} of SUSY W-algebras eventually gives the free field realization.
In particular, the structure of $E_2^k=E_1^k$ is given by SUSY Heisenberg vertex algebra, and the action of $d_2$ on the space $E^k_1$ gives the screening operators. In the following sections, we give these investigations one by one.


\subsection{The first total complex} \label{sec: the first total complex}\hfill \\
In this section, we observe the structure of $E_1^k=H(\widetilde{C}^k,{d_{\textup{st}}}_{(0|0)})$. The key idea is that one can analyze the structure of $E^k_1$ in a classical setting. Recall the space $\widetilde{C}^{\infty}$ in Section \ref{sec: classical SUSY W-algebra} obtained by the quasi-classical limit of $\widetilde{C}^k$. For the differential ${d_{\textup{st}}}_{(0|0)}$ on $\widetilde{C}^{\infty}$, set
\begin{equation} \label{eq: E_1^infty}
  E_1^{\infty}=H(\widetilde{C}^{\infty},{d_{\textup{st}}}_{(0|0)}).
\end{equation}
Note that the charge in \eqref{eq: charge} induces a $\ZZ_{\geq 0}$-grading on $\widetilde{C}^{\infty}$ which we also call the charge.

\begin{lemma} \label{lem: Zariski dense}
  The set
  \begin{equation*}
    S:=\big\{k\in \CC \mid H(\widetilde{C}^k,{d_{\textup{st}}}_{(0|0)})\simeq H(\widetilde{C}^{\infty},{d_{\textup{st}}}_{(0|0)})\big\}
  \end{equation*}
  is Zariski dense in $\CC$, where $(\widetilde{C}^{\infty}, {d_{\textup{st}}}_{(0|0)})$ is the chain complex introduced in Section \ref{sec: classical SUSY W-algebra}. 
\end{lemma}
\begin{proof}
  Let $\epsilon=\frac{1}{k+h^{\vee}}$ and recall the complex $\widehat{C}_y$ and $\CC[y]$-module $\CC_{\epsilon}$ in Section \ref{sec: classical SUSY W-algebra}. From the algebraic geometry theory, we have the set
  \begin{equation} \label{eq: set S tild}
    \widetilde{S}:=\big\{k\in \CC\mid H\big(\widehat{C}_y\tens{\CC[y]}\CC_{\epsilon},{d_{\textup{st}}}_{(0|0)}\big) \simeq H\big(\widehat{C}_y,{d_{\textup{st}}}_{(0|0)}\big)\tens{\CC[y]}\CC_{\epsilon}\big\},
  \end{equation}
  which is Zariski dense in $\CC$. We prove the lemma by showing that $\widetilde{S}$ is contained in $S$. From \eqref{eq: C 1/k}, we know that the LHS of the condition in \eqref{eq: set S tild} is equal to $H(\widetilde{C}^k, {d_{\textup{st}}}_{(0|0)})$. Hence, we need to show that
  \begin{equation} \label{eq: cohomology proof goal}
    H\big(\widehat{C}_y,{d_{\textup{st}}}_{(0|0)}\big)\tens{\CC[y]}\CC_{\epsilon} \simeq H(\widetilde{C}^{\infty}, {d_{\textup{st}}}_{(0|0)}).
  \end{equation}
  Consider the filtration on $\widehat{C}_y$ given by $F_p \widehat{C}_y=y^p\widehat{C}_y$ for nonnegative integer $p$. Note from \eqref{eq: classical SUSY differential} that the differential ${d_{\textup{st}}}_{(0|0)}$ preserves the filtration, implying that the spectral sequence made up from the above filtration has a vacuous differential on every $n$th page for $n\geq 1$. Therefore, we get the isomorphism of cohomologies
  \begin{equation} \label{eq: cohomology proof1}
    H(\widehat{C}_y,{d_{\textup{st}}}_{(0|0)})\simeq H(E_{y,0},{d_{\textup{st}}}_{(0|0)}),
  \end{equation}
  where $E_{y,0}$ is the total complex on the $0$th page. By the construction of the spectral sequence, $E_{y,0}$ has the same generators as $\widehat{C}_y/y\widehat{C}$, but the ground ring is $\CC[y]$ not $\CC$. Combining this with \eqref{eq: classical SUSY complex}, we get
  \begin{equation} \label{eq: cohomology proof2}
    H(E_{y,0},{d_{\textup{st}}}_{(0|0)})\simeq H(\widetilde{C}^{\infty}, {d_{\textup{st}}}_{(0|0)})\otimes \CC[y].
  \end{equation}
  Evaluate $y=\epsilon$ in \eqref{eq: cohomology proof1} and \eqref{eq: cohomology proof2}, then the isomorphism \eqref{eq: cohomology proof goal} follows. Thus, we can conclude that $S\supset \widetilde{S}$ is also Zariski dense in $\CC$.
\end{proof}
Notice that the set $S$ in Lemma \ref{lem: Zariski dense} can be restated as $S=\{k\in \CC\,|\, E_1^k\simeq E_1^{\infty}\}$.  Hence, the study of $E_1^{\infty}$ would contribute to understanding $E_1^k$ of generic $k$. Furthermore, since the cohomology \eqref{eq: SUSY W-alg total complex isomorphism} is concentrated on the charge $0$ space and the differential $d_2$ raises the charge by $1$, it is enough to see the subspace of $E_1^k$ of charge at most $1$. For $E^k_1$ and $E^{\infty}_1$, denote their subspace of charge at most $1$ by
\begin{equation} \label{eq: subspace of charge at most 1}
  E^k_1(\leq 1), \quad E^{\infty}_1(\leq 1),
\end{equation}
respectively. Now, we see that $E^k_1(\leq 1)$ can be viewed as a $V^{\xi_k}(\bar{\g}_0)$-module. Here, the space $V^{\xi_k}(\bar{\g}_0)$ is the SUSY affine vertex algebra with shifted level
\begin{equation*}
  \xi_k=k+\frac{1}{2}\kappa_{\g}-\frac{1}{2}\kappa_{\g_0}
\end{equation*}
as in \eqref{eq: susy affine level shift}.
\begin{proposition} \label{prop: the first total complex}
  Let $E_1^{\infty}$ be the cohomology \eqref{eq: E_1^infty} and consider its subspace $E^{\infty}_1(\leq 1)$ of charge at most $1$. As a vector superspace,
  \begin{equation*}
    E_1^{\infty}(\leq 1)\simeq \mathcal{V}(\bar{\g}_0)\oplus \Big(\mathcal{V}(\bar{\g}_0)\otimes \bigoplus_{\alpha\in I_0}\phi^{\alpha}\Big),
  \end{equation*}  
  where $\mathcal{V}(\bar{\g}_0)$ is the supercommutative SUSY vertex subalgebra of $E^{\infty}_1$ freely generated by $J_{\bar{\g}_0}=\{J_{\bar{a}_m}\,|\,m\in M\}$, where $\{a_m\,|\,m\in M\}$ is a basis for $\g_0$.
\end{proposition}
\begin{proof}
  See Appendix \ref{appendix: proof total complex}.
\end{proof}

As a result of Proposition \ref{prop: the first total complex}, we get
\begin{equation} \label{eq: E_1^k generic k}
  E_1^k {(\leq 1)}\simeq V^{\xi_k}(\bar{\g}_0) \oplus \Big(V^{\xi_k}(\bar{\g}_0)\otimes \bigoplus_{\alpha\in I_0}\phi^{\alpha}\Big),
\end{equation}
where $V^{\xi_k}(\bar{\g}_0)$ is the SUSY vertex subalgebra of $E^k_1$ freely generated by $J_{\bar{\g}_0}$. In particular, \eqref{eq: E_1^k generic k} implies that one does not need to consider the elements $\phi^{\alpha}$ for decomposable $\alpha\in I_+$. Hence, the $\Lambda$-brackets between the elements of \eqref{eq: E_1^k generic k} are all the same with Proposition \ref{prop: C^k elements relation}, except that the summation is over the set $I_0$.

Regarding $V^{\xi_k}(\bar{\g}_0)$ as a vertex algebra, the total complex $E_1^k$ can be seen as a $V^{\xi_k}(\bar{\g}_0)$-module with the adjoint action. Since we have the isomorphism $\eqref{eq: E_1^k generic k}$, it is easy to see that $E_1^k(\leq 1)$ is closed under the action of $V^{\xi_k}(\bar{\g}_0)$. However, take notice that each subspace  $V^{\xi_k}(\bar{\g}_0)\otimes \phi^{\alpha}$ for $\alpha\in I_0$ may not be closed under the action of $V^{\xi_k}(\bar{\g}_0)$. Instead, we have the decomposition as in Theorem \ref{thm: module structure of the first total complex}.

\begin{theorem} \label{thm: module structure of the first total complex}
 For generic $k\in \CC$, the space $E_1^k(\leq 1)$ in \eqref{eq: subspace of charge at most 1} is a $V^{\xi_k}(\bar{\g}_0)$-module with the adjoint action. Moreover, it can be decomposed as a direct sum of submodules as
  \begin{equation}
    E_1^k {(\leq 1)}\simeq V^{\xi_k}(\bar{\g}_0)\oplus \bigoplus_{[\alpha]\in [I_0]} M_{[\alpha]},
  \end{equation}
  where  $M_{[\alpha]}=V^{\xi_k}(\bar{\g}_0)\otimes \bigoplus_{\beta\in [\alpha]}\phi^{\beta}$ for each $[\alpha]\in [I_0]$. Here, the set $[I_0]$ is the set of equivalence classes with respect to \eqref{eq: equivalence relation}.
\end{theorem}
\begin{proof}
 Due to the observation above, it is enough to show that each $M_{[\alpha]}$ is a submodule of $E_1^k (\leq 1)$. As a result of \eqref{eq: E_1^k generic k} and Lemma \ref{lem: conditions with Lie brackets}, we have
  \begin{equation} \label{eq: M_alpha module structure}
    [J_{\bar{a}}{}_{\Lambda}\phi^{\alpha}]=\sum_{\beta\in [\alpha]}(-1)^{(p(\alpha)+1)p(a)}([a,u^{\alpha}]|u_{\beta})\phi^{\beta}
  \end{equation}
  for $a\in \g_0$ and $\alpha\in I_0$. Since the summation on the RHS is over the set $[\alpha]$, each $M_{[\alpha]}$ is a submodule of $E_1^k{(\leq 1)}$.
\end{proof}
When $\g_0=\h$, the description in Theorem \ref{thm: module structure of the first total complex} gets simpler. First, the SUSY vertex algebra $V^{\xi_k}(\bar{\g}_0)$ is isomorphic to the SUSY Heisenberg vertex algebra. Second, since the set $\Phi_0$ in \eqref{eq: Phi_0} is empty in this case, we get
\begin{equation} \label{eq: sets simpler description}
  I_+=\Phi_+, \quad I_0=\Pi, \quad [I_0]=I_0
\end{equation}
making the equivalence relation \eqref{eq: equivalence relation} trivial. Moreover, each submodule $M_{[\alpha]}$ can be understood as a Fock representation $\pi_{\beta}$ or $\widetilde{\pi}_{\beta}$ for some $\beta\in \h^*$.
\begin{corollary} \label{cor: first total complex susy fock}
  Assume that $\g_0=\h$ and denote the corresponding SUSY Heisenberg vertex algebra by $\widehat{\pi}$. Then for generic $k\in \CC$, $E_1^k{(\leq 1)}$ is a $\widehat{\pi}$-module and it is decomposed as a direct sum of $\widehat{\pi}$-submodules as
  \begin{equation}
    E_1^k{(\leq 1)}\simeq \bigoplus_{\alpha\in \Pi \sqcup \{0\}}\widehat{\pi}_{-\frac{\alpha}{\nu}},
  \end{equation}
  where $\nu=\sqrt{k+h^{\vee}}$ and $\widehat{\pi}_{-\frac{\alpha}{\nu}}$'s are the Fock representations defined by
\begin{equation} \label{eq: fock rep new notation}
  \widehat{\pi}_{-\frac{\alpha}{\nu}}=
  \left\{
    \begin{array}{ll}
      \pi^{\h}_{-\frac{\alpha}{\nu}} & \text{ if }\alpha \text{ is odd or zero},\\
      \widetilde{\pi}^{\h}_{-\frac{\alpha}{\nu}} & \text{ if }\alpha \text{ is even}.\\
    \end{array}\right.
\end{equation}
\end{corollary}
\begin{proof}
  Recall that $\widehat{\pi}$ is isomorphic to the SUSY affine vertex algebra $V^1(\bar{\h})$. In fact, the SUSY affine vertex algebra $V^{\xi_k}(\bar{\g}_0)=V^{k+h^{\vee}}(\bar{\h})$ for $k\neq -h^{\vee}$ is regardless of its level, since the generators $j_h:=\frac{1}{\nu}J_{\bar{h}}\in V^{k+h^{\vee}}(\bar{\h})$ for $h\in \h$ satisfy the relation
  \begin{equation} \label{eq: proof Heisenberg basis}
    [j_h{}_{\Lambda}j_{h'}]=(h|h')\chi,
  \end{equation}
  resolving the level difference. It implies that $V^{\xi_k}(\bar{\g}_0)\simeq V^1(\bar{\h})\simeq \widehat{\pi}$. Also, we have $[I_0]=\Pi$ in \eqref{eq: sets simpler description} and hence Theorem \ref{thm: module structure of the first total complex} gives us 
  \begin{equation}
    E_1^k{(\leq 1)}\simeq \widehat{\pi}_0 \oplus \bigoplus_{\alpha\in \Pi} M_{\alpha},
  \end{equation}
  where $M_{\alpha}=V^{\xi_k}(\bar{\g}_0)\otimes \phi^{\alpha}$ is a $\widehat{\pi}$-module with the adjoint action for each $\alpha\in \Pi$. Rewrite \eqref{eq: M_alpha module structure} with the new generators $j_h$'s for $\widehat{\pi}$  to get 
  \begin{equation} \label{eq: susy heisenberg action with Lambda bracket}
    [j_{\bar{h}}{}_{\Lambda}\phi^{\alpha}]=\frac{1}{\nu}([h,u^{\alpha}]| u_{\alpha})\phi^{\alpha}=-\frac{1}{\nu}\alpha(h) \phi^{\alpha}.
  \end{equation}
  Notice that the brackets \eqref{eq: proof Heisenberg basis} and \eqref{eq: susy heisenberg action with Lambda bracket} look the same as \eqref{eq: lattice VA Lambda bracket}. Thus, we have the following isomorphism of $\widehat{\pi}$-modules:
  \begin{equation} \label{eq: M_alpha as fock space}
    M_{\alpha} \rightarrow \widehat{\pi}_{-\frac{\alpha}{\nu}}, \quad \phi^{\alpha}\mapsto \lvert -\tfrac{\alpha}{\nu}\rangle.
  \end{equation}
\end{proof}

\subsection{Free field realization of SUSY W-algebras} \label{sec: free field realization} \hfill \\
We gather the results of Section \ref{sec: weight filtration} and \ref{sec: the first total complex} to get the realization of SUSY W-algebras in Theorem \ref{thm: main} and \ref{thm: main-screening}. Recall in \eqref{eq: SUSY W-alg total complex isomorphism} that SUSY W-algebras are isomorphic to the cohomology $H(E^k_1, d_2)$. Since we already have the description of $E^k_1$ from Section \ref{sec: the first total complex}, it remains to study the differential $d_2$ on $E^k_1$. In particular, when $\g_0=\h$, the action of $d_2$ induces the formula for the screening operators. Recall from \eqref{eq: second page differential} that $d_2={d_f}_{(0|0)}$ on $E^k_1$, where $d_f=-\sum_{\alpha\in I_+} (f|u_{\alpha})\phi^{\alpha}$.

\begin{lemma} \label{lem: action of phi^alpha decomposable}
  For generic $k\in \CC$, any $J_{\bar{a}}\in E_1^k$ for $a\in \g_0$ satisfies
  \begin{equation}
   {\phi^{\alpha}}_{(0|0)}J_{\bar{a}}=0
  \end{equation}  
  if $\alpha\in I_+$ is decomposable.
\end{lemma}
\begin{proof}
 Since $\alpha$ is decomposable, assume $u^{\alpha}=[u^{\beta},u^{\gamma}]$ for some $\beta, \gamma\in I_+$. Then Proposition \ref{prop: C^k elements relation} implies that in $\widetilde{C}^k$,
 \begin{equation} \label{eq: proof for decomposables}
  [\phi^{\alpha}{}_{\Lambda}J_{\bar{a}}]=(-1)^{p(\alpha)+1}\phi^{[[a,u^{\beta}],u^{\gamma}]}+(-1)^{p(a)p(\beta)+p(\alpha)+1}\phi^{[u^{\beta}[a,u^{\gamma}]]}
 \end{equation}
 for any $a\in \g_0$. Here, each of $[a,u^{\beta}]$ and $[a,u^{\gamma}]$ belongs to $\n$, which implies that each term in the RHS of \eqref{eq: proof for decomposables} is expressed as a linear sum of $\phi^{\tau}$'s for decomposable $\tau\in I_+$. As pointed out in \eqref{eq: E_1^k generic k}, $\phi^{\tau}=0$ in $E^k_1$ if $\tau$ is decomposable. Thus, the statement follows.
\end{proof}
Due to Lemma \ref{lem: action of phi^alpha decomposable}, the element $\phi^{\alpha}$ for decomposed $\alpha\in I_+$ does not contribute to the action of ${d_f}_{(0|0)}$ on $E^k_1$. Hence, one can regard the differential ${d_f}_{(0|0)}:E_1^k\rightarrow E_1^k$ as
\begin{equation} \label{eq: d_f modified}
  d_f=-\sum_{\alpha\in I_0} (f|u_{\alpha})\phi^{\alpha},
\end{equation}
whose summation is over the set of indecomposables $I_0$. Considering the $V^{\xi_k}(\bar{\g}_0)$-module structure of the first total complex in Theorem \ref{thm: module structure of the first total complex}, one can divide the summation \eqref{eq: d_f modified} using the equivalence relation in the set $I_0$.

\begin{theorem} \label{thm: main}
  For generic $k\in \CC$, the SUSY W-algebra of level $k$ is isomorphic to
  \begin{equation} \label{eq: statement W-algebra as kernel}
    W^k(\bar{\g},f)\simeq \bigcap_{[\alpha]\in[I_0]}\ker\Big(\sum_{\beta\in[\alpha]}(f|u_{\beta}){\phi^{\beta}}_{(0|0)}\Big)\subset V^{\xi_k}(\bar{\g}_0),
  \end{equation}
  where ${\phi^{\alpha}}_{(0|0)}: V^{\xi_k}(\bar{\g}_0)\rightarrow M_{[\alpha]}$ is given by the adjoint action on $E^k_1$.
\end{theorem}
\begin{proof}
Combining Theorem \ref{thm: module structure of the first total complex}, \eqref{eq: d_f modified} and \eqref{eq: SUSY W-alg total complex isomorphism}, we have
\begin{equation} \label{eq: proof W-algebra as kernel}
  W^k(\bar{\g},f)\simeq \ker \Big(\sum_{\alpha\in I_0}(f|u_{\alpha})\phi^{\alpha}_{(0|0)}\Big)\subset V^{\xi_k}(\bar{\g}_0).
\end{equation}
This is because the SUSY W-algebra is concentrated on the charge zero subspace. To show the statement, it is enough to see that the range of $\phi^{\alpha}_{(0|0)}$ and $\phi^{\beta}_{(0|0)}$ for $\alpha, \beta\in I_0$ are disjoint in $E_1^k{(\leq 1)}$ if $\alpha, \beta$ are not equivalent. In particular, we show that the range of $\phi^{\alpha}_{(0|0)}$ is contained in the $V^{\xi_k}(\bar{\g}_0)$-module $M_{[\alpha]}$ in Theorem \ref{thm: module structure of the first total complex}.
  Notice from Proposition \ref{prop: C^k elements relation} and Lemma \ref{lem: conditions with Lie brackets} that for any $\alpha\in I_0$ and $a\in \g_0$, we have
  \begin{equation*}
    [\phi^{\alpha}{}_{\Lambda}J_{\bar{a}}]=\sum_{\beta\in [\alpha]}(-1)^{p(\alpha)+1}([a,u^{\alpha}]|u_{\beta})\phi^{\beta}
  \end{equation*}
  on $E^k_1$. Since the summation on the RHS is over the equivalence class $[\alpha]$, we conclude that the range of the linear map
  \begin{equation}
    {\phi^{\alpha}}_{(0|0)}:V^{\xi_k}(\bar{\g}_0)\rightarrow M_{[\alpha]}, \quad J_{\bar{a}}\mapsto {\phi^{\alpha}}_{(0|0)}J_{\bar{a}}.
  \end{equation}
  is contained in $M_{[\alpha]}$.
\end{proof}
When $\g_0=\h$, the SUSY affine vertex algebra $V^{\xi_k}(\bar{\g}_0)$ is isomorphic to the SUSY Heisenberg vertex algebra $\widehat{\pi}$ as in Corollary \ref{cor: first total complex susy fock}. Furthermore, we have the simpler description of the set $[I_0]$ and the range of $\phi^{\alpha}$'s in \eqref{eq: sets simpler description} and \eqref{eq: M_alpha as fock space}. Thus, Theorem \ref{thm: main} immediately shows that the SUSY W-algebra of a generic level is embedded into the SUSY Heisenberg vertex algebra.
\begin{corollary} \label{cor: FF realization of susy W-alg}
  Assume that $\g_0=\h$. Consider the corresponding SUSY Heisenberg vertex algebra $\widehat{\pi}$, and their Fock representations $\widehat{\pi}_{-\frac{\alpha}{\nu}}$'s for $\alpha\in \Pi$ defined in \eqref{eq: fock rep new notation}. Then for generic $k\in \CC$, the SUSY W-algebra of level $k$ can be realized as a vertex subalgebra of $\widehat{\pi}$ given by
  \begin{equation} \label{eq: statement W-algebra FF realization}
    W^k(\bar{\g},f)\simeq \bigcap_{\substack{\alpha\in \Pi\\ (f|u_{\alpha})\neq 0}}\ker\left({\left| -\frac{\alpha}{\nu} \right>}_{(0|0)}\right)\subset \widehat{\pi},
  \end{equation}
  where $\nu=\sqrt{k+h^{\vee}}$ and ${\left| -\frac{\alpha}{\nu} \right>}_{(0|0)}:\widehat{\pi}\rightarrow \widehat{\pi}_{-\frac{\alpha}{\nu}}$ is given by the adjoint action in $E^k_1(\leq 1)\simeq \bigoplus_{\alpha\in \Pi\sqcup\{0\}}\widehat{\pi}_{-\frac{\alpha}{\nu}}$ via the identification $\left| -\frac{\alpha}{\nu}\right>\in \widehat{\pi}_{-\frac{\alpha}{\nu}}$ with $\phi^{\alpha}$ in $E^k_1(\leq 1)$. \qed
\end{corollary}

The identification in Corollary \ref{cor: FF realization of susy W-alg} is also used in proving Corollary \ref{cor: first total complex susy fock}. Recall in the proof that the isomorphisms are given by
\begin{equation} \label{eq: heisenberg and fock isomorphism}
  \begin{array}{rclrcl}
    V^{k+h^{\vee}}(\bar{\h}) & \rightarrow& \widehat{\pi}\ , & \quad  j_h=\frac{1}{\nu}J_{\bar{h}} & \mapsto& \bar{h}_{(-1)}\vac\ ; \\
    V^{k+h^{\vee}}(\bar{\h})\otimes \phi^{\alpha} &\rightarrow &\widehat{\pi}_{-\frac{\alpha}{\nu}}, & \quad \phi^{\alpha} & \mapsto& \left|-\frac{\alpha}{\nu} \right> \ .
  \end{array}
\end{equation}
From now on, we fix $\widehat{\pi}_{-\frac{\alpha}{\nu}}$ to denote the Fock representations in \eqref{eq: fock rep new notation} for $\alpha\in \Pi\sqcup \{0\}$, and do not distinguish the elements of each pair of the isomorphic spaces under the map \eqref{eq: heisenberg and fock isomorphism}. Using the properties of the superfields in Appendix \ref{appendix: superfield formalism}, one can write explicitly the formula of $\left|-\frac{\alpha}{\nu}\right>_{(0|0)}$ for each $\alpha\in \Pi$ in \eqref{eq: statement W-algebra FF realization}. 

For the statement, we introduce the necessary notations. For an even variable $z$ and an odd variable $\theta$, denote them as a tuple $Z=(z, \theta)$ if $\theta^2=0$. For $n\in \ZZ$ and $i=0,1$, define the $n|i$th power of $Z$ by
\begin{equation} \label{eq: notation tuple power}
  Z^{n|i}:=z^n \theta^i.
\end{equation}
Let $V$ be a vector superspace. The space of formal Laurent series in $Z$ with coefficients in $V$ is denoted by $V(\!(Z)\!):=\CC(\!(z)\!)\otimes \CC[\theta]\otimes V$. An element of $V(\!(Z)\!)$ is of the form
\begin{equation} \label{eq: Laurent series}
  A(Z)=\sum_{\substack{n\geq N\\ i=0,1}} Z^{n|i}A_{n|i}=\sum_{n\geq N} z^n A_{n|0}+\theta\sum_{n\geq N}z^n A_{n|1}
\end{equation}
for some $N\in \ZZ$ and $A_{n|i}\in V$. 
The linear map called \textit{super residue} is given by 
\begin{equation} \label{eq: super residue}
  \int dZ : V(\!(Z)\!) \rightarrow V,\quad A(Z) \mapsto\int A(Z) dZ=A_{-1|1}
\end{equation}
 for  $A(Z)$ in \eqref{eq: Laurent series}. For any simple root $\alpha\in \Pi$, define $e^{-\frac{1}{\nu}\int \alpha(Z)}$ to be the formal Laurent series with coefficients in $\textup{End}_{\CC}(\widehat{\pi}, \widehat{\pi}_{-\frac{\alpha}{\nu}})$ by
\begin{equation} \label{eq: exponential definition}
  \begin{aligned}
    e^{-\frac{1}{\nu}\int \alpha(Z)}=&s_{-\frac{\alpha}{\nu}}\exp\Big(\!-\frac{1}{\nu}\sum_{n<0}Z^{-n-1|1}\alpha_{(n|1)}\Big)\exp\Big(\, \frac{1}{\nu}\sum_{n<0}\frac{Z^{-n|0}}{j}\alpha_{(n|0)}\Big)\\
    &\exp\Big(\,\frac{1}{\nu}\sum_{n>0}\frac{Z^{-n|0}}{n}\alpha_{(n|0)}\Big)\exp\Big(\!-\frac{1}{\nu}\sum_{n\geq 0}Z^{-n-1|1}\alpha_{(n|1)}\Big)Z^{-\frac{1}{\nu}\alpha_{(0|0)}|0}.
  \end{aligned}
\end{equation}
In \eqref{eq: exponential definition}, the exponential of the Laurent series is defined using the power sum, and the element $\alpha\in \widehat{\pi}$ is given by $\alpha=(\bar{h}_{\alpha})_{(-1)}\vac \in \widehat{\pi}$ for the coroot $h_{\alpha}\in \h$ of $\alpha$. Also, $s_{-\frac{\alpha}{\nu}}:\widehat{\pi}\rightarrow \widehat{\pi}_{-\frac{\alpha}{\nu}}$ is the linear operator of parity $p(u_\alpha)+1$ determined by the properties
\begin{equation} \label{eq: commutator with shift operator}
  \begin{aligned}
  s_{-\frac{\alpha}{\nu}}\vac=\left|-\frac{\alpha}{\nu}\right>, \quad [\bar{h}_{(n|i)},s_{-\frac{\alpha}{\nu}}]=0 \text{ for } (n|i)\neq 0, \quad [\bar{h}_{(0|0)}, s_{-\frac{\alpha}{\nu}}]=-\frac{1}{\nu} \beta(h)s_{-\frac{\alpha}{\nu}}
  \end{aligned}
\end{equation} 
for any $h\in \h$. 
\begin{theorem} \label{thm: main-screening}
  Assume that $\g_0=\h$. Then the SUSY W-algebra of generic level $k$ can be realized as a SUSY vertex subalgebra of the SUSY Heisenberg vertex algebra $\widehat{\pi}$ in the following way:
  \begin{equation} \label{eq: susy screening}
    W^k(\bar{\g},f)\simeq \bigcap_{\substack{\alpha\in \Pi\\(f|u_{\alpha})\neq 0}}\! \ker \int  e^{-\frac{1}{\nu} \int \alpha(Z)} dZ\subset \widehat{\pi}.
  \end{equation}
   Here, the formula of $e^{-\frac{1}{\nu}\int \alpha(Z)}$ is given in \eqref{eq: exponential definition} for $\nu=\sqrt{k+h^{\vee}}$, and $\int dZ$ is the super residue in \eqref{eq: super residue}. For each $\alpha\in \Pi$ with $(f|u_{\alpha})\neq 0$, we call $\int  e^{-\frac{1}{\nu} \int \alpha(Z)} dZ$ the \textit{screening operator} for $\alpha.$
\end{theorem}
\begin{proof}
  We postpone the proof until Appendix \ref{appendix: proof screening operator} since it uses the properties of the superfields in Appendix \ref{appendix: superfield formalism}.
\end{proof}

\begin{remark}
  Recall in the W-algebra theory, the formulas of the screening operators are given by the vertex operators of the highest weight vectors in lattice vertex algebras. For SUSY W-algebras, the analogous statement holds. This can be checked by the paper \cite{HK08lattice} which shows that \eqref{eq: exponential definition} is the superfield of $\left| -\frac{\alpha}{\nu}\right>$ in the SUSY lattice vertex algebra $V:=\bigoplus_{\alpha\in \Pi\sqcup\{0\}}\widehat{\pi}_{-\frac{\alpha}{\nu}}$. Moreover, one can see from Appendix \ref{appendix: proof screening operator} that the SUSY vertex algebra structure of $V$ given by the isomorphism in Corollary \ref{cor: first total complex susy fock} is the same as the one in \cite{HK08lattice}.
\end{remark}
\subsection{Relation with Miura map} \label{sec: relation with Miura} \hfill \\
In this section, we give the concrete expression of the realization of SUSY W-algebras in Theorem \ref{thm: main}. To be precise, we show that the isomorphism is the same as the Miura map.  From this description, one can find the corresponding element in $V^{\xi_k}(\bar{\g}_0)$ whenever we find the element of SUSY W-algebra in terms of the building blocks.

Recall from Theorem \ref{thm: susy W-alg charge 0} that the SUSY W-algebra $W^k(\bar{\g},f)$ can be realized as a vertex subalgebra of $V^{\zeta_k}(\bar{\g}_{\leq 0})$. Compose it with the the projection $V^{\zeta_k}(\bar{\g}_{\leq 0})\twoheadrightarrow V^{\xi_k}(\bar{\g}_0)$ to obtain the linear map
\begin{equation} \label{eq: miura map definition}
 \mu_k:  W^k(\bar{\g},f) \rightarrow V^{\xi_k}(\bar{\g}_0),
\end{equation}
called the \textit{Miura map} for a SUSY W-algebra of level $k\in \CC$. Note that any element of $W^k(\bar{\g},f)$ can be presented using the building blocks $J_{\bar{a}}$'s for $a\in \g_{\leq 0}$, and the Miura map $\mu_k$ sends all the $J_{\bar{a}}$'s to $0$ unless $a\in \g_{0}$.
\begin{theorem} \label{thm: miura}
  For generic $k$, the image of the Miura map $\mu_k$ in \eqref{eq: miura map definition} coincides with the description of SUSY W-algebra in Theorem \ref{thm: main}. In particular, the Miura map $\mu_k$ is injective generically.
\end{theorem}
\begin{proof}
  Recall that the isomorphism in \eqref{eq: SUSY W-alg total complex isomorphism} gives the image of SUSY W-algebras in Theorem \ref{thm: main}. If we also consider the charge degree, we have
  \begin{equation} \label{eq: miura map proof isomorphism}
    W^k(\bar{\g},f)=H^0(\widetilde{C}^k,d_{(0|0)})\simeq H^0(E_1^k, d_2).
  \end{equation}
  Note that this isomorphism is induced from the spectral sequence argument for the weight filtration. Considering the mechanism of the spectral sequence \cite{Chow06,Weibel94}, the terms of minimal weight of the cohomology class in $H(\widetilde{C}^k,d_{(0|0)})$ are realized as the corresponding element of $E_1^k$. However, the space $E^k_1$ contains extra elements that do not correspond to the cohomology class in $H(\widetilde{C}^k, d_{(0|0)})$, and the cohomology in \eqref{eq: miura map proof isomorphism} with the differential $d_2$ removes these extra elements. As a result of \eqref{eq: E_1^k generic k}, the charge $0$ subspace of $E_1^k$ is equal to the space $V^{\xi_k}(\bar{\g}_0)$, which is exactly same as the weight $0$ subspace of $\widetilde{C}^k$. Hence, it implies that each cohomology class in $H(\widetilde{C}^k, d_{(0|0)})$ has nonzero weight $0$ terms and the isomorphism \eqref{eq: miura map proof isomorphism} is the projection onto the weight $0$ summand. It coincides with the definition of Miura map $\mu_k$. 
\end{proof}

We still need to figure out the value of $k$ that Theorem \ref{thm: miura} holds. In the theory of W-algebras associated with Lie algebras, this problem is partially resolved via the Wakimoto modules. Recall that if we consider a Lie algebra $\g$, the corresponding affine vertex algebra of a generic level has its free field realization given by the Wakimoto module \cite{Wakimoto86, FeiginFrenkel90}. Applying the Hamiltonian reduction functor(or Whittaker reduction functor for finite case) to this Wakimoto module, one gets the free field realization of the corresponding W-algebra \cite{Arakawa07}. Furthermore, it lets us to some extent find values of $k\in \CC$ which enable the free field realization of W-algebras. The existence of the SUSY Wakimoto module was conjectured in \cite{DKPR85,Ito92,MR94}, so one can expect the supersymmetric version of this theory. This is a work in progress.

\begin{remark}
  In the case of classical SUSY W-algebras, the injectivity of the Miura map follows from the explicit formulas of the generators. Using the notations of \cite{Suh20}, let $\{r_j\,|\,j\in J^f\}$ and $\{r^j\,|\,j\in J^f\}$ be bases of $\ker(\textup{ad}f)$ and $\ker(\textup{ad}e)$ in $\g$ satisfying $(r^i|r_j)=\delta_{i,j}$. Here $e$ and $f$ are the odd nilpotent elements in the subalgebra of $\g$ isomorphic to $\mathfrak{osp}(1|2)$. Let
  \begin{equation*}
    r^i_m=(\textup{ad}f)^m r^i, \quad r_j^n=c_{j,n}(\textup{ad}e)^n r_j
  \end{equation*}
  for $i,j\in J^f$, $m,n\in \ZZ_{\geq 0}$ and some constants $c_{j,n}$ so that $(r^i_m|r_j^n)=\delta_{m,n}\delta_{i,j}$. Then, the classical W-algebra is generated by $\{\omega(\bar{r}_j)\,|\,j\in J^f\}$, where $\omega(\bar{r}_j)$ is an element of the classical SUSY W-algebra containing $\bar{r}_j$ as a linear term. Hence, it is enough to show that the Miura map is injective on the set of generators. Due to the presentation of $\omega(\bar{r}_j)$ in Theorem 6.3 of \cite{Suh20}, the terms of $\omega(\bar{r}_j)$ contained in $\CC[\nabla]\otimes \g_0$ is equal to $(-D)^{2l}\,\bar{r}_j^{2l}$, when $r_j\in \g_{-l}$. Thus, the injectivity follows.
\end{remark}

\section{Application: principal SUSY W-algebras} \label{sec: application to principal SUSY W-algebras}
The most important types of Lie superalgebras that satisfy the assumption $\g_0=\h$ would be the Lie superalgebras that admit the existence of the principal $\mathfrak{osp}(1|2)$ embedding. For such Lie superalgebra $\g$, denote the odd nilpotent element in the principal $\mathfrak{osp}(1|2)$ subalgebra by $f_{\textup{prin}}$. Additionally, we denote the shifted level of a SUSY W-algebra by $\Psi:=k-h^{\vee}(\g)$ for $k\in \CC$, where $h^{\vee}(\g)$ is the dual coxeter number of $\g$. In this section, we analyze the structure of the principal SUSY W-algebra $W^{\Psi}(\bar{\g}):=W^{\Psi}(\bar{\g}, f_{\textup{prin}})$ using Theorem \ref{thm: main-screening}. Notice here that the shifted level $\Psi$ depends on the choice of the Lie superalgebra $\g$. Without further mention, we assume that all the statements apply to SUSY W-algebras of generic levels.


\begin{proposition}[\cite{FRS93,LSS86}]\label{prop: principal osp(1|2) classification}
  Let $\g$ be a finite basic simple Lie superalgebra that admits a principal $\mathfrak{osp}(1|2)$ embedding. Then $\g$ is isomorphic to one of the following:
  \begin{equation*}
    \mathfrak{sl}(n\pm 1|n),\ \mathfrak{osp}(2n\pm 1|2n),\ \mathfrak{osp}(2n+2|2n),\ \mathfrak{osp}(2n|2n),\ D(2,1;\alpha)\ (\alpha\in \CC\setminus\{0, \pm 1\}).
  \end{equation*}
\end{proposition} 

Due to the above classification, we only deal with Lie superalgebras listed in Proposition \ref{prop: principal osp(1|2) classification} except for $D(2,1;\alpha)$.

\subsection{Analysis of principal SUSY W-algebras} \hfill \\
We first provide the statement and the main idea of Theorem \ref{thm: principal SUSY W-algebra} in this section. For each Lie superalgebra $\g$, we denote its Cartan subalgebra and the associated SUSY Heisenberg vertex algebra by $\h$ and $\widehat{\pi}$, respectively.

\begin{theorem} \label{thm: principal SUSY W-algebra}
  Let $\g$ be a finite basic simple Lie superalgebra that admits a principal $\mathfrak{osp}(1|2)$ embedding. Then except for $\g=D(2,1;\alpha)$, the principal SUSY W-algebra $W^{\Psi}(\bar{\g})$ of generic level can be expressed as the intersection of the vertex algebras in the following list:
  \begin{equation} \label{eq: principal list}
    W^{\Psi}(\overline{\mathfrak{osp}(1|2)}),\quad W^{\Psi}(\overline{\mathfrak{osp}(2|2)}),\quad W^{\Psi}(\overline{\mathfrak{osp}(3|2)}),\quad W^{\Psi}(\overline{\mathfrak{osp}(4|2)}),
  \end{equation}
  where the list \eqref{eq: principal list} is up to the tensor product with a SUSY Heisenberg vertex algebra.
\end{theorem}
It is worth to be noted that in the list \eqref{eq: principal list}, each level $\Psi$ is shifted with the dual coxeter number for each Lie superalgebra. If we clarify the dual coxeter numbers, \eqref{eq: principal list} is written as
  \[W^{k-\frac{3}{2}}(\overline{\mathfrak{osp}(1|2)}), \quad W^{k-1}(\overline{\mathfrak{osp}(2|2)}), \quad W^{k-\frac{1}{2}}(\overline{\mathfrak{osp}(3|2)}), \quad W^{k}(\overline{\mathfrak{osp}(4|2)}).\]
  Refer to \cite{Musson12} for explicit values of the dual coxeter numbers for basic simple Lie superalgebras.
  
  For the proof, we use the results of Section \ref{bigsec: FF realization of SUSY W-algebras}. Note that all the results are applied with $\nu$ being replaced with $\sqrt{k}$ since we are dealing with the shifted level $\Psi$. Furthermore, we can assume that the intersection in \eqref{eq: susy screening} runs over all the odd simple roots. It follows from the proof of Proposition \ref{prop: principal osp(1|2) classification} in \cite{LSS86}. They show that each Lie superalgebra in the list has a simple root system $\Pi$ which only consists of odd roots. Moreover, the odd principal nilpotent element $f_{\textup{prin}}$ of $\g$ is given by the linear sum of the simple root vectors. It implies that none of the simple root $\alpha\in \Pi$ satisfies $(f_{\textup{prin}}|u_{\alpha})=0$. Lastly, the following observation in Lemma \ref{lem: perp decomposition} enables us to prove the theorem.

\begin{lemma} \label{lem: perp decomposition}
  Let $\g$ be one of the Lie superalgebras listed in Proposition \ref{prop: principal osp(1|2) classification}. Fix one simple root $\alpha$ of $\g$ and let $\mathfrak{u}_{\alpha}^{\perp}$ be the subspace of $\h$ given by
  \begin{equation} \label{eq: subspace u_alpha}
  \mathfrak{u}_{\alpha}^{\perp}=\textup{Span}_{\CC}\{h\in \h\mid \alpha(h)=0\}.
  \end{equation}
 Then the SUSY Heisenberg vertex algebra $\widehat{\pi}_{\alpha}^{\perp}$ associated with $\mathfrak{u}_{\alpha}^{\perp}$ can be regarded as a subalgebra of $\widehat{\pi}$ and the screening operator $\int e^{-\frac{1}{\sqrt{k}}\int \alpha(Z)}dZ$ acts trivially on $\widehat{\pi}_{\alpha}^{\perp}\subset \widehat{\pi}$.
\end{lemma}
\begin{proof}
  The screening operator for $\alpha$ follows from the action $\phi^{\alpha}_{(0|0)}$ in the SUSY BRST complex. Since $u^{\alpha}$ is a root vector of $-\alpha$, we have
  \begin{equation*}
    [\phi^{\alpha}{}_{\Lambda}J_{\bar{h}}]=\sum_{\beta\in I_+}(-1)^{p(\alpha)+1}([h,u^{\alpha}]|u_{\beta})\phi^{\beta}=(-1)^{p(\alpha)}\alpha(h)\phi^{\alpha}
  \end{equation*}
  for any $h\in\h$. Then the statement directly follows.
\end{proof}

\subsection{Proof of Theorem \ref{thm: principal SUSY W-algebra}} \hfill \\
We devote this section to proving Theorem \ref{thm: principal SUSY W-algebra}, case by case. To denote the simple roots for the Lie superalgebras, we use the notation $\epsilon_i$'s and $\delta_j$'s in \cite{ChengWang12}.

\subsubsection{$\mathfrak{sl}(n\pm 1|n)$} \label{sec: sl example} \hfill \\
Consider $\g=\mathfrak{sl}(n+1|n)$ and identify any root of $\g$ with its coroot using the nondegeneracy of the bilinear form. It has a simple root system $\Pi$, whose Dynkin diagram is given as follows.
\begin{center}
  \begin{tikzpicture}[cross/.style={path picture={ 
    \draw[black]
  (path picture bounding box.south east) -- (path picture bounding box.north west) (path picture bounding box.south west) -- (path picture bounding box.north east);
  }}, every label/.append style={font=\scriptsize}]
\node(AA) at (-2,0) {$\mathfrak{sl}(n+1|n)$ :};
\node[draw,circle,cross,minimum size=0.3cm,label=above:${\delta_1-\epsilon_1}$](A) at (0,0) {};
\node[draw,circle,cross,minimum size=0.3cm,label=above:${\epsilon_1-\delta_2}$](B) at (1.5,0) {};
\node[draw,circle,cross,minimum size=0.3cm,label=above:${\delta_2-\epsilon_2}$](C) at (3,0) {};
\node(D) at (4,0) {$\cdots$};
\node[draw,circle,cross,minimum size=0.3cm,label=above:${\delta_n-\epsilon_n}$](E) at (5,0) {};
\node[draw,circle,cross,minimum size=0.3cm,label=above:${\epsilon_n-\delta_{n+1}}$](F) at (6.5,0) {};
\draw[-] (A) -- (B);
\draw[-] (B) -- (C);
\draw[-] (C) -- (D);
\draw[-] (D) -- (E);
\draw[-] (E) -- (F);
  \end{tikzpicture}
\end{center}
Denote each odd simple root in the above diagram by $\alpha_1. \cdots, \alpha_{2n}$ in sequence. Write the corresponding subspace \eqref{eq: subspace u_alpha} of $\h$ by $\mathfrak{u}_1^{\perp}, \cdots, \mathfrak{u}_{2n}^{\perp}$ and their intersections by $\mathfrak{u}_{i,j}^{\perp}:=\mathfrak{u}_i^{\perp}\cap \mathfrak{u}_j^{\perp}.$ For each $i=1, \cdots, 2n-1$, we have
\begin{equation} \label{eq: cartan decompose into u and u perp}
  \h=\mathfrak{u}_{i, i+1}^{\perp} \oplus \mathfrak{u}_{i,i+1}
\end{equation}
for $\mathfrak{u}_{i,i+1}:=\CC \alpha_{i}+\CC\alpha_{i+1}$. With the use of \eqref{eq: cartan decompose into u and u perp}, we decompose the domain and codomains of the screening operators for $\alpha_i$ and $\alpha_{i+1}$.

Let $\widehat{\pi}_{i, i+1}^{\perp}$ and $\widehat{\pi}_{i, i+1}$ be the SUSY Heisenberg vertex algebras associated with $\mathfrak{u}_{i,i+1}^{\perp}$ and $\mathfrak{u}_{i,i+1}$, respectively. Then, the vertex algebra $\widehat{\pi}$ decomposes as
\begin{equation} \label{eq: Heisenberg decompose into u and u perp}
  \widehat{\pi}\simeq \widehat{\pi}_{i, i+1}^{\perp} \otimes \widehat{\pi}_{i,i+1}.
\end{equation}
Likewise, as a $\CC$-vector space, the codomain $\widehat{\pi}_{-\frac{\alpha_j}{\sqrt{k}}}$ decomposes as
\[\widehat{\pi}_{-\frac{\alpha_j}{\sqrt{k}}}\simeq \widehat{\pi}_{i,i+1}^{\perp}\otimes (\widehat{\pi}_{-\frac{\alpha_j}{\sqrt{k}}})_{i, i+1},\]
where $(\widehat{\pi}_{-\frac{\alpha_j}{\sqrt{k}}})_{i, i+1}$ is the even Fock representation of $\widehat{\pi}_{i,i+1}$ with an highest weight $-\frac{\alpha_j}{\sqrt{k}}$. By Lemma \ref{lem: perp decomposition}, the screening operators for $\alpha_i$ and $\alpha_{i+1}$ both act trivially on $\widehat{\pi}_{i, i+1}^{\perp}$. Thus, the common kernel of the two operators is isomorphic to the common kernel of their restrictions
\begin{equation*}
  \int e^{-\frac{1}{\sqrt{k}}\int \alpha_{i}(Z)}dZ : \widehat{\pi}_{i, i+1} \rightarrow \big(\widehat{\pi}_{-\frac{\alpha_i}{\sqrt{k}}}\big)_{i, i+1}, \quad \int e^{-\frac{1}{\sqrt{k}}\int \alpha_{i+1}(Z)}dZ : \widehat{\pi}_{i, i+1} \rightarrow \big(\widehat{\pi}_{-\frac{\alpha_{i+1}}{\sqrt{k}}}\big)_{i, i+1}
\end{equation*} 
tensored with $\widehat{\pi}_{i, i+1}^{\perp}$. Note here that one cannot erase the interference of other simple roots in this inspection, since each $\alpha_i$ sits inside the space $\mathfrak{u}_i^{\perp}$. Now, make the pairings of the simple roots as follows.
\begin{center}
  \begin{tikzpicture}[cross/.style={path picture={ 
    \draw[black]
  (path picture bounding box.south east) -- (path picture bounding box.north west) (path picture bounding box.south west) -- (path picture bounding box.north east);
  }}, every label/.append style={font=\scriptsize}]
\node(AA) at (-2,0) {$\mathfrak{sl}(n+1|n)$ :};
\node[draw,circle,cross,minimum size=0.3cm,label=above:${\delta_1-\epsilon_1}$](A) at (0,0) {};
\node[draw,circle,cross,minimum size=0.3cm,label=above:${\epsilon_1-\delta_2}$](B) at (1.5,0) {};
\node[draw,circle,cross,minimum size=0.3cm,label=above:${\delta_2-\epsilon_2}$](C) at (3,0) {};
\node[draw,circle,cross,minimum size=0.3cm,label=above:${\epsilon_2-\delta_3}$](D) at (4.5,0) {};
\node(E) at (5.5,0) {$\cdots$};
\node[draw,circle,cross,minimum size=0.3cm,label=above:${\delta_n-\epsilon_n}$](F) at (6.5,0) {};
\node[draw,circle,cross,minimum size=0.3cm,label=above:${\epsilon_n-\delta_{n+1}}$](G) at (8,0) {};
\node(GG) at (5.5,-0.83){$\cdots$};
\draw[-] (A) -- (B);
\draw[-] (B) -- (C);
\draw[-] (C) -- (D);
\draw[-] (D) -- (E);
\draw[-] (E) -- (F);
\draw[-] (F) -- (G);
\draw [decorate,decoration={brace,amplitude=5pt,mirror,raise=1.5ex}] (0,0) -- (1.5,0) node[midway,yshift=-2em]{\scriptsize{$\mathfrak{osp}(2|2)$}};
\draw [decorate,decoration={brace,amplitude=5pt,mirror,raise=1.5ex}] (3,0) -- (4.5,0) node[midway,yshift=-2em]{\scriptsize{$\mathfrak{osp}(2|2)$}};
\draw [decorate,decoration={brace,amplitude=5pt,mirror,raise=1.5ex}] (6.5,0) -- (8,0) node[midway,yshift=-2em]{\scriptsize{$\mathfrak{osp}(2|2)$}};
  \end{tikzpicture}
\end{center}
Considering the shape of the Dynkin diagram, the common kernel of the screening operators for the two adjacent simple roots is isomorphic to the principal SUSY W-algebra corresponding to $\mathfrak{osp}(2|2)$, applying Theorem \ref{thm: main-screening} again. In other words, for $i=1, \cdots n$,
\begin{equation} \label{eq: sl example analogous statement}
  \ker \int e^{-\frac{1}{\sqrt{k}}\int \alpha_{2i-1}(Z)}dZ \cap \ker \int e^{-\frac{1}{\sqrt{k}}\int \alpha_{2i}(Z)}dZ \simeq W^{\Psi}(\overline{\mathfrak{osp}(2|2)})\otimes \widehat{\pi}_{2i-1, 2i}^{\perp}.
\end{equation}
A similar argument for $\mathfrak{sl}(n|n+1)$ follows if we switch the role of $\epsilon$ and $\delta$. In this case, one can also deduce that the common kernel of the $i$th and $(i+1)$th screening operators can be explained by the tensor product of 
$W^{\Psi}(\overline{\mathfrak{osp}(2|2)})$ and a Fock representation.

\begin{center}
  \begin{tikzpicture}[cross/.style={path picture={ 
    \draw[black]
  (path picture bounding box.south east) -- (path picture bounding box.north west) (path picture bounding box.south west) -- (path picture bounding box.north east);
  }}, every label/.append style={font=\scriptsize}]
\node(AA) at (-2,0) {$\mathfrak{sl}(n|n+1)$ :};
\node[draw,circle,cross,minimum size=0.3cm,label=above:${\epsilon_1-\delta_1}$](A) at (0,0) {};
\node[draw,circle,cross,minimum size=0.3cm,label=above:${\delta_1-\epsilon_2}$](B) at (1.5,0) {};
\node[draw,circle,cross,minimum size=0.3cm,label=above:${\epsilon_2-\delta_2}$](C) at (3,0) {};
\node[draw,circle,cross,minimum size=0.3cm,label=above:${\delta_2-\epsilon_3}$](D) at (4.5,0) {};
\node(E) at (5.5,0) {$\cdots$};
\node[draw,circle,cross,minimum size=0.3cm,label=above:${\epsilon_n-\delta_n}$](F) at (6.5,0) {};
\node[draw,circle,cross,minimum size=0.3cm,label=above:${\delta_n-\epsilon_{n+1}}$](G) at (8,0) {};
\node(GG) at (5.5,-0.83){$\cdots$};
\draw[-] (A) -- (B);
\draw[-] (B) -- (C);
\draw[-] (C) -- (D);
\draw[-] (D) -- (E);
\draw[-] (E) -- (F);
\draw[-] (F) -- (G);
\draw [decorate,decoration={brace,amplitude=5pt,mirror,raise=1.5ex}] (0,0) -- (1.5,0) node[midway,yshift=-2em]{\scriptsize{$\mathfrak{osp}(2|2)$}};
\draw [decorate,decoration={brace,amplitude=5pt,mirror,raise=1.5ex}] (3,0) -- (4.5,0) node[midway,yshift=-2em]{\scriptsize{$\mathfrak{osp}(2|2)$}};
\draw [decorate,decoration={brace,amplitude=5pt,mirror,raise=1.5ex}] (6.5,0) -- (8,0) node[midway,yshift=-2em]{\scriptsize{$\mathfrak{osp}(2|2)$}};
  \end{tikzpicture}
\end{center}

\subsubsection{$\mathfrak{osp}(2n+1|2n)$ and $\mathfrak{osp}(2n|2n)$} \label{sec: osp(2n + 1|2n) osp(2n|2n) example} \hfill \\
When $\g=\mathfrak{osp}(2n+1|2n)$ or $\mathfrak{osp}(2n|2n)$, $\g$ has an odd root system $\Pi$ such that $\vert\Pi\vert$ is even. Therefore, the simple roots make pairings two by two as in Section \ref{sec: sl example}. First, consider $\mathfrak{osp}(2n+1|2n)$ equipped with a supertrace form. It has a simple root system which draws a Dynkin diagram as follows.
\begin{center}
  \begin{tikzpicture}[cross/.style={path picture={ 
    \draw[black]
  (path picture bounding box.south east) -- (path picture bounding box.north west) (path picture bounding box.south west) -- (path picture bounding box.north east);
  }}, every label/.append style={font=\scriptsize}, >={Classical TikZ Rightarrow[]}]
\node(AA) at (-2,0) {$\mathfrak{osp}(2n+1|2n)$ :};
\node[draw,circle,cross,minimum size=0.3cm,label=above:${\epsilon_1-\delta_1}$](A) at (0,0) {};
\node[draw,circle,cross,minimum size=0.3cm,label=above:${\delta_1-\epsilon_2}$](B) at (1.5,0) {};
\node[draw,circle,cross,minimum size=0.3cm,label=above:${\epsilon_2-\delta_2}$](C) at (3,0) {};
\node[draw,circle,cross,minimum size=0.3cm,label=above:${\delta_2-\epsilon_3}$](D) at (4.5,0) {};
\node(E) at (5.5,0) {$\cdots$};
\node(GG) at (5.5,-0.83){$\cdots$};
\node[draw,circle,cross,minimum size=0.3cm,label=above:${\epsilon_{n-1}-\delta_{n-1}}$](F) at (6.5,0) {};
\node[draw,circle,cross,minimum size=0.3cm,label=above:${\delta_{n-1}-\epsilon_{n}}$](G) at (8,0) {};
\node[draw,circle,cross,minimum size=0.3cm,label=above:${\epsilon_n-\delta_n}$](H) at (9.5,0) {};
\node[draw,circle,fill=black,minimum size=0.3cm,label=above:$\delta_n$](I) at (11,0) {};
\draw[-] (A) -- (B);
\draw[-] (B) -- (C);
\draw[-] (C) -- (D);
\draw[-] (D) -- (E);
\draw[-] (E) -- (F);
\draw[-] (F) -- (G);
\draw[-] (G) -- (H);
\draw[->, double, double distance=1pt] (H) -- (I);
\draw [decorate,decoration={brace,amplitude=5pt,mirror,raise=1.5ex}] (0,0) -- (1.5,0) node[midway,yshift=-2em]{\scriptsize{$\mathfrak{osp}(2|2)$}};
\draw [decorate,decoration={brace,amplitude=5pt,mirror,raise=1.5ex}] (3,0) -- (4.5,0) node[midway,yshift=-2em]{\scriptsize{$\mathfrak{osp}(2|2)$}};
\draw [decorate,decoration={brace,amplitude=5pt,mirror,raise=1.5ex}] (6.5,0) -- (8,0) node[midway,yshift=-2em]{\scriptsize{$\mathfrak{osp}(2|2)$}};
\draw [decorate,decoration={brace,amplitude=5pt,mirror,raise=1.5ex}] (9.5,0) -- (11,0) node[midway,yshift=-2em]{\scriptsize{$\mathfrak{osp}(3|2)$}};
  \end{tikzpicture}
\end{center}
Denote the odd simple roots in the above diagram by $\alpha_1, \cdots, \alpha_{2n}$, and make pairings of simple roots as above. Then one can analyze the common kernel of $(2i-1)$th and $2i$th screening operators as in Section \ref{sec: sl example}. It follows that
\begin{align}
  \ker \int e^{-\frac{1}{\sqrt{k}}\int \alpha_{2i-1}(Z)}dZ \cap \ker \int e^{-\frac{1}{\sqrt{k}}\int \alpha_{2i}(Z)}dZ &\simeq W^{\Psi}(\overline{\mathfrak{osp}(2|2)})\otimes \widehat{\pi}_{2i-1, 2i}^{\perp},\\
  \ker \int e^{-\frac{1}{\sqrt{k}}\int \alpha_{2n-1}(Z)}dZ \cap \ker \int e^{-\frac{1}{\sqrt{k}}\int \alpha_{2n}(Z)}dZ &\simeq W^{\Psi}(\overline{\mathfrak{osp}(3|2)})\otimes \widehat{\pi}_{2n-1, 2n}^{\perp},
\end{align}
for $i=1, \cdots, n-1$ if we use the same notation in Section \ref{sec: sl example}. Similarly, if we name the following odd simple roots of $\mathfrak{osp}(2n|2n)$ by $\beta_1, \cdots \beta_{2n}$,
\begin{center}
  \begin{tikzpicture}[cross/.style={path picture={ 
    \draw[black]
  (path picture bounding box.south east) -- (path picture bounding box.north west) (path picture bounding box.south west) -- (path picture bounding box.north east);
  }}, every label/.append style={font=\scriptsize}, >={Classical TikZ Rightarrow[]}]
\node(AA) at (-2,0) {$\mathfrak{osp}(2n|2n)$ :};
\node[draw,circle,cross,minimum size=0.3cm,label=above:${\delta_1-\epsilon_1}$](A) at (0,0) {};
\node[draw,circle,cross,minimum size=0.3cm,label=above:${\epsilon_1-\delta_2}$](B) at (2,0) {};
\node(C) at (3,0) {$\cdots$};
\node(GG) at (3,-0.83){$\cdots$};
\node[draw,circle,cross,minimum size=0.3cm,label=above:${\delta_{n-1}-\epsilon_{n-1}}$](D) at (4,0) {};
\node[draw,circle,cross,minimum size=0.3cm,label=above:${\epsilon_{n-1}-\delta_{n}}$](E) at (6,0) {};
\node[draw,circle,cross,minimum size=0.3cm,label=above:${\delta_n-\epsilon_n}$](F) at (8,1) {};
\node[draw,circle,cross,minimum size=0.3cm,label=below:$\epsilon_n+\delta_n$](G) at (8,-1) {};
\draw[-] (A) -- (B);
\draw[-] (B) -- (C);
\draw[-] (C) -- (D);
\draw[-] (D) -- (E);
\draw[-] (E) -- (F);
\draw[-] (F) -- (G);
\draw[-] (E) --(G);
\draw [decorate,decoration={brace,amplitude=5pt,mirror,raise=1.5ex}] (0,0) -- (2,0) node[midway,yshift=-2em]{\scriptsize{$\mathfrak{osp}(2|2)$}};
\draw [decorate,decoration={brace,amplitude=5pt,mirror,raise=1.5ex}] (4,0) -- (6,0) node[midway,yshift=-2em]{\scriptsize{$\mathfrak{osp}(2|2)$}};
\draw [decorate,decoration={brace,amplitude=5pt,raise=1.5ex}] (8,1) -- (8,-1) node[midway,right, xshift=1em]{\scriptsize{$\mathfrak{osp}(2|2)$}};
 \end{tikzpicture}
\end{center}
they make pairs two by two so that the kernel of screening operators for $\beta_{2i-1}$ and $\beta_{2i}$ can be written as a tensor product of $W^{\Psi}(\overline{\mathfrak{osp}(2|2)})$ and a Fock representation.

\subsubsection{$\mathfrak{osp}(2n-1|2n)$ and $\mathfrak{osp}(2n+2|2n)$} \label{sec: osp(2n-1|2n) osp(2n+2|2n) example} \hfill \\
The remainders are the types where the number of roots in the odd root system $\Pi$ is odd, so one cannot make pairings of simple roots.
First, consider $\mathfrak{osp}(2n-1|2n)$.
\begin{center}
  \begin{tikzpicture}[cross/.style={path picture={ 
    \draw[black]
  (path picture bounding box.south east) -- (path picture bounding box.north west) (path picture bounding box.south west) -- (path picture bounding box.north east);
  }}, every label/.append style={font=\scriptsize}, >={Classical TikZ Rightarrow[]}]
\node(AA) at (-2,0) {$\mathfrak{osp}(2n-1|2n)$ :};
\node[draw,circle,cross,minimum size=0.3cm,label=above:${\delta_1-\epsilon_1}$](A) at (0,0) {};
\node[draw,circle,cross,minimum size=0.3cm,label=above:${\epsilon_1-\delta_2}$](B) at (1.7,0) {};
\node[draw,circle,cross,minimum size=0.3cm,label=above:${\delta_2-\epsilon_2}$](C) at (3.4,0) {};
\node[draw,circle,cross,minimum size=0.3cm,label=above:${\epsilon_2-\delta_3}$](D) at (5.1,0) {};
\node(E) at (6.1,0) {$\cdots$};
\node(GG) at (6.1,-0.83){$\cdots$};
\node[draw,circle,cross,minimum size=0.3cm,label=above:${\delta_{n-1}-\epsilon_{n-1}}$](F) at (7.8,0) {};
\node[draw,circle,cross,minimum size=0.3cm,label=above:${\epsilon_{n-1}-\delta_{n}}$](G) at (9.5,0) {};
\node[draw,circle,fill=black,minimum size=0.3cm,label=above:$\delta_n$](H) at (11.2,0) {};
\draw[-] (A) -- (B);
\draw[-] (B) -- (C);
\draw[-] (C) -- (D);
\draw[-] (D) -- (E);
\draw[-] (E) -- (F);
\draw[-] (F) -- (G);
\draw[->, double, double distance=1pt] (G) -- (H);
\draw [decorate,decoration={brace,amplitude=5pt,mirror,raise=1.5ex}] (0,0) -- (1.7,0) node[midway,yshift=-2em]{\scriptsize{$\mathfrak{osp}(2|2)$}};
\draw [decorate,decoration={brace,amplitude=5pt,mirror,raise=1.5ex}] (3.4,0) -- (5.1,0) node[midway,yshift=-2em]{\scriptsize{$\mathfrak{osp}(2|2)$}};
\draw [decorate,decoration={brace,amplitude=5pt,mirror,raise=1.5ex}] (7.8,0) -- (9.5,0) node[midway,yshift=-2em]{\scriptsize{$\mathfrak{osp}(2|2)$}};
\draw [decorate,decoration={brace,amplitude=5pt,mirror,raise=1.5ex}] (11,0) -- (11.4,0) node[midway,yshift=-2em]{\scriptsize{$\mathfrak{osp}(1|2)$}};
  \end{tikzpicture}
\end{center}
Name the simple roots above by $\alpha_1, \cdots, \alpha_{2n-1}$ in sequence, and denote the corresponding spaces \eqref{eq: subspace u_alpha} by $\mathfrak{u}_{i}^{\perp}$'s. Except for the last one, one can make pairings of the simple roots and show the statements \eqref{eq: sl example analogous statement} for $\alpha_1, \cdots \alpha_{2n-2}$. For the last root, we have
\begin{equation} \label{eq: cartan decompose nonisotropic}
  \h= \mathfrak{u}_{2n-1}\oplus \mathfrak{u}_{2n-1}^{\perp}.
\end{equation}
for $\mathfrak{u}_{2n-1}:=\CC \alpha_{2n-1}$. It is worth emphasizing that the statement \eqref{eq: cartan decompose nonisotropic} holds since $\alpha_{2n-1}$ is nonisotropic. If it were, the subspace $\mathfrak{u}_{2n-1}$ is contained in its perpendicular space $\mathfrak{u}_{2n-1}^{\perp}$.

Let $\widehat{\pi}_{2n-1}$ and $\widehat{\pi}^{\perp}_{2n-1}$ be the SUSY Heisenberg vertex algebras for $\mathfrak{u}_{2n-1}$ and $\mathfrak{u}^{\perp}_{2n-1}$, respectively. Also, let $\Big(\widehat{\pi}_{-\frac{\alpha_{2n-1}}{\sqrt{k}}}\Big)_{2n-1}$ be the even Fock representation of $\widehat{\pi}_{2n-1}$ with a highest weight $-\frac{\alpha_{2n-1}}{\sqrt{k}}$. Then, the domain and the codomain of the screening operator $\int e^{-\frac{1}{\sqrt{k}}\int \alpha_{2n-1}(Z)}dZ$ are decomposed as
\begin{equation*}
  \widehat{\pi}\simeq \widehat{\pi}_{2n-1} \otimes \widehat{\pi}_{2n-1}^{\perp}, \quad \widehat{\pi}_{-\frac{\alpha_{2n-1}}{\sqrt{k}}}\simeq \Big(\widehat{\pi}_{-\frac{\alpha_{2n-1}}{\sqrt{k}}}\Big)_{2n-1} \otimes \widehat{\pi}_{2n-1}^{\perp},
\end{equation*}
and it acts as zero on the space $\widehat{\pi}_{2n-1}^{\perp}$ by Lemma \ref{lem: perp decomposition}. Thus, the kernel of the last screening operator is isomorphic to
\begin{equation}
  \ker \int e^{-\frac{1}{\sqrt{k}}\int \alpha_{2n-1}(Z)}dZ \simeq W^{\Psi}(\overline{\mathfrak{osp}(1|2)})\otimes \widehat{\pi}_{2n-1}^{\perp}.
\end{equation}
In the case of $\mathfrak{osp}(2n+2|2n)$, the number of odd simple roots is odd, but one cannot proceed with a similar statement for the last root since all of the simple roots are isotropic. In this case, consider the last three simple roots as a group.
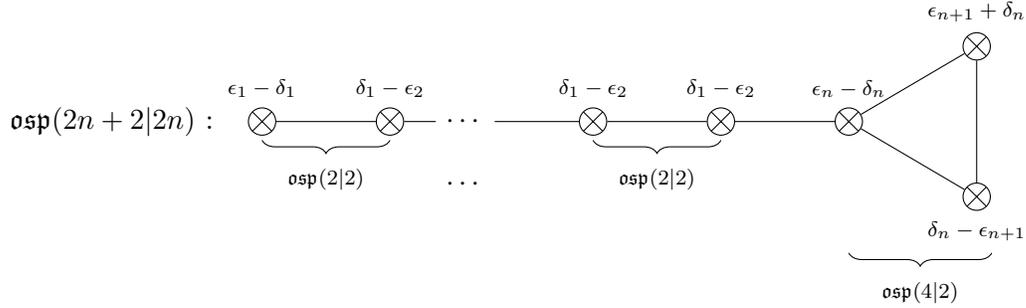
\begin{center}
  \begin{tikzpicture}[cross/.style={path picture={ 
    \draw[black]
  (path picture bounding box.south east) -- (path picture bounding box.north west) (path picture bounding box.south west) -- (path picture bounding box.north east);
  }}, every label/.append style={font=\scriptsize}, >={Classical TikZ Rightarrow[]}]
\node(AA) at (-2,0) {$\mathfrak{osp}(2n+2|2n)$ :};
\node[draw,circle,cross,minimum size=0.3cm,label=above:${\epsilon_1-\delta_1}$](A) at (0,0) {};
\node[draw,circle,cross,minimum size=0.3cm,label=above:${\delta_1-\epsilon_2}$](B) at (1.7,0) {};
\node(C) at (2.7,0) {$\cdots$};
\node(GG) at (2.7,-0.83){$\cdots$};
 \node[draw,circle,cross,minimum size=0.3cm,label=above:${\epsilon_{n-1}-\delta_{n-1}}$](D) at (4.4,0) {};
  \node[draw,circle,cross,minimum size=0.3cm,label=above:${\delta_{n-1}-\epsilon_n}$](E) at (6.1,0) {};
\node[draw,circle,cross,minimum size=0.3cm,label=above:${\epsilon_n-\delta_n}$](F) at (7.8,0) {};
\node[draw,circle,cross,minimum size=0.3cm,label=below:${\delta_n-\epsilon_{n+1}}$](G) at (9.5,-1) {};
\node[draw,circle,cross,minimum size=0.3cm,label=above:$\epsilon_{n+1}+\delta_n$](H) at (9.5,1) {};
\draw[-] (A) -- (B);
\draw[-] (B) -- (C);
\draw[-] (C) -- (D);
\draw[-] (D) -- (E);
\draw[-] (E) -- (F);
\draw[-] (F) -- (G);
\draw[-] (G) -- (H);
\draw[-] (F) -- (H);
\draw [decorate,decoration={brace,amplitude=5pt,mirror,raise=1.5ex}] (0,0) -- (1.7,0) node[midway,yshift=-2em]{\scriptsize{$\mathfrak{osp}(2|2)$}};
\draw [decorate,decoration={brace,amplitude=5pt,mirror,raise=1.5ex}] (4.4,0) -- (6.1,0) node[midway,yshift=-2em]{\scriptsize{$\mathfrak{osp}(2|2)$}};
\draw [decorate,decoration={brace,amplitude=5pt,mirror,raise=1.5ex}] (7.8,-1.5) -- (9.7,-1.5) node[midway,yshift=-2em]{\scriptsize{$\mathfrak{osp}(4|2)$}};
 \end{tikzpicture}
\end{center}
Hence, for the simple roots $\beta_1, \cdots \beta_{2n+1}$ in the above diagram, the analogous statement \eqref{eq: sl example analogous statement} holds for the first $2n-2$ roots, and
\begin{equation}
  \bigcap_{j=2n-1}^{2n+1} \ker \int e^{-\frac{1}{\sqrt{k}}\int \beta_{j}(Z)}dZ\simeq W^{\Psi}(\overline{\mathfrak{osp}(4|2)})\otimes \widehat{\pi}_{2n-1, 2n, 2n+1}^{\perp},
\end{equation}
where $\widehat{\pi}_{2n-1, 2n, 2n+1}^{\perp}$ is the even Fock representation corresponding to $\cap_{j=2n-1}^{2n+1}\mathfrak{u}_{j}^{\perp}$.

\begin{remark}
  As in the $\mathfrak{osp}(4|2)$ case, the principal SUSY W-algebra corresponding to $D(2,1;\alpha)$ cannot be split up anymore. It is because the odd simple root system of $D(2,1;\alpha)$ draws the same Dynkin diagram as in the case of $\mathfrak{osp}(4|2)$, but with constant labels on the line segments between the nodes.
\end{remark}


\appendix
\section{Appendices}
\subsection{Superfield formalism} \label{appendix: superfield formalism} \hfill \\
In this section, we recall another definition of SUSY vertex algebra in terms of a superfield formalism in various physics papers \cite{SS74, FZW74}. For more details, we refer to the first definition of $N_K=1$ SUSY vertex algebra in \cite{HK07}.  Also, we provide here an alternative proof for Theorem \ref{thm: SUSY VA vs. VA} with this definition.
\subsubsection{Supersymmetric vertex algebras via superfield formalism} \hfill \\
We use the tuple $\nabla=(\partial, D)$ of endomorphisms introduced in $\eqref{eq: nabla}$ and the notations for the super variable $Z=(z, \theta)$ from \eqref{eq: notation tuple power} to \eqref{eq: super residue}. In addition, we denote by $\CC[\![Z,Z^{-1}]\!]:=\CC[\![z,z^{-1}]\!]\otimes \CC[\theta]$ the space of formal power series for the super variable $Z$. For any vector superspace $V$, if an element $A(Z)\in \textup{End}_{\CC}(V)\otimes \CC[\![Z,Z^{-1}]\!]$ satisfies
\begin{equation}
  A(Z)v\in V\otimes \CC(\!(Z)\!)
\end{equation}
for any $v\in V$, we call $A(Z)$ an \textit{$\textup{End}_{\CC}(V)$-valued superfield}.
\begin{definition} \label{def: appendix SUSY VA}
  A supersymmetric(SUSY) vertex algebra is a vector superspace $V$ with an even vector $\vac$, a tuple of endomorphisms $\nabla=(\partial, D)$ and a parity preserving linear map
  \begin{equation} \label{eq: superfield}
    \begin{aligned}
    Y(\,\cdot\,, Z):V &\rightarrow \{\textup{End}_{\CC}(V)\text{-valued superfields}\}\\
    A &\mapsto Y(A,Z)=\sum_{\substack{n\in \ZZ\\i=0,1}}Z^{-n-1|1-i}A_{(n|i)} 
    \end{aligned}
  \end{equation}
satisfying the following axioms:
\begin{enumerate}[]
  \item (Vacuum axiom) the superfield corresponding to $\vac$ is $Y(\vac, Z)=\textup{Id}_{V}$ and 
  \begin{equation*}
    Y(A, Z)\!\vac \in \CC[\![Z]\!]\otimes V, \quad Y(A,Z)\!\vac\rvert_{Z=0}=A,
  \end{equation*}
  \item (Translation invariance) the commutator of $D$ and the superfield $Y(A,Z)$ is given by 
  \begin{equation}\label{eq: translation invariance}
    [D, Y(A, Z)]=(\partial_{\theta}-\theta\partial_z)Y(A, Z),
  \end{equation}
  \item (Locality) there exists $n\in \ZZ_{\geq 0}$ such that 
  \begin{equation}
    (z-w)^n[Y(A,Z),Y(B,W)]=0,
  \end{equation} 
\end{enumerate}
for $A,B\in V.$
\end{definition}
Note that the differentials $\partial_z$ and $\partial_{\theta}$ in \eqref{eq: translation invariance} are defined in an obvious way. To be precise, $\partial_z$ is the usual differential with respect to $z$, while $\partial_{\theta}$ is the differential defined by
\begin{equation}
  \partial_{\theta}\Big(\sum_{\substack{n\in \ZZ\\ i=0,1}}Z^{-n-1|1-i}A_{(n|i)}\Big)=\sum_{n\in \ZZ}Z^{-n-1|0}A_{(n|0)}.
\end{equation}
 For any elements $A$ and $B$ in a SUSY vertex algebra $V$, the \textit{normally ordered product} of the two elements is defined by
\begin{equation} \label{eq: appendix normally ordered product definition}
  :\!AB\!:\, =A_{(-1|1)}B.
\end{equation}
We often call it the $(-1|1)$th product of $A$ and $B$. On the other hand, the \textit{normally ordered product} of the corresponding superfields $A(Z):=Y(A,Z)$ and $B(z):=Y(B,Z)$ is defined by
\begin{equation} \label{eq: appendix normally ordered superfields}
  :\!A(Z)B(Z)\!:\, =A(Z)_+B(Z)+(-1)^{p(A)p(B)}B(Z)A(Z)_-,
\end{equation}
where $A(Z)_+$ and $A(Z)_-$ are the creation and annihilation parts of $A(Z)$, that is 
\begin{equation*}
  A(Z)_+=\sum_{\substack{n<0\\i=0,1}}Z^{-n-1|1-i}A_{(n|i)}, \quad A(Z)_-=\sum_{\substack{n\geq 0\\i=0,1}}Z^{-n-1|1-i}A_{(n|i)}.
\end{equation*}
Note here that the normally ordered product of two superfields is also a superfield.

The $(n|i)$th products for positive $n$ of a SUSY vertex algebra can be understood by so-called the SUSY version of the decomposition lemma. To describe the lemma, let us introduce the SUSY delta distribution 
\begin{equation} \label{eq: delta distribution}
  \delta(Z,W)=(\theta-\eta)\sum_{n\in \ZZ}z^n w^{-n-1},
\end{equation}
where $Z=(z, \theta)$ and $W=(w, \eta)$ are super variables. In addition, a pair $(a(Z), b(Z))$ of two $\textup{End}_{\CC}(V)$-valued superfields is called a \textit{local pair} if $(z-w)^n[a(Z),b(W)]=0$
for some $n\in \ZZ_{\geq 0}$. 

\begin{lemma}[Decomposition Lemma \cite{HK07}] \label{lem: decomposition lemma}
  Let $\big(a(Z), b(Z)\big)$ be a local pair. Then the commutator of the two superfields is written as a finite sum
  \begin{equation} \label{eq: decomposition lemma}
    [a(Z),b(W)]=\sum_{\substack{n\in \ZZ_{\geq 0}\\i=0,1}}D_W^{(n|i)}\delta(Z,W)\ C_{n|i}(W),
  \end{equation}
  where $C_{n|i}(W)=\int (Z-W)^{n|i}[a(Z),b(W)] dZ$ and $D_W^{(n|i)}=(-1)^i \frac{1}{n!}\partial_w^n (\partial_\eta+\eta\partial_w)^i.$ \qed
\end{lemma}
For a SUSY vertex algebra $V$, the locality axiom in Definition \ref{def: appendix SUSY VA} means that any pair of two superfields $(Y(A, Z), Y(B, Z))$ for $A, B\in V$ is local. Hence, as a result of Lemma \ref{lem: decomposition lemma}, we get the decomposition of the commutator as
\begin{equation} \label{eq: decomposition lemma applied}
  [Y(A,Z),Y(B,W)]=\sum_{\substack{n\in \ZZ_{\geq 0}\\i=0,1}}D_W^{(n|i)}\delta(Z,W)\ C_{n|i}(W)
\end{equation}
for any $A,B$ in $V$.
We conclude this section by stating the key properties of the superfields in a SUSY vertex algebra.

\begin{proposition}[\cite{HK07}] \label{prop: superfield property}
  Let $V$ be a SUSY vertex algebra. For any $A,B\in V$ and the corresponding superfields $Y(A,Z)$ and $Y(B, Z)$, the following statements hold:
  \begin{enumerate}
    \item $Y(:\!AB\!:,Z)=\, :\!Y(A,Z)Y(B,Z)\!:$,
    \item $Y(DA,Z)=(\partial_{\theta}+\theta \partial_z)Y(A,Z)$,
    \item $C_{n|i}(W)=Y(A_{(n|i)}B,W)$ for $C_{n|i}(W)$ in \eqref{eq: decomposition lemma applied}. \qed
  \end{enumerate}
\end{proposition}
It is worth noting that in (2) of Proposition \ref{prop: superfield property}, the differential on the RHS is $\partial_{\theta}+\theta \partial_z$, not the differential $\partial_{\theta}-\theta \partial_z$ in the translation invariance axiom \eqref{eq: translation invariance}. However, both of the two differentials are square roots of $\partial_z$, i.e.,
\begin{equation} \label{eq: appendix D_Z square}
  (\partial_{\theta}+\theta \partial_z)^2=\partial_z=(\partial_{\theta}-\theta \partial_z)^2.
\end{equation}


\begin{remark} \label{rmk: SUSY VA defs equivalence}
  If we take the supersymmetric formal Fourier transformation to the commutator \eqref{eq: decomposition lemma applied} of the superfields, we get the $\Lambda$-bracket $[A{}_{\Lambda}B]$ in Definition \ref{def: SUSY VA}. Furthermore, the coefficients in \eqref{eq: Lambda bracket coefficients} coincide with the $(n|i)$th products in \eqref{eq: superfield}. We refer to \cite{HK07} for the proof of the equivalence of the two definitions of SUSY vertex algebras.
\end{remark}

\subsubsection{Another proof of Theorem \ref{thm: SUSY VA vs. VA}} \hfill \\
We choose Definition \ref{def: appendix SUSY VA} for SUSY vertex algebras, and prove Theorem \ref{thm: SUSY VA vs. VA} using the superfields. Let $V$ be a SUSY vertex algebra with a map $Y(\cdot, Z)$. Erasing the $\theta$-containing part of each superfield, define $Y(\cdot, z): V\rightarrow \{\textup{End}_{\CC}(V)\textup{-valued fields}\}$ by
\begin{equation} \label{eq: appendix VA structure of SUSY VA}
  \begin{aligned}
  Y(\,\cdot\,,z): V&\rightarrow \CC[\![z,z^{-1}]\!]\otimes \textup{End}_{\CC}(V),\\ A&\mapsto Y(A,z)=\sum_{n\in \ZZ}z^{-n-1}A_{(n)}:=\sum_{n\in \ZZ}z^{-n-1}A_{(n|1)}
  \end{aligned}
\end{equation}
for $A_{(n|1)}$'s in \eqref{eq: superfield}. The vacuum and locality axioms for vertex algebras clearly hold by the definitions and thus it is enough to check the translation invariance. Note that the translation invariance \eqref{eq: translation invariance} for SUSY vertex algebra and the property (2) of Proposition \ref{prop: superfield property} imply that
\begin{equation} \label{eq: commutator nth product}
  [D, A_{(n|1)}]=A_{(n|0)}=(DA)_{(n|1)}, \quad [D, A_{(n|0)}]=-n A_{(n-1|1)}=(DA)_{(n|0)} 
\end{equation}
for $A\in V$ and $n\in \ZZ$. Due to the relations in \eqref{eq: commutator nth product}, we get
\begin{equation*}
  [\partial, A_{(n)}]=\frac{1}{2}[[D,D],A_{(n|1)}]=[D,[D,A_{(n|1)}]]=[D, A_{(n|0)}]=-nA_{(n-1|1)}=-nA_{(n-1)},
\end{equation*}
 which is equivalent to the translation invariance for vertex algebras. Thus, $V$ is a vertex algebra.

  Conversely, assume that $V$ is a vertex algebra with a state-field correspondence $Y(\cdot, z)$, and is equipped with an odd derivation $D$ satisfying $D^2=\partial$. Since the odd derivation $D$ satisfies $[D, A_{(n)}]=(DA)_{(n)}$ for any $A\in V$ and $n\in \ZZ$, we get
  \begin{equation} \label{eq: D commutator field version}
    [D,Y(A,z)]=Y(DA,z).
  \end{equation}
  Define $Y(\cdot,Z)$ to be the parity-preserving linear map
  \begin{equation} \label{eq: induced susy state-field}
    \begin{aligned}
    Y(\,\cdot\,,Z):V&\rightarrow \CC[\![Z,Z^{-1}]\!]\otimes \textup{End}_{\CC}(V)\\
    A&\mapsto Y(A,z)+\theta Y(DA,z).
    \end{aligned}
  \end{equation}
Then the range of the map \eqref{eq: induced susy state-field} is automatically contained in the space of $\textup{End}_{\CC}(V)$-valued superfields, and the superfields satisfy the vacuum and locality axioms for $V$ to be a SUSY vertex algebra. The translation invariance \eqref{eq: translation invariance} can be checked using \eqref{eq: D commutator field version} as follows:
\begin{equation}\label{eq:SUSY vs va}
  \begin{aligned}
    [D, Y(A,Z)]&=[D,Y(A,z)]-\theta[D,Y(DA,z)]\\
    &=Y(DA,z)-\theta Y(\partial A, z)\\
    &=Y(DA,z)-\theta \partial_z Y(A, z)=(\partial_{\theta}-\theta \partial_z)Y(A,Z).
  \end{aligned}
\end{equation}
  In \eqref{eq:SUSY vs va}, the first equality follows from $[D, \theta]=0$, and the third equality is a result of the translation invariance of the vertex algebra. We remark here that the choices of the maps \eqref{eq: appendix VA structure of SUSY VA} and \eqref{eq: induced susy state-field} match with the choices of $\Lambda$ and $\lambda$-brackets in the proof of Theorem \ref{thm: SUSY VA vs. VA} in Section \ref{sec: SUSY VA vs. VA}.

\subsection{Proof of Proposition \ref{prop: the first total complex}} \label{appendix: proof total complex} \hfill \\
We give a proof of Proposition \ref{prop: the first total complex} in this section, using the notations in Section \ref{sec: the first total complex}. For our convenience, denote $d^{\infty}={d_{\textup{st}}}_{(0|0)}$. As we have seen in Section \ref{sec: classical SUSY W-algebra}, the space $\widetilde{C}^{\infty}$ consists of polynomials with variables in
\begin{gather*}
  T:=\{\partial^j J_{\bar{u}^{\alpha}},\, \partial^j D J_{\bar{u}^{\alpha}}\}_{\alpha\in I_+, j\geq 0} \cup\{\partial^j J_{\bar{a}_m},\, \partial^j D J_{\bar{a}_m}\}_{m\in M, j\geq 0}\cup \{\partial^j \phi^{\alpha},\, \partial^j D\phi^{\alpha}\}_{\alpha\in I_+, j\geq 0},
\end{gather*}
where $\{a_m\}_{n\in M}$ is a basis of $\g_0$. For a nonnegative integer $p$, let us consider the filtration 
\begin{equation} \label{eq: proof polynomial filtration}
  F_p \widetilde{C}^{\infty}=\textup{Span}_{\CC}\{X_1 X_2\cdots X_r\,|\, X_i\in T, r\geq p\}\subset \widetilde{C}^{\infty}.
\end{equation}
On the space $\widetilde{C}^{\infty}$, the action of $d^{\infty}$ is induced from \eqref{eq: classical SUSY differential}. To be explicit,
\begin{equation}\label{eq: proof classical SUSY differential}
  \begin{aligned}
      \{d_{\textup{st}}{}_{\Lambda}{J}_{\bar{a}}\}=&\sum_{\beta\in I_+}(-1)^{(p(a)+1)p(\beta)}{\phi}^{\beta} {J}_{\overline{\pi_{\leq 0}[u_{\beta},a]}}-\sum_{\beta\in I_+}(-1)^{p(\beta)}(D+\chi)(u_{\beta}|a){\phi}^{\beta},\\
      \{d_{\textup{st}}{}_{\Lambda}{\phi}^{\alpha}\}=&\frac{1}{2} \sum_{\beta,\gamma\in I_+}(-1)^{(p(\alpha)+1)p(\beta)}([u_{\beta},u^{\alpha}]|u_{\gamma}){\phi}^{\beta}{\phi}^{\gamma}
    \end{aligned}
  \end{equation}
and $d^{\infty}$ picks up the 
 terms without $\lambda$ or $\chi$. One can see from \eqref{eq: proof classical SUSY differential} that $d^{\infty}$ increases the degree by at most $1$. Hence, concerning the spectral sequence induced from the filtration \eqref{eq: proof polynomial filtration}, the differential on the $r$th page is $0$ for $r\geq 2$. It implies that the induced total complex of the spectral sequence made up from this filtration is convergent to the original cohomology $E_1^{\infty}=H(\widetilde{C}^{\infty}, d^{\infty})$. Denote the differential on the $0$th page by $d^{\infty}_0$. Then for the elements in the set $T$, the images by $d_0^\infty$ are
\begin{equation} \label{eq: proof d_0, C_1, C_2}
  \begin{aligned}
    d^{\infty}_0\, (\partial^j J_{\bar{u}^{\alpha}})= - \partial^jD\phi^{\alpha},\quad  d^{\infty}_0\, (\partial^j D J_{\bar{u}^{\alpha}})= \partial^{j+1}\phi^{\alpha}
  \end{aligned}
\end{equation}
for $j\in \mathbb{Z}_{\geq 0}$ and $\alpha\in I_+$, while the other elements are killed by $d^{\infty}_0$. The equalities \eqref{eq: proof d_0, C_1, C_2} can be induced from \eqref{eq: proof classical SUSY differential} by erasing the polynomials of degree$\geq 2$ on the RHS. 

Inspired by \eqref{eq: proof d_0, C_1, C_2}, we decompose the complex $\widetilde{C}^{\infty}$ into four parts, where one part does not affect the other parts in the process of taking the cohomology with $d^{\infty}_0$. For each $\alpha\in I_+$ and $j\geq 0$, let $\mathcal{V}(\bar{\g}_0)$, $\widetilde{C}_+$, $\widetilde{C}_{\alpha, j}$ and $D\widetilde{C}_{\alpha, j}$ be the subspaces of $\widetilde{C}^{\infty}$ given by the spaces of polynomials with variables in
\begin{equation}\label{eq: proof decompose to use kunneth}
\begin{gathered}
  \{\partial^i J_{\bar{a}_m}, \partial^i DJ_{\bar{a}_m}\,|\,m\in M, i\geq 0\},\ \{\phi^{\beta}\,|\,\beta\in I_+\},\\
   \{\partial^j J_{\bar{u}^{\alpha}},\partial^j D\phi^{\alpha}\},\ \{\partial^j DJ_{\bar{u}^{\alpha}},\partial^{j+1}\phi^{\alpha}\},
\end{gathered}
\end{equation}
respectively. Then, the complex $\widetilde{C}^{\infty}$ decomposes as
\begin{equation*}
  \widetilde{C}^{\infty}=\mathcal{V}(\bar{\g}_0) \otimes \widetilde{C}_+ \otimes \bigotimes_{\substack{\alpha\in I_+\\ j\geq 0}}\widetilde{C}_{\alpha,j}\otimes \bigotimes_{\substack{\alpha\in I_+\\ j\geq 0}}D\widetilde{C}_{\alpha,j}.
\end{equation*}
Now, take the cohomology with the use of the differential $d^{\infty}_0$ to get the first total complex $\widetilde{A}_1^{\infty}$ of the spectral sequence. As a result of the K\"{u}nneth formula, we have 
\begin{equation} \label{eq: proof kunneth applied}
  \widetilde{A}_1^{\infty}=\mathcal{V}(\bar{\g}_0)\otimes \widetilde{C}_+ \otimes \bigotimes_{\substack{\alpha\in I_+\\ j\geq 0}}H(\widetilde{C}_{\alpha,j}, d^{\infty}_0) \otimes \bigotimes_{\substack{\alpha\in I_+\\ j\geq 0}}H(D\widetilde{C}_{\alpha,j}, d^{\infty}_0).
\end{equation}
Moreover, the last two cohomologies in \eqref{eq: proof kunneth applied} are trivial due to the following lemma.
\begin{lemma} \label{lem: proof trivial cohomology}
  For $\alpha\in I_+$ and $j\geq 0$, let $\widetilde{C}_{\alpha,j}$ and $D\widetilde{C}_{\alpha, j}$ be the subspaces of $\widetilde{C}$ defined in \eqref{eq: proof decompose to use kunneth}. For the differential $d^{\infty}_0$ on $\widetilde{C}^{\infty}$ given by \eqref{eq: proof d_0, C_1, C_2}, each subspace is closed under $d^{\infty}_0$. Furthermore, cohomologies with respect to $d^{\infty}_0$ are
  \begin{equation} \label{eq: rom 1,2 cohomology}
    H(\widetilde{C}_{\alpha, j},d^{\infty}_0)=\CC, \quad H(D\widetilde{C}_{\alpha, j}, d^{\infty}_0)=\CC.
  \end{equation}
\end{lemma}
\begin{proof}
First, we prove the statement for the subspace $\widetilde{C}_{\alpha, j}$. Since the space is supercommutative, $X^2=0$ for any odd element $X$ in $\widetilde{C}_{\alpha, j}$. Note that the parities of $\partial^j J_{\bar{u}^{\alpha}}$ and $\partial^j D \phi^{\alpha}$ are different. Therefore, the space $\widetilde{C}_{\alpha, j}$ is described as
\begin{equation} \label{eq: tild C elements form}
  \widetilde{C}_{\alpha, j}=\left\{
\begin{array}{ll}
  \textup{Span}_{\CC}\left\{\left(\partial^{j}D\phi^{\alpha}\right)^r, \partial^jJ_{\bar{u}^{\alpha}}\left(\partial^{j}D\phi^{\alpha}\right)^r\,\big\vert\,r\in \ZZ_{\geq 0}\right\} &\text{ if }p(\alpha)=0,\\
  \textup{Span}_{\CC}\left\{\left(\partial^{j} J_{\bar{u}^{\alpha}}\right)^r, \left(\partial^{j} J_{\bar{u}^{\alpha}}\right)^r \partial^{j}D\phi^{\alpha}\,\big\vert\, r\in \ZZ_{\geq 0}\right\} &\text{ if }p(\alpha)=1.
\end{array}\right.
\end{equation}
Notice that we have an explicit action $\eqref{eq: proof d_0, C_1, C_2}$ of $d^{\infty}_0$ on the degree $1$ polynomial, and $d^{\infty}_0$ acts as a differential on the polynomial whose degree is bigger than $1$. Thus, when $p(\alpha)=0$,
\begin{equation*}
  d^{\infty}_0\left(\left(\partial^{j} D \phi^{\alpha}\right)^r\right)=0,\quad  d^{\infty}_0 \left(\partial^j J_{\bar{u}^{\alpha}}\left(\partial^{j} D \phi^{\alpha}\right)^r\right)=-\left(\partial^{j} D \phi^{\alpha}\right)^{r+1}
\end{equation*}
for each $r\geq 0$, which implies that $H(\widetilde{C}^{\alpha, j},d^{\infty}_0)=\CC$. Similarly, when $p(\alpha)=1$,
\begin{equation*}
  d^{\infty}_0\left(\left(\partial^{j+1} J_{\bar{u}^{\alpha}}\right)^r \partial^{j}D\phi^{\alpha}\right)=0,\quad  d^{\infty}_0 \left(\left(\partial^j J_{\bar{u}^{\alpha}}\right)^r\right)=-r\left(\partial^j J_{\bar{u}^{\alpha}}\right)^{r-1}\partial^j D \phi^{\alpha}
\end{equation*}
for each $r\geq 0$, which also leads to the conclusion $H(\widetilde{C}_{\alpha, j},d^{\infty}_0)=\CC$. In case of $D\widetilde{C}_{\alpha, j}$, one can show the statement analogously with the description of $D\widetilde{C}_{\alpha, j}$ given by
\begin{equation}
  D\widetilde{C}_{\alpha, j}=\left\{
\begin{array}{ll}
  \textup{Span}_{\CC}\left\{\left(\partial^{j} DJ_{\bar{u}^{\alpha}}\right)^r, \left(\partial^{j} DJ_{\bar{u}^{\alpha}}\right)^r \partial^{j+1}\phi^{\alpha}\,\big\vert\, r\in \ZZ_{\geq 0}\right\} &\text{ if }p(\alpha)=0,\\
  \textup{Span}_{\CC}\left\{\left(\partial^{j+1}\phi^{\alpha}\right)^r, \partial^jDJ_{\bar{u}^{\alpha}}\left(\partial^{j+1}\phi^{\alpha}\right)^r\,\big\vert\, r\in \ZZ_{\geq 0}\right\} &\text{ if }p(\alpha)=1.
\end{array}\right.
\end{equation}
\end{proof}
Combining \eqref{eq: proof kunneth applied} with Lemma \ref{lem: proof trivial cohomology}, we have
\begin{equation} \label{eq: proof tild A_1}
  \widetilde{A}^{\infty}_1\simeq \mathcal{V}(\bar{\g}_0)\otimes \widetilde{C}_+
\end{equation}
for the subspace $\mathcal{V}(\bar{\g}_0)$ and $\widetilde{C}_+$ of $\widetilde{C}$ in \eqref{eq: proof decompose to use kunneth}. Denote the differential on the first page of the spectral sequence by $d^{\infty}_1$, then
\begin{equation}
  E^{\infty}_1=H(\widetilde{A}^{\infty}_1,d^{\infty}_1).
\end{equation}
Since we only need the description of the subspace $E^{\infty}_1(\leq 1)$ of charge at most $1$, it is enough to consider the subspace $\widetilde{A}^{\infty}_1(\leq 1 )$ of charge at most $1$ in \eqref{eq: proof tild A_1}. As a direct consequence of \eqref{eq: proof classical SUSY differential} and Lemma \ref{lem: conditions with Lie brackets}, we get that $d_1^{\infty}(J_{\bar{a}})=0$ for any $a\in \g_0$, and $d_1^{\infty}(\phi^{\alpha})=0$ if and only if $\alpha$ is indecomposable in $I_+$. Therefore, we get the desired statement as follows:
\begin{equation*}
  E^{\infty}_1(\leq 1)\simeq H(\widetilde{A}^{\infty}_1, d^{\infty}_1)\simeq\mathcal{V}(\bar{\g}_0) \oplus \Big(\mathcal{V}(\bar{\g}_0)\otimes \bigoplus_{\alpha\in I_0}\phi^{\alpha}\Big).
\end{equation*}

\subsection{Proof of Theorem \ref{thm: main-screening}} \label{appendix: proof screening operator} \hfill \\
Let us use the notations in Section \ref{sec: free field realization} and Appendix \ref{appendix: superfield formalism} to prove Theorem \ref{thm: main-screening}. As in Section \ref{bigsec: FF realization of SUSY W-algebras}, all the statements are for generic $k\in \CC$. Recall that we assume $\g_0=\h$ and we do not distinguish the elements under the isomorphisms \eqref{eq: heisenberg and fock isomorphism}. By Corollary \ref{cor: FF realization of susy W-alg}, we can restrict our attention to $\left|-\frac{\alpha}{\nu}\right>_{(0|0)}$ for each simple root $\alpha$. To prove the theorem, it is enough to show that the formula of the superfield $Y(\left|-\frac{\alpha}{\nu} \right>, Z)$ in $E^k_1$ is equal to the exponential form \eqref{eq: exponential definition}. Denote the coroot of $\alpha\in \Pi$ by $h_{\alpha}$. Then $[u_{\alpha}, u^{\alpha}]=(u_{\alpha}|u^{\alpha})h_{\alpha}=(-1)^{p(\alpha)}h_{\alpha}$, which induces
\begin{equation} \label{eq: proof making equation}
  {d_{\textup{st}}}_{(0|0)}J_{\bar{u}^{\alpha}}=(-1)^{p(\alpha)}:\!\phi^{\alpha}J_{\bar{h}_{\alpha}}\!:-(k+h^{\vee}) D\phi^{\alpha}.
\end{equation}
From the description of $E^k_1(\leq 1)$ in Theorem \ref{thm: module structure of the first total complex}, the LHS of \eqref{eq: proof making equation} should be $0$ in $E_1^k(\leq 1)$. Therefore, in $E_1^k{(\leq 1)}$,
\begin{equation}
  D\phi^{\alpha}=(-1)^{p(\alpha)}\frac{1}{k+h^{\vee}}:\!\phi^{\alpha}J_{\bar{h}_{\alpha}}\!:=-\frac{1}{k+h^{\vee}}:\!J_{\bar{h}_{\alpha}}\phi^{\alpha}\!:.
\end{equation}
Equivalently, $D\left|-\tfrac{\alpha}{\nu}\right>=-\frac{1}{\nu}:\!\alpha \left|-\tfrac{\alpha}{\nu}\right>\!:$ for $\alpha=(\bar{h}_{\alpha})_{(-1)}\vac$. Hence, the superfield $\superz$ satisfies the equation
\begin{equation} \label{eq: superfield equation}
  (\partial_{\theta}+\theta \partial_z)\superz =-\frac{1}{\nu}:\!\alpha(Z) \superz\!:
\end{equation}
for $\alpha(Z)=Y(\alpha, Z)$, due to Proposition \ref{prop: superfield property}. Also, from the vacuum axiom of SUSY vertex algebras, and the $\widehat{\pi}$-module structure of $E^k_1(\leq 1)$, we can deduce further relations. For each $\beta\in \h^*$, let $\beta\in \widehat{\pi}$ be an element given by $\beta=(\bar{h}_{\beta})_{(-1)}\vac$ for $h_{\beta}\in \h$ satisfying $(h_{\beta}|h)=\beta(h)$ for any $h\in \h$. Then the relations are stated as
\begin{gather}
  \superz \vac \rvert_{Z=0}=\left|-\tfrac{\alpha}{\nu}\right>, \label{eq: superfield equation vacuum}\\
  \big[\beta(Z), \superw \big]=-\frac{1}{\nu}\,(\alpha|\beta)\,\delta(Z,W)\,\superw \label{eq: superfield equation commutator}
\end{gather}
for any $\beta\in \h^*$ and $\beta(Z)=Y(\beta,Z)$. 
 Now we can show that $\superz=\expo$ by proving it is the solution to the equations \eqref{eq: superfield equation}, \eqref{eq: superfield equation vacuum} and \eqref{eq: superfield equation commutator}. Hence, by the following lemmas, the theorem follows.

\begin{lemma} \label{lem: appendix screening condition}
  For a simple root $\alpha\in\Pi$, recall the formal series $\expo$ in \eqref{eq: exponential definition}. This series satisfies the equality
  \begin{equation*}
    [\beta(Z),\expo]=-\frac{1}{\nu}\,(\alpha|\beta)\,\delta(Z,W)\,\expo
  \end{equation*}
  for any $\beta\in \h^*$.
\end{lemma}
\begin{proof}
  For any $\beta\in \h^*$, the corresponding element $\beta\in \widehat{\pi}$ satisfies $[\beta {}_{\Lambda}\alpha]=(\alpha|\beta)\chi$. By Remark \ref{rmk: SUSY VA defs equivalence}, it can be translated to the commutator of the superfields as
  \begin{equation} \label{eq: lemma proof commutator of superfields}
    [\beta(Z), \alpha(W)]=-(\partial_{\eta}+\eta \partial_w)\delta(Z,W)(\alpha|\beta).
  \end{equation}
  Comparing the coefficients of \eqref{eq: lemma proof commutator of superfields}, we get
  \begin{equation} \label{eq: lemma proof commutator coefficients}
    [\beta_{(m|0)}, \alpha_{(n|0)}]=\delta_{m+n,0}m (\alpha|\beta), \quad [\beta_{(m|1)}, \alpha_{(n|1)}]=\delta_{m+n, -1}(\alpha|\beta)
  \end{equation}
  and $[\beta_{(m|i)},\alpha_{(n|j)}]=0$ if $i\neq j$, where the $(m|i)$th products are introduced in \eqref{eq: superfield}. Therefore, by applying the result of \eqref{eq: lemma proof commutator coefficients}, we have
  \begin{equation}
    \begin{aligned}
    &\Big[\beta(Z),\ \frac{1}{\nu}\sum_{j\neq 0}\frac{W^{-n|0}}{n} \alpha_{(n|0)}-\frac{1}{\nu}\sum_{j\in \ZZ}W^{-j-1|1}\alpha_{(j|1)}\Big]\\
    &=-\frac{1}{\nu}\sum_{m\neq 0} Z^{-m-1|1}W^{m|0}(\alpha|\beta)+\frac{1}{\nu} \sum_{m\in \ZZ}Z^{-m-1|0}W^{m|1}(\alpha|\beta)\\
    &=-\frac{1}{\nu}(\alpha|\beta)\delta(Z,W)+\frac{1}{\nu}(\alpha|\beta)Z^{-1|1}.
      \end{aligned}
  \end{equation}
  Also, from the definition of $s_{-\frac{\alpha}{\nu}}$ in \eqref{eq: commutator with shift operator}, one can compute that
  \begin{equation}
    \big[\beta(Z), s_{-\frac{\alpha}{\nu}}\big]=-\frac{1}{\nu}(\alpha|\beta)Z^{-1|1}s_{-\frac{\alpha}{\nu}}.
  \end{equation}
  Hence, if we write $e^{-\frac{1}{\nu}\int \alpha(W)}=s_{-\frac{\alpha}{\nu}}\, \widetilde{\exp}$ in the expression \eqref{eq: exponential definition}, the commutator of $e^{-\frac{1}{\nu}\int \alpha(W)}$ and $\beta(Z)$ is equal to
  \begin{equation*}
    \begin{aligned}
    &\Big[\beta(Z), e^{-\frac{1}{\nu}\int \alpha(W)}\Big]\\
    &=\Big[\beta(Z), s_{-\frac{\alpha}{\nu}}\Big]\widetilde{\exp}+(-1)^{p(\alpha)+1}s_{-\frac{\alpha}{\nu}}\Big[\beta(Z),\ \frac{1}{\nu}\sum_{j\neq 0}\frac{W^{-n|0}}{n} \alpha_{(n|0)}-\frac{1}{\nu}\sum_{j\in \ZZ}W^{-j-1|1}\alpha_{(j|1)}\Big]\widetilde{\exp}\\
    &=-\frac{1}{\nu}(\alpha|\beta)Z^{-1|1}e^{-\frac{1}{\nu}\int \alpha (W)}+\Big(-\frac{1}{\nu}(\alpha|\beta)\delta(Z,W)+\frac{1}{\nu} (\alpha|\beta)Z^{-1|1}\Big)e^{-\frac{1}{\nu}\int \alpha (W)}\\
    &=-\frac{1}{\nu}\, (\alpha|\beta)\, \delta(Z,W)\,e^{-\frac{1}{\nu}\int \alpha(W)}.
    \end{aligned}
  \end{equation*}
\end{proof}

\begin{lemma}
  For a simple root $\alpha\in\Pi$, recall the formal series $\expo$ in \eqref{eq: exponential definition}. This series satisfies the following equality
  \begin{equation*}
    (\partial_{\theta}+\theta \partial_z)\expo=-\frac{1}{\nu}:\!\alpha(Z)\expo\!:.
  \end{equation*}
\end{lemma}
\begin{proof}
  We prove that
  \begin{align}
    \theta \partial_z \big(\expo\big)=\,-\frac{1}{\nu}:\!\Big(\sum_{n\in \ZZ}Z^{-n-1|1}\alpha_{(n|0)}\Big)\expo\!:, \label{eq: lemma proof goal1}\\
    \partial_{\theta}\big(\expo\big)=\,-\frac{1}{\nu}:\!\Big(\sum_{n\in \ZZ}Z^{-n-1|0}\alpha_{(n|1)}\Big)\expo\!:, \label{eq: lemma proof goal2}
  \end{align}
   which sum up to the desired equality. Recall the expression $\widetilde{\exp}$ used in the proof of Lemma \ref{lem: appendix screening condition}. We denote the series
  \begin{equation} \label{eq: lemma proof exponential components}
    \begin{gathered}
    A(Z)_+=-\frac{1}{\nu}\sum_{n<0}Z^{-n-1|1}\alpha_{(n|1)},\quad A(Z)_-=-\frac{1}{\nu}\sum_{n\geq 0}Z^{-n-1|1}\alpha_{(n|1)},\\
    B(Z)_+=\frac{1}{\nu}\sum_{n<0}\frac{Z^{-n|0}}{n} \alpha_{(n|0)}, \quad B(Z)_-=\frac{1}{\nu}\sum_{n>0} \frac{Z^{-n|0}}{n} \alpha_{(n|0)}
    \end{gathered}
  \end{equation}
  so that $\widetilde{\exp}=\exp(A(Z)_+)\exp(B(Z)_+)\exp(B(Z)_-)\exp(A(Z)_-)Z^{-\frac{1}{\nu}\alpha_{(0|0)}|0}$. Since $\partial_z$ acts as an even derivation, one can compute that
  \begin{equation*}
    \begin{aligned}
    &\theta \partial_z\big(\expo\big)\\
    &=\theta\,s_{-\frac{\alpha}{\nu}}\,(\partial_z B(Z)_+)\cdot\widetilde{\exp}+ \theta\,s_{-\frac{\alpha}{\nu}}\,\widetilde{\exp}\cdot (\partial_z B(Z)_-)+\theta\,s_{-\frac{\alpha}{\nu}}\,\widetilde{\exp} \cdot\Big(-\frac{1}{\nu}Z^{-1|0}\alpha_{(0|0)}\Big)\\
    &=\Big(-\frac{1}{\nu}\sum_{n<0}Z^{-n-1|1}\alpha_{(n|0)}\Big)\expo+(-1)^{p(\alpha)+1}\expo\Big(-\frac{1}{\nu}\sum_{n\geq 0}Z^{-n-1|1}\alpha_{(n|0)}\Big)\\
    &=-\frac{1}{\nu}:\!\Big(\sum_{n\in \ZZ}Z^{-n-1|1}\alpha_{(n|0)}\Big)\expo\!:,
    \end{aligned}
  \end{equation*}
  which proves \eqref{eq: lemma proof goal1}. Now, note that the differential $\partial_{\theta}$ acts as an odd differential on the space of formal power series in $Z$ and $Z^{-1}$. Also, for for any series $X(Z)$ of even parity, 
  \begin{equation} \label{eq: lemma proof exp differentiation}
    \partial_{\theta}\big(\exp(X(Z))\big)=(\partial_{\theta}X(Z))\big(\!\exp (X(Z))\big)=\big(\!\exp(X(Z))\big)(\partial_{\theta}X(Z)),
  \end{equation}
  when $\partial_{\theta}X(Z)$ and $X(Z)$ supercommute with each other. Apply this relation to obtain
\begin{equation*}
  \begin{aligned}
  &\partial_{\theta}\big(\expo\big)=(-1)^{p(\alpha)+1}s_{-\frac{\alpha}{\nu}}\,\partial_{\theta}\big(A(Z)_+\big)\cdot \widetilde{\exp}+(-1)^{p(\alpha)+1}s_{-\frac{\alpha}{\nu}}\,\widetilde{\exp}\cdot \partial_{\theta}(A(Z)_-)\\
  &=\Big(-\frac{1}{\nu}\sum_{n<0}Z^{-n-1|0}\alpha_{(n|1)}\Big)\expo +(-1)^{p(\alpha)+1}\expo \Big(-\frac{1}{\nu}\sum_{n\geq 0}Z^{-n-1|0}\alpha_{(n|1)}\Big)\\
  &=-\frac{1}{\nu}:\!\Big(\sum_{n\in \ZZ}Z^{-n-1|0}\alpha_{(n|1)}\Big)\expo\!:,
  \end{aligned}
\end{equation*}
which shows the equality \eqref{eq: lemma proof goal2}.
\end{proof}

\section*{Acknowledgement}
I would like to express my gratitude to my supervisor, Uhi Rinn Suh, for her invaluable support and guidance in enhancing this paper. Additionally, I am deeply thankful to Naoki Genra for the enriching discussions we had during his visit to Seoul National University in 2023. My appreciation also goes to Gahng Sahn Lee and Sin-Myung Lee for valuable discussions about Lie superalgebras and their representations.

\end{document}